\documentclass[12pt]{iopartmodfinal}
\pdfoutput =1
\usepackage{cite}
\usepackage{epsfig}
\usepackage{amssymb}
\usepackage{iopams}
\usepackage{braket}
\usepackage{enumerate}
\usepackage{cancel}
\usepackage{simplewick}
\usepackage[usenames,dvipsnames]{color}
\usepackage{bbm}
\usepackage[colorlinks,bookmarks=false,citecolor=NavyBlue,linkcolor=OliveGreen,
urlcolor=blue]{hyperref}
\newcommand{\be}{\begin{equation}}
\newcommand{\ee}{\end{equation}}
\newcommand{\bea}{\begin{eqnarray}}
\newcommand{\eea}{\end{eqnarray}}
\def\nn{\nonumber\\}

\eqnobysec

\newtheorem{theorem}{Theorem}[section]
\newtheorem{lemma}[theorem]{Lemma}
\newenvironment{proof}[1][Proof]{\begin{trivlist}
\item[\hskip \labelsep {\bfseries #1}]}{\end{trivlist}}
\newenvironment{definition}[1][Definition]{\begin{trivlist}
\item[\hskip \labelsep {\bfseries #1}]}{\end{trivlist}}
\newenvironment{remark}[1][Remark]{\begin{trivlist}
\item[\hskip \labelsep {\bfseries #1}]}{\end{trivlist}}
\newtheorem{corollary}[theorem]{Corollary}
\newtheorem{property}[theorem]{Property}

\usepackage{bbm}
\newcommand{\1}{\mathbbm{1}}

\begin{document}
\title{Pre-relaxation in weakly interacting models}  

\author{Bruno Bertini$^1$ and Maurizio Fagotti$^{1,2}$}
\address{$^1$The Rudolf Peierls Centre for Theoretical Physics, Oxford University, Oxford, \mbox{OX1 3NP}, United Kingdom}
\address{$^2$D\'epartement de Physique, \'Ecole normale sup\'erieure, CNRS, 24 rue Lhomond,
75005 Paris, France}
 
\begin{abstract}
We consider time evolution in models close to integrable points with hidden symmetries that 
generate infinitely many local conservation laws that do not commute with one another.
The system is expected to (locally) relax to a thermal ensemble if integrability is broken, 
or to a so-called generalised Gibbs ensemble if unbroken. In some circumstances expectation 
values exhibit quasi-stationary behaviour long before their typical relaxation time. For 
integrability-breaking perturbations, these are also called pre-thermalisation plateaux, 
and emerge e.g. in the strong coupling limit of the Bose-Hubbard model. As a result of the hidden
symmetries, quasi-stationarity appears also in integrable models, for example in the Ising limit
of the XXZ model. We investigate a weak coupling limit, identify a time window in
which the effects of the perturbations become significant and solve the time evolution
through a mean-field mapping. As an explicit example we study the XYZ spin-$\frac{1}{2}$ chain
with additional perturbations that break integrability. One of the most intriguing results of the analysis is the appearance of persistent oscillatory 
behaviour. To unravel its origin, we study in 
detail a toy model: the transverse-field Ising chain with an additional nonlocal interaction proportional to the square of the transverse spin per unit length [Phys. Rev. Lett. \href{http://dx.doi.org/10.1103/PhysRevLett.111.197203}{\bf 111}, 197203 (2013)]. Despite being nonlocal, this belongs to a class of models that emerge as intermediate steps of the mean-field mapping and shares many dynamical properties with the weakly interacting models under consideration.
\end{abstract} 
\newpage
\tableofcontents
\maketitle

\section{Introduction} %

The emergence of stationary behaviour in \emph{closed} quantum many-body systems is perhaps one of the most striking features of non-equilibrium dynamics.
Rather generally subsystems behave as they were in a ``bath'' and correlation functions relax to stationary values that can be described in a statistical fashion.

The maturation and refinement of experimental techniques have led to the design of experiments ever more effective in the extraction of information on the long time dynamics \cite{gm-02, kww-06, Getal, Tetal-12, chetal-12, schetal-12, Metal-13, Fateal-13, FSetal-13, Metal-14}. A sensitive communication between theory and experiment has then made it possible to identify the most interesting aspects of the non-equilibrium time evolution in quantum many-body systems\cite{QQrev}. 

In particular, it was realised that integrable models behave differently from generic ones. 
The stationary properties are indeed affected by the local conservation laws, which in integrable models are infinite in number. This led to the concept of generalised Gibbs ensemble (GGE) \cite{{rigol}, {rigol2}, {C-06}, {IC-09}, {CE-13}, {C-08}, {BS-08}, {FE:RDM}, {C-10}, {CEF}, {EK}, {rig2}, {CK}, {CIC}, {Gu}, {CSC-13}, {HSWR-13}, {andrei}, {DBZ-12}, {NI-13}, {KCC-14}, {RS-14}, {SC:cluster}, {FM-10}, {P-13}, {FE-12}, {M-13}, {KSCCI-13}, {FCEC-14}, {NWBC-14}, {QAXXZ-14}, {PMWKZT-14}, {A:bound}, {CC:2},FCG:FDT, EEF:dyn,P14:qb,EMP15}, which is often defined as the mixed state with maximal entropy under the constraints of the local conservation laws. 
Non-integrable models with no other local conservation laws except for the Hamiltonian itself are supposed to ``thermalise'' at some effective temperature \cite{deu-91, sred-94, rignat-08, bir-10, ban-11, bran-12, sir-14}, whereas integrable models retain infinite information about the initial state. 

In the same way as relaxation and thermalisation were associated with the late time behaviour in integrable and non-integrable models, pre-thermalisation \cite{Getal, QHubbard,  metaOpt, Kollar, IsingNI, E:preT} has been recognised as a typical feature of generic models close to integrable points. Essentially, at intermediate times the non-purely-elastic processes typical of non-integrable models are almost absent and the system behaves as if it were integrable. Despite the strenuous efforts to understand the process of thermalisation in the presence of pre-thermalisation plateaux, the picture is still far from being clear and, so far, only the earliest plateau has found satisfactory descriptions \cite{Kollar, E:preT}.

This state of affairs boosted the research into conserved and quasi-conserved operators in non-integrable models \cite{KAM,F:charge}, on the one hand, and put physicists' ingenuity to the test to propose sufficiently simple models to study pre-thermalisation \cite{IsingNI, E:preT}, on the other.

One interesting proposal \cite{IsingNI} was to break the integrability of the transverse field Ising chain (TFIC) by adding a highly nonlocal interaction proportional to the global magnetisation squared per unit length. 
Even though the model possesses infinite local conservation laws (odd under reflection symmetry) \cite{F:pair, FE:RDM}, it was argued to behave like a non-integrable model in the sector of reflection symmetric states. In particular, \cite{IsingNI} developed a perturbation theory that allows one to follow the time evolution of the ground state of a TFIC for sufficiently long times to see a pre-thermalisation plateau. 

In fact, the situation seems to be more complicated. 
Some techniques that will be developed in this paper allow us to analytically study the dynamics  of a class of nonlocal Hamiltonians that includes the model introduced in \cite{IsingNI}. 
The method we use is exact in the thermodynamic limit (in which the large-system limit is taken first): we prove some conjectures of \cite{IsingNI}, but we find that the time evolution does not result in thermalisation. Therefore, in the thermodynamic limit, models like the one of \cite{IsingNI}  can not be naively related to the physics underlying pre-thermalisation. Nevertheless, we show that similar types of nonlocal Hamiltonians can emerge  at intermediate times as effective descriptions of the dynamics generated by local Hamiltonians.  
Thus, instead of spoiling the interest in such models, our findings give new motivations for their study.

Our main goal is to investigate the time evolution of local observables under Hamiltonians with local interactions in the particular time windows where such effective descriptions can be used. 
There are indeed interesting cases where the expectation values start moving significantly from a plateau that could have been approximately described by the stationary state of the unperturbed model. We show that the crossover is driven by the presence of infinitely many local conservation laws that do not commute with one another, which will be referred to as \emph{non-abelian integrability}. 
We therefore identify two necessary requirements for a nontrivial time evolution.
First, the unperturbed model must have a non-abelian set of local conservation laws. Second, the perturbation must break non-abelian integrability. 

Since the crossover appears also in the presence of perturbations preserving integrability, we call it ``pre-relaxation'' and the limit ``pre-relaxation limit''.

The pre-relaxation limit has been already considered in \cite{F:super}, where a typical crossover behaviour between two plateaux
has been identified in noninteracting models like the XY quantum spin chain. 
Ref. \cite{F:super} also obtained similar results for a particular quench in an interacting model (XXZ spin-$\frac{1}{2}$ chain). However, despite an effective description was  proposed  that is supposed to capture the relaxation process for quite general interactions, the idea was tested only on simple cases in which the dynamics is essentially noninteracting.

In this paper we start filling this gap by investigating the pre-relaxation behaviour triggered off by more general (interacting) perturbations. This is a highly non trivial generalisation. The first universal picture of the time evolution of correlation functions after a quantum quench was delineated in conformal field theories~\cite{CC, CC:2}, but most of the analytic results have been in fact obtained in models that can be mapped to free fermions or bosons~\cite{CEF,C-06,CSC-13,KCC-14,RS-14,F:super,SC:cluster,FM-10, EEF:dyn,BKC:exc,MCKC-14,RS:long}. For the serious complications introduced by the interactions, there are far less examples~\cite{B:sG, NC-14, D-14} in which the time evolution of some nontrivial observable has been worked out in interacting models. In the pre-relaxation limit some obstacles can be overcome. In particular we show that, at the leading order in the perturbation strength, the dynamics generated by (local) weakly interacting Hamiltonians are \emph{equivalent} to those generated by time-dependent (quasi-)local (mean-field) Hamiltonians, which can be solved in a self-consistent way. Differently from the common situations, the mean-field mapping presented here is not an uncontrolled approximation, but arises naturally in the timescale investigated under few reasonable assumptions. The possibility to write a compact system of nonlinear differential equations for the time evolution of local observables can therefore be used to investigate the essential aspects of the pre-relaxation limit even analytically.

\subsection{Organisation of the paper}%

The paper is organised as follows. 

\begin{itemize}

\item[-] \Sref{s:summary} is a summary of our main results. 

\item[-] In \Sref{s:pre-rel} we propose an effective description of the dynamics within a time window in which  perturbations to a non-abelian integrable model become relevant.

\item[-] \Sref{s:int} is devoted to identify the class of effective Hamiltonians that emerge in the pre-relaxation limit. We show that they can be written as polynomials of the local conservation laws of the unperturbed model (with the correct scaling factors). 

\item[-] In \Sref{s:mf} we introduce mean-field Hamiltonians which, in the thermodynamic limit, generate \emph{exactly} the same dynamics of the effective Hamiltonians.

\item[-] The formalism is explicitly applied to the XYZ spin-$\frac{1}{2}$ chain in \Sref{s:pre-relI}, where pre-relaxation is investigated also in the presence of interactions that break integrability. 

\item[-] \Sref{s:Ising} provides a detailed analysis of the model considered in \cite{IsingNI}. 
We examine its relaxation properties and rule out thermalisation (in the thermodynamic limit). Besides its intrinsic interest, the model will be useful to understand the emergence of oscillatory behaviour observed in the pre-relaxation limit of the XYZ model.

\item[-] \Sref{s:conclusions} contains our conclusions.

\item[-] Several appendices complement the main text with the proofs of the theorems and additional details. 

\end{itemize}

\section{Summary of the results}\label{s:summary}%

We consider the time evolution of some initial state $\ket{\Psi_0}$ with cluster decomposition properties\footnote{
We say that the state $\ket{\Psi_0}$ has cluster decomposition properties if 
$$
\fl\qquad\lim_{\min\limits_{i\neq j}|x_i-x_j|\rightarrow\infty}\Bigl(\braket{\mathcal O_1(x_1)\mathcal O_2(x_2)\cdots \mathcal O_n(x_n)}-\braket{\mathcal O_1(x_1)}\braket{\mathcal O_2(x_2)}\cdots \braket{\mathcal O_n(x_n)}\Bigr)=0
$$
where the operators $\mathcal O_{i}(x_i)$ are local (act trivially far away from the site $x_i$) and the expectation values are taken with respect to $\ket{\Psi_0}$.}
under translation invariant Hamiltonians of the form
\be\label{eq:H}
H=H_0+g V\, ,
\ee
where $V$ is a global perturbation and $g$ is a small coupling constant. 

We focus on perturbations $V$ that break some symmetries of $H_0$ in such a way that the limit of infinite time of the expectation value of a local observable $\mathcal O$ does not commute with the limit of infinitesimal $g$
\be\label{eq:pre-relax}
\fl \qquad \qquad \lim_{g\rightarrow0}\lim_{t\rightarrow\infty}\braket{\Psi_0|e^{i(H_0+g V)t}\mathcal O e^{-i(H_0+g V)t}|\Psi_0}\neq \lim_{t\rightarrow\infty}\braket{\Psi_0|e^{i H_0t}\mathcal O e^{-i H_0 t}|\Psi_0}\, .
\ee
This is the typical situation in which local degrees of freedom experience a pre-thermalisation/pre-relaxation behaviour. Indeed, at not too large times the effect of the perturbation is negligible and the expectation value has time to settle at the stationary value of the unperturbed Hamiltonian. On the other hand, at later times the perturbation can not be ignored anymore and the expectation value varies with a typical timescale that depends on the perturbation strength.
Importantly, the amplitude of the variation is $\mathcal O(g^0)$, therefore the pre-relaxation behaviour can be understood in the limit of infinitesimal $g$.  

In principle, there could be many pre-relaxation plateaux, depending on how the time $t$ scales with the small parameter $g$. 
Here we focus on the limit $g\ll1$ and large time in such a way that $T=gt\sim O(g^0)$. \Fref{f:scheme} summarises the various steps of the formalism that will be developed in the next three sections to investigate the pre-relaxation limit. The analysis of explicit examples will be  instead carried out in the last two sections.

\begin{figure}[t]
\begin{center}
\includegraphics[width=0.95\textwidth]{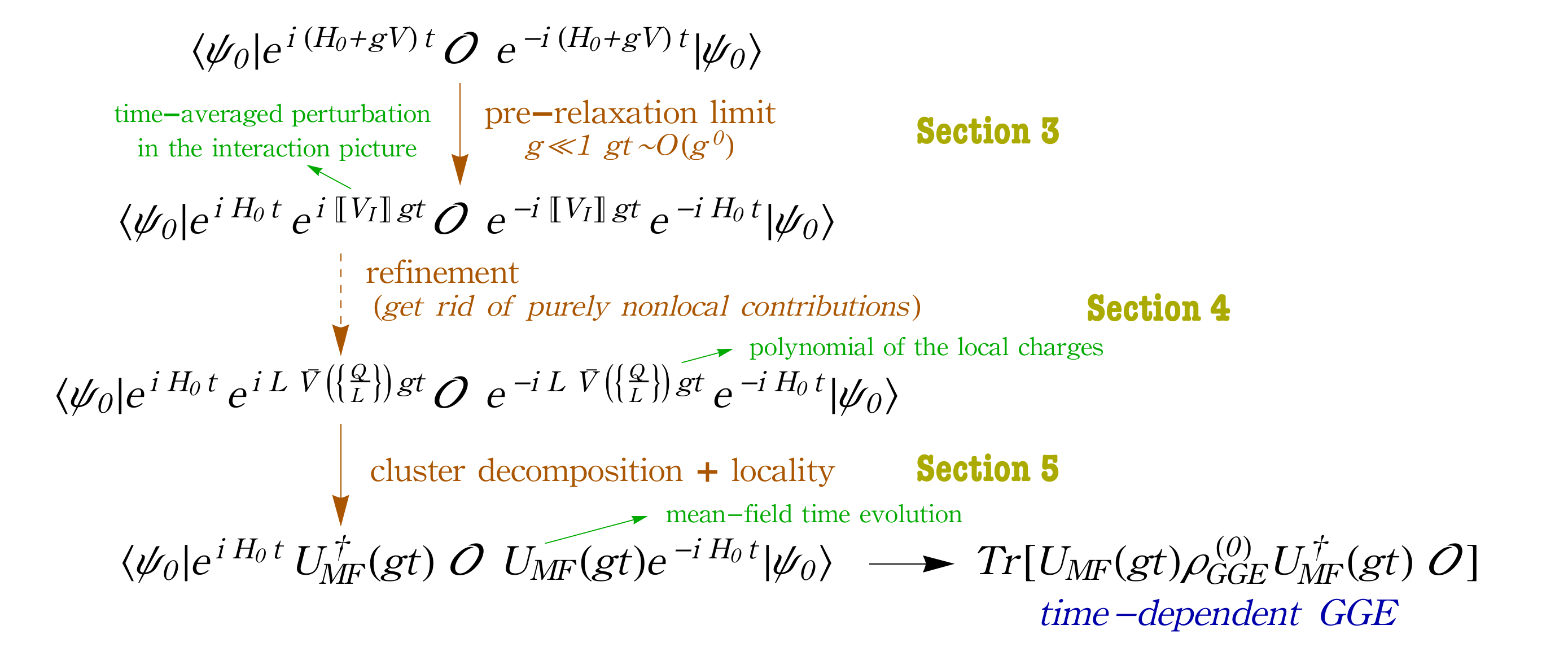}
\end{center}\caption{Scheme of the formalism employed to investigate a pre-relaxation limit in models described by Hamiltonian $H_0$ perturbed by some interacting term $g V$. In dark yellow, the section where the particular step is developed. 
}\label{f:scheme}
\end{figure}

We now present a comprehensive summary of the main results.  

\begin{itemize}

\item {}[\Sref{s:pre-rel} and \Sref{s:int}] Under some assumptions, the pre-relaxation limit of weakly interacting local Hamiltonians can be described by effective nonlocal Hamiltonians of the form 
\be\label{eq:opform0}
\bar H=\frac{1}{L^{n_1-1}}Q^{(1)}_1 \dots Q^{(1)}_{n_1}+\dots +\frac{1}{L^{n_m-1}}Q^{(m)}_1 \dots Q^{(m)}_{n_m}\, ,
\ee
where $Q_j^{(\ell)}$ are \mbox{(quasi-)local} conservation laws\footnote{It is customary to call `local' a (translation invariant) conservation law with local density.} of the (integrable) unperturbed model. In order to be a nontrivial limit,  $Q_j^{(\ell)}$ should not commute with one another.   

\item {}[Section \ref{s:mf} and Appendix \ref{a:MF}] In the thermodynamic limit, the time evolution of a state with cluster decomposition properties under Hamiltonians of the form \eref{eq:opform0} is completely equivalent to the time evolution under the time-dependent mean-field Hamiltonian
\be\label{eq:VMF}
\bar H_{\rm MF}^{\Psi_0}(gt)=\sum_{i=1}^m\sum_{j=1}^{n_i} c^{(i)}_{j}(gt;\Psi_0)Q^{(i)}_{j}\, ,
\ee
where the coefficients $c^{(i)}_{j}(gt;\Psi_0)$ are obtained self-consistently as follows:
\be
c^{(i)}_{j}(gt;\Psi_0)=\prod_{n\neq j}\mathrm{Tr}\Bigl[\rho_{\rm GGE}\bar U_{\rm MF}^{\dag}(gt;\Psi_0) \frac{Q_{n}^{(i)}}{L}\bar U_{\rm MF}(gt;\Psi_0)\Bigr]\, .
\ee
Here $\rho_{\rm GGE}$ is the generalised Gibbs ensemble that emerges in the time evolution under the unperturbed Hamiltonian and $\bar U_{\rm MF}(gt;\Psi_0)$ is the time evolution operator of $\bar H_{\rm MF}^{\Psi_0}(gt)$
\be
\bar U_{\rm MF}(gt;\Psi_0)={\rm T} \exp\Bigl(-i\int_0^{gt}\mathrm d\tau \bar H_{\rm MF}^{\Psi_0}(\tau)\Bigr)\, ;
\ee
$\rm T$ is the time-ordering operator, which formally orders operators depending on $\tau$  in such a way that those on the left are associated with larger (or equivalent) values of $\tau$.

Despite mean-field approximations being extremely common also in the field of non-equilibrium physics \cite{GC-11,SC-10},  in the thermodynamic limit the mean-field mapping presented here is actually exact under a mild assumption.

\item {}[Section \ref{s:pre-relI}] 
We consider Slater determinant initial states evolving under
 \be
\fl H=J \sum_\ell\Bigl(\frac{1+\gamma}{4}\sigma_{\ell}^x\sigma_{\ell+1}^x+\frac{1-\gamma}{4}\sigma_{\ell}^y\sigma_{\ell+1}^y+\frac{g}{4}\sigma_{\ell}^z\sigma_{\ell+1}^z+\frac{g U }{4}\sigma_{\ell}^z\sigma_{\ell+2}^z\Bigr)+\frac{gh}{2}\sum_\ell\sigma_\ell^z
\ee
in the the pre-relaxation limit $g t\sim O(g^0)$ with $g\ll 1$. 
We identify three different behaviours:

\begin{itemize}

\item[-] Emergence of a \emph{second plateau} that  can be described by a GGE constructed with the local conservation laws of the unperturbed model. The information about the initial state is \emph{not} encoded in a finite number of parameters because the effective Hamiltonian commutes with the unperturbed Hamiltonian and with infinitely many other conservation laws in involution. 
If broken, one-site shift invariance is generally not restored. 

\item[-] \emph{Oscillatory behaviour}: the expectation values of local observables keep oscillating with frequency proportional to $g$. Nevertheless, they can be described by a time-dependent GGE (that is not one-site shift invariant). 

\item[-] In the pre-relaxation limit the expectation values are \emph{independent of time}. This happens whenever the initial state is one-site shift invariant, or, more generally, when it is an excited state of $\bar H_{\rm MF}^{\Psi_0}(0)$ \eref{eq:VMF} (notice that the very definition of $\bar H_{\rm MF}^{\Psi_0}$ depends on the state, so there are implicit self-consistent conditions to be satisfied). 
In the latter case, one-site shift invariance is generally broken. 

\end{itemize}

We stress that, rather unexpectedly, a genuine one-site shift invariant interaction is not always sufficient to induce the restoration of one-site shift invariance. This is an indication that non-abelian integrability can survive interacting perturbations (at least at the lowest orders of perturbation theory).  

We will report explicit examples of the aforementioned behaviours.

\item {}[Section \ref{s:Ising}] The equivalence with the mean-field description does not rely on the fact that the operators $Q_j^{(\ell)}$ of \eref{eq:opform0} commute with the unperturbed Hamiltonian; relaxing such hypothesis allows us to construct more general models.  Ref.~\cite{IsingNI} has recently proposed the model with Hamiltonian
\be\label{eq:TFICNIS}
H(\texttt{g},\lambda)=-\sum_\ell^L(\sigma_\ell^x\sigma_{\ell+1}^x+\texttt{g}\sigma_\ell^z)+\frac{\lambda}{L}\left(\sum_\ell^L\sigma_\ell^z-\overline{\sum_\ell^L\sigma_\ell^z}\right)^2
\ee
as a convenient framework for studying pre-thermalisation/thermalisation issues. 
Here $\overline{\phantom{(}\!\!\!\cdots\phantom{)}\!\!\!}$ denotes the time average with respect to $H(\texttt{g},0)$. 
In the thermodynamic limit $L\rightarrow\infty$, 
we show that the time evolution does not result in thermalisation.

The apparent conflict between our results and \cite{IsingNI} can be traced back to the different order of limits. We indeed investigate the expectation values of local observables in the limit\footnote{As a matter of fact, using the perturbative results of \cite{IsingNI}, for which the lifetime of the Ising quasiparticles that diagonalise $H_f$ \eref{eq:Hf} scales as $L^{-1}$, it is reasonable to expect that our results hold true even in the limit $J t\sim L^\alpha$ with $\alpha<1$ and $L\rightarrow\infty$.}
\be\label{eq:limit}
(\lim_{t\rightarrow\infty})\lim_{L\rightarrow\infty}\braket{\Psi_0|e^{iH(\tilde g,\lambda) t}\mathcal Oe^{-iH(\tilde g,\lambda) t}|\Psi_0}\, ,
\ee
while, for what concerns the stationary properties,  \cite{IsingNI} considered the time average in finite systems
\be
\lim_{T\rightarrow\infty}\frac{1}{T}\int_0^T\mathrm d t \braket{\Psi_0|e^{iH(\tilde g,\lambda) t}\mathcal Oe^{-iH(\tilde g,\lambda) t}|\Psi_0}\, .
\ee
We find that the stationary behaviour of local observables in the limit \eref{eq:limit} is not characterised by a finite number of parameters.  More generally, we argue that thermal-like behaviour can be excluded when every linear combination of the operators $Q_j^{(\ell)}$ is the Hamiltonian of an integrable model.

\begin{figure}[tbp]
\begin{center}
\includegraphics[width=0.8\textwidth]{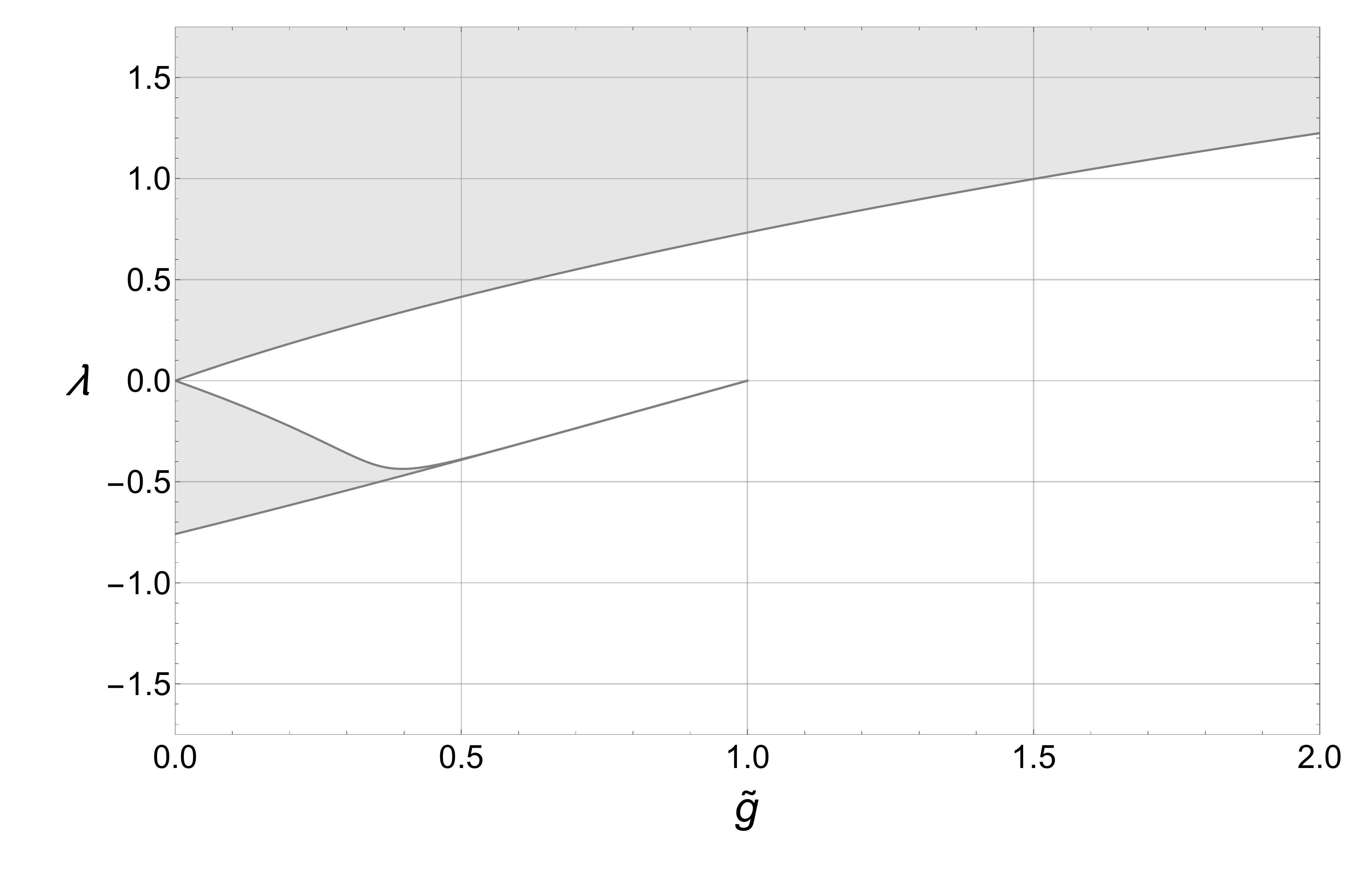}
\end{center}\caption{Quench dephasing diagram of the model \eref{eq:TFICNI1} in the limit of small quench with energy close to the ground state one. In the dark region there is persistent oscillatory behaviour at any time after the quench.  
}\label{f:diagram}
\end{figure}

\item {}[Section \ref{s:Ising}] 
For Hamiltonians like \eref{eq:TFICNI1} one can generally identify `critical regions' for which the late time dynamics is unstable under a small change of the parameters: an infinitesimal variation can lead both to relaxation and to persistent oscillatory behaviour (see \emph{e.g.} \cite{GC-11} for similar discussions in quantum field theories).
In addition, the variance of the expectation value of some local operator in an arbitrarily large time window 
\be\label{eq:variance}
\Delta\mathcal O=\lim_{T\rightarrow\infty }\Bigl(\frac{1}{T}\int_T^{2T}\mathrm d t \braket{\mathcal O(t)}^2-\Bigl(\frac{1}{T}\int_T^{2T}\mathrm d t \braket{\mathcal O(t)}\Bigr)^2\Bigr)^{1/2}
\ee
behaves like an `order parameter' for the transition, indeed it is not analytic at the boundaries of the relaxation region and vanishes inside.

The crossover between relaxation and persistent oscillations can be illustrated by means of diagrams that depict the relaxation properties as a function of the Hamiltonian parameters; in particular here we focus on the limit of small quench (see also \cite{CEF} for some clarifications about the meaning of `small quench'). To highlight the close connection with the dephasing mechanisms \cite{BS-08} that allow local relaxation we call the diagrams `quench dephasing diagrams'\footnote{
The terminology has not been yet standardised and in the scientific literature similar diagrams were sometimes called `dynamical phase diagrams' (see \emph{e.g.} \cite{SB:conn}) but also `quench phase diagrams' (see \emph{e.g.} \cite{QPD})}. 
We point out that restricting ourselves to small quenches makes it easier to interpret the results in terms of low lying excitations.

We investigated the quench dephasing diagram of a simplified version of \eref{eq:TFICNIS}, \emph{i.e.}
\be\label{eq:TFICNI1}
H(\tilde g,\lambda)=-\sum_\ell^L(\sigma_\ell^x\sigma_{\ell+1}^x+\tilde g\sigma_\ell^z)+\frac{\lambda}{L}\Bigl(\sum_\ell^L\sigma_\ell^z\Bigr)^2\, ,
\ee
in the limit in which the initial state is almost the ground state of the Hamiltonian. In Figure~\ref{f:diagram} we can identify a `critical' piecewise smooth curve at the boundaries of the relaxation region (the bright area).
Generally, as we move towards the critical lines from the inside of the regions with persistent oscillations, the variance \eref{eq:variance} of the transverse magnetisation approaches zero linearly with the distance in the $(\tilde g,\lambda)$ parameter space (\emph{cf.} \fref{f:relax}). 
Close to the Ising critical point ($\tilde g=1$ and $\lambda=0$) the region of persistent oscillatory behaviour degenerates into a line ending at the critical point. 

Finally we show that persistent oscillatory behaviour can be related to the emergence of localised excitations. 

\end{itemize}

\section{Pre-relaxation limit}\label{s:pre-rel}%

The time evolution operator for the Hamiltonian \eref{eq:H} can be formally written as follows
\be\label{eq:teo}
e^{-i H t}=U_I(T)e^{-i H_0 t}\, ,
\ee
where $T=gt$,
\be
U_I(T)=\mathrm {T}^\dag\exp\Bigl(-i \int_0^T\mathrm d \tau e^{-i \frac{H_0}{g}\tau} V e^{i \frac{H_0}{g}\tau} \Bigr)\, ,
\ee
and $\mathrm {T}^\dag$ is the anti-time-ordering operator. 
$U_I(T)$ can be interpreted as the evolution backwards in time under the time-dependent effective Hamiltonian
\be
V(\tau)=e^{-i \frac{H_0}{g}\tau} V e^{i \frac{H_0}{g}\tau}\, .
\ee
In the pre-relaxation limit, $\tau$ is finite while $g$ is infinitesimal; it is therefore convenient to isolate the stationary (\emph{i.e.} diagonal) contributions from $V(\tau)$
\be
V(\tau)=\bar V+\delta V(\tau)\, ,
\ee
where $\bar V$ can be formally written as follows
\be\label{eq:Vbar}
\bar V=\lim_{t\rightarrow\infty}\frac{1}{t}\int_0^t\mathrm d\tau V(\tau)\, .
\ee 
Le us rewrite $\delta V(\tau)$ in terms of its Fourier transform $\delta \tilde V(\varepsilon)$
\be
\delta V(\tau)=\int_{-\infty}^\infty\mathrm d \varepsilon\ e^{i\frac{\varepsilon\tau}{g}}\delta \tilde V(\varepsilon)\qquad \delta \tilde V^\dag(\varepsilon)=\delta \tilde V(-\varepsilon)\, .
\ee
We notice that for noninteracting $H_0$ (which is the case we are going to consider), the locality of $V$ (in the fermionic picture) implies that $\delta \tilde V(\varepsilon)$ is zero for $|\varepsilon|>\varepsilon_{\rm max}\sim O(1)$. 
After some formal manipulations we find
\be\label{eq:UA}
\fl\qquad U_I(T)(\mathrm I -i g A(T))=\mathrm I -i g A(0)-i\int_0^T\mathrm d\tau  U_I(\tau) [\bar V-i g \delta V(\tau)A(\tau)]
\ee
with
\be
A(\tau)=i \int_{-\infty}^\infty\mathrm d \varepsilon\ e^{i\frac{\varepsilon \tau}{g}}\frac{\delta \tilde V(\varepsilon)}{\varepsilon}\, .
\ee
Eq. \eref{eq:UA} makes sense as long as $A(\tau)$ does. In particular, if $A(\tau)$ per unit length is a bounded operator, in the limit $g\rightarrow 0$ all the terms of \eref{eq:UA} that are multiplied by $g$ can be neglected, \emph{i.e.} $\delta V(\tau)$ is negligible.
More generally, it is sufficient that the matrix elements of $A(\tau)$ that give a relevant contribution in the expectation values of local observables are bounded. 
In \cite{F:super} the operator $\delta V(\tau)$ was worked out for a particular noninteracting perturbation and the previous assumption turned out to be satisfied.  
Besides the noninteracting case, the irrelevance of $\delta V(\tau)$ was also implicitly assumed in \cite{E:preT}. There, using the continuous unitary transformation (CUT)  formalism~\cite{W:CUT,QHubbard},  the authors worked out the dynamics after a quantum quench in a weakly interacting model and checked the results against numerical data. For the readers familiar with CUT, we indeed notice that the formal simplification of the terms proportional to $g$ in \eref{eq:UA} is equivalent to CUT at $O(g^0)$: At the lowest order of perturbation theory the CUT unitary transformation can be replaced by the identity; 
a residual dependence on $g$ remains  in the CUT Hamiltonian $H_{\rm CUT}$, because we are considering the limit of infinite time with $T=g t$ finite, so $H_{\rm CUT}$ must be computed at $O(g)$. The latter is simply the time average of the Hamiltonian over its unperturbed part, namely $H_{\rm CUT}=H_0+g\bar V+o(g)$. 
The excellent agreement of \cite{E:preT} with the numerical simulations suggests that the interaction can be replaced by its time average for rather general perturbations. Therefore, from now on we shall assume that the effect of $\delta V(\tau)$ is negligible in the pre-relaxation limit.

Going back to \eref{eq:teo}, if $\delta V(\tau)$ is negligible we find
\be
e^{-i H t}\rightarrow e^{-i T \bar V}e^{-i H_0 T/g}\, ,
\ee
where $[\bar V,H_0]=0$. 
We now consider the time evolution of the expectation value of some local operator $\mathcal O$
\be\label{eq:prsl}
\braket{\Psi_0|e^{i H t}\mathcal O e^{-i H t}|\Psi_0}\rightarrow \braket{\Psi_0|e^{i H_0 T/g}e^{i T \bar V}\mathcal Oe^{-i T \bar V} e^{-i H_0 T/g}|\Psi_0}\, .
\ee
It is well established that the stationary properties of \mbox{(quasi-)local} observables after quenches in translation invariant noninteracting models from states with cluster decomposition properties can be described by means of a generalised Gibbs ensemble (GGE) of the form
\be
\rho_{\rm GGE}=\frac{e^{-\sum_j\lambda_j Q_j}}{Z}\, ,
\ee
where $Q_j$ are local conservation laws and $\lambda_j$ are real parameters determined by the initial state \cite{FE:RDM, SC:cluster}. We also remind the reader that, in order to avoid an explicit dependence of the charges on the initial state, in some special cases the set of charges $\{Q_j\}$ could be non-abelian \cite{F:super}. 

Let us now assume that the perturbation is sufficiently ``nice'' that the late time dynamics of $e^{i T \bar V}\mathcal Oe^{-i T \bar V}$ under $H_0$ can be obtained by replacing the state with the corresponding GGE
\be\label{eq:toGGE}
\fl\qquad e^{-i H_0 T/g}\ket{\Psi_0}\bra{\Psi_0}e^{i H_0 T/g}\rightarrow \rho_{\rm GGE}=\lim_{|S|\rightarrow\infty}\lim_{t\rightarrow\infty}\tr_{\bar S}[e^{-i H_0 t}\ket{\Psi_0}\bra{\Psi_0}e^{i H_0 t}]\, .
\ee
This step can be easily justified for $T=gt\ll t$ if $\bar V$ is \mbox{(quasi-)local}, however in the next sections we'll show that, in the presence of interactions, $\bar V$ belongs to a larger class of operators, so it is convenient to postpone the explanation of \eref{eq:toGGE} after having clarified the properties of $\bar V$.  

From \eref{eq:toGGE} it follows 
\be\label{eq:toGGE2}
\fl\qquad\qquad\lim_{g\rightarrow 0 }\braket{\Psi_0|e^{i H_0 T/g}e^{i T \bar V}\mathcal Oe^{-i T \bar V} e^{-i H_0 T/g}|\Psi_0}=\tr[\rho_{\rm GGE}e^{i T \bar V}\mathcal Oe^{-i T \bar V} ]\, ,
\ee
which suggests that the pre-relaxation limit can be described by the time-dependent ensemble
\be\label{eq:tGGE}
\rho_{\rm tGGE}(t)=e^{-i \bar V g t}\rho_{\rm GGE}e^{i \bar V g t}\, ,
\ee
where we re-expressed the rescaled time $T=gt$ in terms of the time. 

It is important to note that both $\rho_{\rm GGE}$ and $\bar V$ commute with the Hamiltonian. 
Consequently, if the two operators can be written in terms of the same set of local conservation laws in involution, the time dependence disappears $\rho_{\rm tGGE}(t)=\rho_{\rm GGE}$. In the next section we will show that in many cases of interest $\bar V$ can be approximated by a polynomial of the local conservation laws.
Therefore, in order to see some nontrivial pre-relaxation behaviour, the unperturbed Hamiltonian $H_0$ must have a non-abelian set of local charges. 
We refer the reader to \cite{F:super} for an extensive discussion of noninteracting models with that property.  

\section{Effective Hamiltonians}\label{s:int}%

The simplest noninteracting model that possesses local conservation laws that are not mutually commuting is the XY model, whose Hamiltonian is given by
\be\label{eq:XY}
H_{\rm XY}=J \sum_\ell(\frac{1+\gamma}{4}\sigma_\ell^x\sigma_{\ell+1}^x+\frac{1-\gamma}{4}\sigma_\ell^y\sigma_{\ell+1}^y)\, ,
\ee
where $\sigma_\ell^\alpha$ act like Pauli matrices on the site $\ell$ and like the identity elsewhere. 
If the initial state $\ket{\Psi_0}$ breaks one-site shift invariance, the latter symmetry is generally not restored in the GGE that describes local observables at infinite time after the quench.

On the other hand, an infinitesimally small one-site shift invariant perturbation that breaks the non-abelian integrability of \eref{eq:XY} is expected to catalyse symmetry restoration, which may be captured by the pre-relaxation limit. 
Similar issues of symmetry restoration have been pointed out long ago, \emph{e.g.} in \cite{Glassy}.  

A perturbation that preserves the noninteracting character of the Hamiltonian was already considered in \cite{F:super}. 
Here we investigate perturbations that have a 4-fermion representation in terms of the noninteracting fermions that diagonalise \eref{eq:XY}, namely
\be\label{eq:V}
V\sim  \sum_\ell a_{\ell+n_1}^{\alpha_1}a_{\ell+n_2}^{\alpha_2}a_{\ell+n_3}^{\alpha_3} a_{\ell+n_4}^{\alpha_4}\, ,
\ee
where $a_\ell^\alpha$ are the Majorana fermions ($\{a_\ell^\alpha,a_n^\beta\}=2\delta_{\alpha \beta}\delta_{\ell n}$)
\be\label{eq:Maj}
a_\ell^\alpha=\Bigl(\prod_{j<\ell}\sigma_j^z\Bigr) \sigma_\ell^\alpha\qquad \alpha\in\{x,y\}\, .
\ee 

From the qualitative argument presented in \Sref{s:pre-rel}, the relevant Hamiltonian in the pre-relaxation limit is determined by the time average of the perturbation.
The calculation is not difficult but rather lengthy.
However, a close inspection of the various contributions reveals a hidden structure that helps simplifying the computation. We indeed find (see Appendix \ref{a:Wick})
\begin{property}\label{L:Wick}
The time average under $H_{\rm XY}$ of a one-site shift invariant four fermion operator  can be written as follows
\be\label{eq:anom}
\frac{1}{L}\overline{\sum_\ell a_{\ell+n_1}^{\alpha_1}a_{\ell+n_2}^{\alpha_2}a_{\ell+n_3}^{\alpha_3} a_{\ell+n_4}^{\alpha_4}}= F_{\{n\}}^{\{\alpha\}}+ A_{\{n\}}^{\{\alpha\}}\, ,
\ee
where $ F $ is a linear combination of factorised terms and $ A$ is an anomalous contribution originated by the nontrivial solutions of the energy constraint
\be
\varepsilon(k_1)+\varepsilon(k_2)=\varepsilon(k_3)+\varepsilon(k_1+k_2+k_3)\, .
\ee
The latter exists only in the thermodynamic limit and strongly depends on the details of the dispersion relation, whereas $F$ has  a structure that is almost model independent:
\be\label{eq:Wick}
F_{\{n\}}^{\{\alpha\}}=\sum_{s=0}^1\underbrace{a_{n_1}^{\alpha_1} a_{n_2}^{\alpha_2}}_{s}\underbrace{a_{n_3}^{\alpha_3} a_{n_4}^{\alpha_4}}_{s}-\underbrace{a_{n_1}^{\alpha_1} a_{n_3}^{\alpha_3}}_{s}\underbrace{a_{n_2}^{\alpha_2} a_{n_4}^{\alpha_4}}_{s}+\underbrace{a_{n_1}^{\alpha_1} a_{n_4}^{\alpha_4}}_{s}\underbrace{a_{n_2}^{\alpha_2} a_{n_3}^{\alpha_3}}_{s}\, ,
\ee
with
\be
\underbrace{a_{n_1}^{\alpha} a_{n_2}^{\beta}}_s=\overline{ \frac{1}{L}\sum_\ell (-1)^{s\ell}a_{\ell+n_1}^{\alpha} a_{\ell+n_2}^{\beta} }\, .
\ee
\end{property}

Essentially, index $s$ appears because the XY model (with zero magnetic field) has local conservation laws with momentum $\pi$, while one-site shift invariance constrains the total momentum to be multiple of $2\pi$. We also notice that the factorised part of the time average can be easily generalised to an arbitrary number of fermions, keeping the same structure of the Wick decomposition.

To the best of our knowledge, the (quasi-)local conservation laws of the XY model \eref{eq:XY} are noninteracting (for $|\gamma|\neq 1$) and $A_{\{n\}}^{\{\alpha\}}$  of \eref{eq:Wick} seems to be a nonlocal conservation law that can \emph{not} be (not even approximately) written as a function of the local charges.  We then expect $ A_{\{n\}}^{\{\alpha\}}$ to become important (for local observables) only at times proportional to the chain length, which are far beyond the pre-relaxation limit.
This persuaded us to conjecture that the anomalous terms are not relevant to our problem, which is equivalent to assume
\be\label{eq:condA}
\braket{\Psi_0|e^{i H t}[\mathcal O,A] e^{-i H t}|\Psi_0}=0
\ee
for any local observable $\mathcal O$. In Appendix \ref{a:self} we check the self-consistency of our approximation, showing that it is compatible with \eref{eq:condA}.
We leave further investigations to future works.

Proposition \ref{L:Wick} suggests that in many cases of interest the effective Hamiltonian describing the pre-relaxation limit takes the form
\be\label{eq:opform}
H_{\rm eff}=\frac{1}{L^{n_1-1}}H^{(1)}_1 \dots H^{(1)}_{n_1}+\dots +\frac{1}{L^{n_m-1}}H^{(m)}_1 \dots H^{(m)}_{n_m}\, ,
\ee
where $H_j^{(\ell)}$ are \mbox{(quasi-)local} (\emph{i.e.} their density is \mbox{(quasi-)local}, see \emph{e.g.} \cite{IP:Drude,F:charge})  translation invariant operators (\emph{i.e.} $n$-site shift invariant for some $n\in \mathbb{N}$).
Indeed, provided that the anomalous terms in \eref{eq:anom} can be disregarded, similar factorisations appear whenever the unperturbed Hamiltonian is noninteracting (\emph{e.g.}, in the model considered in \cite{E:preT}). 

Hamiltonians of the form \eref{eq:opform} are therefore the perfect workbench for pre-relaxation or pre-thermalisation issues.  

We notice that the non-equilibrium dynamics generated by a subclass of Hamiltonians of the form \eref{eq:opform} have been already worked out in \cite{SB:conn}.
The authors considered `completely connected quantum models', in which the Hamiltonian is symmetric under any permutation of the sites, and exhibited a mapping onto an effective classical Hamiltonian dynamics. 

We also point out that the simplest models of the form \eref{eq:opform} (\emph{e.g.} Curie-Weiss quantum Heisenberg models) have often been used as toy models to investigate the statistical properties in the presence of long range interactions \cite{K:long}.

The rest of the paper will be focussed on the following points:
\begin{enumerate}
\item \label{I:1}Solution of the non-equilibrium problem for Hamiltonians of the form \eref{eq:opform};
\item \label{I:2}Characterisation of the pre-relaxation limit in an interacting model, also in the presence of perturbations that break integrability;
\item \label{I:3}Non-equilibrium time evolution under \eref{eq:TFICNI1}.
\end{enumerate}
For the sake of clarity, we stress again that \eref{I:2} relies on two assumptions:
\begin{enumerate}[(a)]
\item \label{ass:a}In the limit $g\rightarrow 0$ with $g t$ finite, the time evolution under $H=H_0+g V$ can be split in two steps: 
\begin{enumerate}[1.]
\item infinite time evolution under the unperturbed Hamiltonian $H_0$, which is supposed to give rise to a generalised Gibbs ensemble $\ket{\Psi_0}\bra{\Psi_0}\rightarrow\rho_{\rm GGE}$;
\item time evolution with rescaled time $T=gt$ under the effective Hamiltonian given by the perturbation $V$ averaged with respect to $H_0$ \eref{eq:Vbar};
\end{enumerate}
\be
 e^{-i (H_0+g V)t}\ket{\Psi_0}\bra{\Psi_0}e^{i (H_0+g V)t}\sim e^{-i g t \bar V}\rho_{\rm GGE}e^{i g t \bar V}\, .
\ee
\item \label{ass:b}The ``anomalous terms'' that appear in the time average of $V$ give a negligible contribution (\emph{cf.} Appendix \ref{a:self}, Property \ref{L:Wick} and discussion below).
\end{enumerate}
On the other hand, \eref{I:1} and \eref{I:3} will be treated as \emph{ab initio} problems.

\section{Solution of the non-equilibrium problem}\label{s:mf} %

In this section we work out Problem \eref{I:1}.
We are going to show that, despite the nonlocal appearance, operators of the form \eref{eq:opform} generate a dynamics which is equivalent to that of a \mbox{(quasi-)local} time-dependent mean-field Hamiltonian.

Here we only report some results and three useful corollaries, the details of the derivation can be found in Appendix \ref{a:MF}.  

For the sake of simplicity we only consider cases in which  $H_j^{(\ell)}$ have local densities, however, as far as we can see, all the results can be generalised to quasi-local operators with tails that decay exponentially with the distance. 

In the light of \eref{eq:opform}, we define a class of operator $\mathcal E$ as follows: 
\begin{definition}
We say that an operator acting on a spin-$\frac{1}{2}$ chain belongs to the class $\mathcal E$ if it is written as in \eref{eq:opform}, namely as a finite linear combination of operators of the form
\be\label{eq:form}
\frac{1}{L^{n-1}}H_1\cdots H_n\, ,
\ee
where $n$ is finite, $H_j$ are local translation invariant operators, and $L$ is the chain length.
\end{definition}
We consider spin chains so that the local Hilbert space is finite dimensional. This turns out to be a fundamental assumption for most of our results.

One of our goals is to show that the time evolution preserves cluster decomposition properties, which is the key element that allows us to simplify the calculation of expectation  values. For example we have
\begin{lemma}\label{L:1}
Let $\mathcal O\in \mathcal E$ and $\ket{\Psi}$ a state with cluster decomposition properties. The expectation value of $\mathcal O/L$ in $\ket{\Psi}$ can be reduced to the expectation values of the local translation invariant operators it consists of:
\be\label{eq:fact}
\lim_{L\rightarrow\infty }\braket{\Psi|\frac{H_1}{L}\cdots \frac{H_n}{L}|\Psi}=\lim_{L\rightarrow\infty }\prod_j \frac{\braket{\Psi|H_j|\Psi}}{L}\, .
\ee
\end{lemma}
Using this lemma it is rather natural to relate the dynamics under \eref{eq:opform} to that under the mean-field Hamiltonian defined as follows: 
\begin{definition} \emph{Mean-field effective Hamiltonian.}
Let $H\in \mathcal E$. We define the time-dependent mean-field Hamiltonian $H^{\Psi_0}_{\rm MF}(t)$ as the operator resulting from mapping any generic term \eref{eq:form} of the Hamiltonian \eref{eq:opform} to an operator with local density, as follows
\be\label{eq:tran}
\frac{1}{L^{n-1}}H_1\cdots H_n\rightarrow \sum_{j=1}^n \prod_{\ell\neq j}\frac{\braket{\Psi_0|\bar U^\dag(t)H_\ell \bar U(t)|\Psi_0}}{L}H_j\, ,
\ee
where $\bar U(t)$ is the time evolution under $H_{\rm MF}^{\Psi_0}(t)$
\be\label{eq:U}
\bar U(t)={\rm T}\exp\Bigl(-i\int_0^t\mathrm d \tau H^{\Psi_0}_{\rm MF}(\tau)\Bigr)\, .
\ee
Thus, generally $H_{MF}^{\Psi_0}(t)$ must be computed in a self-consistent way.  
\end{definition}

For example, the Hamiltonian 
\be
H=-\frac{1}{4}\sum_\ell^L\Bigl(\sigma_\ell^x\sigma_{\ell+1}^x+\sigma_\ell^y\sigma_{\ell+1}^y\Bigr)+\frac{\lambda}{L}\Bigl(\sum_\ell^L\sigma_\ell^z\Bigr)^2
\ee
belongs to $\mathcal E$. 
In this trivial case $\sum_\ell\sigma_\ell^z$ commutes with $H$, so the mean-field Hamiltonian is independent of time and it is given by
\be
\fl\qquad\quad H_{\rm MF}^{\Psi_0}(t)=-\frac{1}{4}\sum_\ell\Bigl(\sigma_\ell^x\sigma_{\ell+1}^x+\sigma_\ell^y\sigma_{\ell+1}^y\Bigr)+2\lambda\braket{\Psi_0|\frac{1}{L}\sum_\ell\sigma_\ell^z|\Psi_0}\sum_\ell\sigma_\ell^z\, .
\ee

We point out that the expectation value (per unit length) of $H\in \mathcal E$ in the state $\bar U(t)\ket{\Psi_0}$ is generally different from that of $H_{\rm MF}^{\Psi_0}$:
\bea
\fl\quad\braket{\Psi_0|\bar U^\dag(t)\frac{1}{L^{n}} H_1\cdots H_n\bar U(t)|\Psi_0}=\prod\limits_{j=1}^n\braket{\Psi_0|\bar U^\dag(t)\frac{H_j}{L}\bar U(t)|\Psi_0} \nn
\fl\quad\braket{\Psi_0|\bar U^\dag(t)\sum\limits_{j=1}^n \prod\limits_{\ell\neq j}\frac{\braket{\Psi_0|\bar U^\dag(t)H_\ell \bar U(t)|\Psi_0}}{L}H_\ell \bar U(t)|\Psi_0}=n\prod\limits_{j=1}^n\braket{\Psi_0|\bar U^\dag(t)\frac{H_j}{L}\bar U(t)|\Psi_0}\, .
\eea

The main property that is proved in Appendix \ref{a:MF} is the exactness of the mean-field description in the thermodynamic limit:
\begin{lemma}\label{T:1}
Let $\ket{\Psi_0}$ be a translation invariant state with cluster decomposition properties and $H,\mathcal O\in \mathcal E$. 
Let the expectation value of $\mathcal O$ in the state that time evolves with $H_{\rm MF}^{\Psi_0}(t)$ be an analytic function of $t$ in the strip $|\mathrm{Im}[t]|<r$, with $r$ a nonzero constant.
In the thermodynamic limit, the time evolution with $H$ can be replaced by the time evolution with the mean-field Hamiltonian:
\be\label{eq:rep}
\lim_{L\rightarrow\infty}\braket{\Psi_0|e^{i H t}\frac{\mathcal O}{L} e^{-i H t}|\Psi_0}=\lim_{L\rightarrow\infty}\braket{\Psi_0|\bar U^\dag(t)\frac{\mathcal O}{L}\bar U(t) |\Psi_0}\, .
\ee
\end{lemma}

\begin{remark}
The validity of the hypothesis of analyticity on a strip can be verified \emph{a posteriori}. The idea is the following. The self-consistent mean-field problem can be generally recast into an infinite nonlinear system of ordinary differential equations. The finiteness of $n$ in  \eref{eq:form} implies that the system can be written as $\dot{\vec u}=\vec F(\vec u, t)$, with $\vec F$ a polynomial. If the system was finite, the solution would have been analytic. This is not always the case for an infinite system but, in practice, the numerical solution is obtained by introducing a cutoff parameter $N$ that makes the system finite. 
If the mean-field time evolution had a point of non-analyticity, the solution of the system of equations should display a non-trivial dependence of the mean-field parameters on the cutoff as $N\rightarrow\infty$. 
\end{remark}

\begin{corollary}\label{C:1}
Lemma \ref{T:1} holds true in particular for local operators.
\end{corollary}
The local equivalence with the mean-field time evolution can also be expressed in terms of reduced density matrices: 
\begin{corollary}\label{C:2}
Let $\ket{\Psi_0}$ a translation invariant state with cluster decomposition properties and $H\in \mathcal E$. 
In the thermodynamic limit, the time evolution of the reduced density matrix (RDM) of some spin block $S$ is equal to the RDM in the state that time evolves with the mean-field Hamiltonian: 
\be\label{eq:RDM}
\rho_S(t)=\tr_{\bar S}[e^{-i H t}\ket{\Psi_0}\bra{\Psi_0}e^{i H t}]=\tr_{\bar S}[\bar U(t)\ket{\Psi_0}\bra{\Psi_0}\bar U^\dag(t)]\, .
\ee
\end{corollary}

The previous lemmas and corollaries are sufficient to reduce the time evolution under $H\in\mathcal E$ to the time evolution under a local time-dependent Hamiltonian. There is however another simple corollary to Lemma \ref{T:1} that will be useful to assess whether or not at large times it is possible to encode the entire information about the initial state in a finite number of parameters (`thermal-like behaviour').
\begin{corollary}\label{C:3}
Let $H\in\mathcal E$ and $\ket{\Psi}$ a state with cluster decomposition properties. 
If $\ket{\Psi}$ is an excited state of the corresponding mean-field Hamiltonian $H_{\rm MF}^{\Psi}$
\be
H_{\rm MF}^{\Psi}\ket{\Psi}=E_\Psi\ket{\Psi}\, ,
\ee
the expectation value of local observables in $e^{-i H t}\ket{\Psi}$ is independent of time. Therefore, $\ket{\Psi}$ behaves locally as an excited state of $H$. 

The reverse is also true. If an excited state of $H$ is locally equivalent to a state with cluster decomposition properties, then the latter is (equivalent to) an excited state of the corresponding mean-field Hamiltonian.
\end{corollary}

\subsection{Time-dependent GGE} \label{ss:t-dGGE} %

We are now in a position to justify \eref{eq:toGGE}, and in turn \eref{eq:toGGE2} and \eref{eq:tGGE}. 

In the limit of small $g$ the expectation value of a local observable $\mathcal O$ reads as \eref{eq:prsl}
\be
\braket{\Psi_0|e^{i H_0 T/g}e^{i T \bar V}\mathcal Oe^{-i T \bar V} e^{-i H_0 T/g}|\Psi_0}\, .
\ee
Since $H_0$ is local, the state $e^{-i H_0 T/g}\ket{\Psi_0}$ has cluster decomposition properties beyond some typical distance proportional to $T/g$ (in order to be outside of the light cone). From Corollary \ref{C:1} it follows that the time evolution under $\bar V$ is equivalent to that under the corresponding mean-field operator. We indeed only need $J T\ll g L$ (\emph{cf}. \eref{eq:rhs}), which is trivially satisfied in the thermodynamic limit. Thus we obtain
\be
\braket{\Psi_T|U_{\bar V}^\dag(T)\mathcal O U_{\bar V}^{\phantom \dag}(T) |\Psi_T}\, ,
\ee
with
\be\label{eq:UV}
U_{\bar V}(t)={\rm T}\exp\Bigl(-i\int_0^t\mathrm d \tau \bar V^{\Psi_T}_{\rm MF}(\tau)\Bigr)
\ee
and
\be
\ket{\Psi_T}=e^{-i H_0 T/g}\ket{\Psi_0}\, .
\ee
Incidentally, we notice that the time-ordering in \eref{eq:UV} can not be simplified because $\bar V^{\Psi_T}_{\rm MF}$ is generally written in terms of conservation laws that are not in involution with one another. 

For the sake of simplicity we assume that the time-dependent coupling constants of $\bar V^{\Psi_T}_{\rm MF}$ are bounded.
The operator $U_{\bar V}^\dag(T)\mathcal O U_{\bar V}^{\phantom \dag}(T) $ is then quasi-local with a typical range $\xi$ proportional to $T$ \cite{BHV:LRbounds}.
On the other hand $\ket{\Psi_T}$ is the time evolution of $\ket{\Psi_0}$ at the time $(\infty \leftarrow )T/g\gg T\sim \xi$, which is the limit in which it is reasonable to expect that the state can be replaced by the corresponding generalised Gibbs ensemble (of the unperturbed Hamiltonian)
\be\label{eq:UVGGE}
\fl\qquad\qquad\braket{\Psi_T|U_{\bar V}^\dag(T)\mathcal O U_{\bar V}^{\phantom \dag}(T) |\Psi_T}\sim \tr[\rho_{\rm GGE} U_{\bar V}^\dag(T)\mathcal O U_{\bar V}^{\phantom \dag}(T) ]\, .
\ee
 The operator $\bar V^{\Psi_T}_{\rm MF}$ is obtained self-consistently by computing the expectation values of (quasi-)local conservation laws, which can be obtained from \eref{eq:UVGGE}. Therefore, in the definition \eref{eq:tran} of the mean-field Hamiltonian we can replace $\ket{\Psi_0}$ by $\rho_{\rm GGE}$
\be
\frac{1}{L^{n-1}}H_1\cdots H_n\rightarrow \sum_{j=1}^n \prod_{\ell\neq j}\frac{\tr[\rho_{\rm GGE}U_{\bar V}^\dag(T)H_\ell U_{\bar V}(T)]}{L}H_\ell\, ,
\ee
which is consistent with \eref{eq:toGGE2}. We denote by $\bar H_{\rm MF}(T)$ the mean-field Hamiltonian with the expectation values computed in the GGE.

\section{Pre-relaxation in XYZ models}\label{s:pre-relI}%

In this section we investigate the (integrable) XYZ spin-$\frac{1}{2}$ chain in the limit of small anisotropy in the $z$ direction and also the effect of a small perturbation that breaks integrability.
The Hamiltonian 
\be\label{eq:XYYU}
\fl\quad  H=J \sum_\ell\Bigl(\frac{1+\gamma}{4}\sigma_{\ell}^x\sigma_{\ell+1}^x+\frac{1-\gamma}{4}\sigma_{\ell}^y\sigma_{\ell+1}^y+\frac{g}{4}\sigma_{\ell}^z\sigma_{\ell+1}^z+\frac{g U }{4}\sigma_{\ell}^z\sigma_{\ell+2}^z\Bigr)+\frac{gh}{2}\sum_\ell\sigma_\ell^z
\ee
has the form \eref{eq:H} with $H_0=H_{\rm XY}$ \eref{eq:XY} and 
\be\label{eq:Vj}
V=\frac{J}{4}\sum_{\ell}(\sigma_\ell^z\sigma_{\ell+1}^z+U\sigma_\ell^z\sigma_{\ell+2}^z)+\frac{h}{2}\sum_{\ell}\sigma_\ell^z\, .
\ee
For a fixed $g\neq 0$, the model is integrable for  $J U=h=0$, corresponding to the spin-$\frac{1}{2}$  XYZ  model, and for $J U=\gamma=0$, corresponding to the XXZ spin-$\frac{1}{2}$ chain; otherwise it is non-integrable.

Following Sections \ref{s:pre-rel} and \ref{s:int}, in the pre-relaxation limit $g\ll 1$ with $gt\sim O(g^0)$, the initial state can be replaced by the corresponding GGE of the unperturbed Hamiltonian
\be
\ket{\Psi_0}\rightarrow \rho_{\rm GGE}=\lim_{|S|\rightarrow\infty}\lim_{t\rightarrow\infty}\tr_{\bar S}[e^{-i H_{\rm XY} t}\ket{\Psi_0}\bra{\Psi_0}e^{i H_{\rm XY}t}]\, ,
\ee
and $V$ by the time averaged perturbation \eref{eq:Vbar}. We notice that the free Hamiltonian $H_{\rm XY}$ does not play any role in the pre-relaxation limit, because it commutes with $\rho_{\rm GGE}$.  
The mapping into a mean-field problem can be decomposed in the following steps:
\begin{itemize}
\item[-] Compute the time averaged perturbation $\bar V$;
\item[-] Construct the mean-field Hamiltonian $\bar H_{\rm MF}$;
\item[-] Solve the time evolution under $\bar H_{\rm MF}$ for \emph{any} local observable.
\end{itemize}
Some properties of the unperturbed Hamiltonian $H_{\rm XY}$ dramatically simplify the first step. $H_{\rm XY}$ is mapped to noninteracting fermions by a Jordan-Wigner transformation. Up to irrelevant (to our purposes) boundary terms, it can be written as follows
\be\label{eq:Hfrep}
H_{XY}\sim\frac{1}{4}\sum_{\ell, n}\left(\begin{array}{cccc}
a_{2\ell-1}^x&a_{2\ell-1}^y&a_{2\ell}^x&a_{2\ell}^y
\end{array}\right)[\mathcal H]^{(2)}_{\ell n}\left(
\begin{array}{cccc}
a_{2n-1}^x \\
a_{2n-1}^y \\
a_{2n}^x \\
a_{2n}^y
\end{array}\right)
\ee
with $a_\ell^\alpha$ the Majorana fermions \eref{eq:Maj}; 
$\mathcal H$ is the block-circulant matrix
\be\label{eq:Helln}
[\mathcal H]^{(2)}_{\ell n}\sim \int_{-\pi}^\pi\frac{\mathrm d k}{2\pi} e^{-i(n-\ell) k} \mathcal H^{(2)}(k)\, ,
\ee
where the 4-by-4 matrix $\mathcal H^{(2)}(k)$ is usually called \emph{symbol} (see also Appendix \ref{a:free}) and it is given by
\be\label{eq:Hsimb}
\mathcal H^{(2)}(k)= -\varepsilon_k \sigma^x e^{ i {k}/{2}\sigma^z}\otimes\sigma^y e^{ i\theta_k\sigma^z}\, ;
\ee 
$\varepsilon_k$ and $\theta_k$ are the dispersion relation and the Bogoliubov angle, respectively
\be
\fl\qquad\qquad\varepsilon_k=J\sqrt{\cos^2 k/2+\gamma^2\sin^2 k/2}\qquad e^{i\theta_k}=\frac{\cos k/2+ i\gamma\sin k/2}{\sqrt{\cos^2 k/2+\gamma^2\sin^2 k/2}}\, .
\ee
Here we have chosen the two-site shift invariant representation of the Hamiltonian (\emph{i.e.} we gathered together the fermionic degrees of freedom of pairs of adjacent sites) in order to be able to treat a larger class of initial states.

In translation invariant noninteracting models almost any calculation can be traced back to operations on the symbol associated with the operator, which is the Fourier transform of a block-row of the block-circulant matrix that appears in the fermionic representation of the operator as in \eref{eq:Hfrep} (see also Appendix \ref{a:free}). More generally  the two-site representation of the symbol is a 4-by-4 Hermitian matrix, function of the momentum and odd under simultaneous transposition and reversion of the momentum. 
A $2n$-by-$2n$ symbol completely identifies a noninteracting operator that is translation invariant by $k$ sites, with $k$ a divisor of $n$, by the same kind of relations that we wrote for the Hamiltonian (\emph{i.e.} \eref{eq:Hfrep}, \eref{eq:Helln} and \eref{eq:Hsimb}). Thus, we will often report the symbols instead of  writing the operators explicitly.

Coming back to the calculation of $\bar V$, we find that the three constituents of the interaction term in \eref{eq:Vj} have the following fermionic representation
\be
\frac{1}{4}\sum_\ell\sigma_\ell^z\sigma_{\ell+j}^z=\frac{1}{4}\sum_\ell i a_\ell^ya_\ell^xi a_{\ell+j}^ya_{\ell+j}^x\,,\qquad j=1,2\,,
\ee
\be
\frac{1}{2}\sum_\ell \sigma_\ell^z=\frac{1}{2}\sum_\ell  i a_\ell^ya_\ell^x\, .
\ee
Therefore, on the basis of our assumptions and decomposition \eref{eq:Wick},  in the limit $g\ll 1$ with $gt\sim O(g^0)$,  we expect the local Hamiltonian \eref{eq:XYYU} to be dynamically equivalent to the following nonlocal one
\be\label{eq:Hrel}
H\rightarrow \bar H=H_{\rm XY}+g \bar V(U, h)
\ee
where $H_{\rm XY}$ is the XY Hamiltonian \eref{eq:XY} and the nonlocal perturbation is given by
\be\label{eq:Vbar1}
\fl\qquad {\bar V(U,h)}=\frac{J}{L}\sum_{s=0}^1\sum_{j=1}^2U^{j-1}\Bigl((-1)^{sj}\bar{H}_{s}^z \bar{H}_{s}^z+\bar{H}_{s,j}^{xy} \bar{H}_{s,j}^{y x}-\bar{H}_{s,j}^{xx} \bar{H}_{s,j}^{y y}\Bigr)+h \bar{H}_0^z\, .
\ee
The time averaged quadratic operators appearing on the right hand side of \eref{eq:Vbar1} are the fundamental blocks of \eref{eq:Wick} and read as
\bea
\label{eq:timeavops}
\fl\qquad\qquad \bar H_s^z=\overline{\frac{1}{2}\sum_\ell (-1)^{s\ell}\sigma_\ell^z}=\overline{\frac{1}{2}\sum_\ell (-1)^{s\ell} i a_\ell^ya_\ell^x}\nn
\fl\qquad\qquad \bar H_{s,j}^{xy}=\overline{\frac{1}{2}\sum_\ell (-1)^{s\ell}\sigma_\ell^x (\sigma_{\ell+1}^z)^{j-1}\sigma_{\ell+j}^y}=\overline{\frac{1}{2}\sum_\ell (-1)^{s\ell} (-i) a_\ell^ya_{\ell+j}^y}\nn
\fl\qquad\qquad \bar H_{s,j}^{yx}=\overline{\frac{1}{2}\sum_\ell (-1)^{s\ell}\sigma_\ell^y (\sigma_{\ell+1}^z)^{j-1}\sigma_{\ell+j}^x}=\overline{\frac{1}{2}\sum_\ell (-1)^{s\ell} i a_\ell^x a_{\ell+j}^x}\nn
\fl\qquad\qquad \bar H_{s,j}^{xx}=\overline{\frac{1}{2}\sum_\ell (-1)^{s\ell}\sigma_\ell^x (\sigma_{\ell+1}^z)^{j-1}\sigma_{\ell+j}^x}=\overline{\frac{1}{2}\sum_\ell (-1)^{s\ell} (-i) a_\ell^y a_{\ell+j}^x}\nn
\fl\qquad\qquad \bar H_{s,j}^{yy}=\overline{\frac{1}{2}\sum_\ell (-1)^{s\ell}\sigma_\ell^y (\sigma_{\ell+1}^z)^{j-1}\sigma_{\ell+j}^y}=\overline{\frac{1}{2}\sum_\ell (-1)^{s\ell} i a_\ell^x a_{\ell+j}^y}\, .
\eea
Since we have to compute the time average of quadratic operators evolving according to a noninteracting Hamiltonian (\emph{cf.} \eref{eq:Vbar} and \eref{eq:Hfrep}), we can use \eref{timeevolution}. This allows to find the following exact result
\bea\label{eq:timeav}
\fl\qquad\bar{\mathcal{O}}(k)=\lim_{T\rightarrow+\infty}\frac{1}{T}\int_{0}^{T}\hspace{-0.2cm}\textrm{d}t\,\mathcal{O}(k,t)=\lim_{T\rightarrow+\infty}\frac{1}{T}\int_{0}^{T}\hspace{-0.2cm}\textrm{d}t\,e^{ i \mathcal{H}^{(2)}(k) t}\mathcal{O}(k)e^{- i \mathcal{H}^{(2)}(k) t} \nn
=\frac{1}{2}\mathcal{O}(k,0)+\frac{1}{2}\left[\sigma^x e^{i\frac{k}{2}\sigma^z}\otimes\sigma^y e^{ i\theta_k\sigma^z}\right]\mathcal{O}(k,0)\left[\sigma^x e^{ i\frac{k}{2}\sigma^z}\otimes\sigma^y e^{ i\theta_k\sigma^z}\right]\, , 
\eea
where $\mathcal{O}(k)$ is the symbol of a quadratic operator and $\overline{\phantom{(}\!\!\!\cdots\phantom{)}\!\!\!}$ denotes the time average.
The symbols of the operators \eref{eq:timeavops} read
\bea
\fl\qquad\quad\bar{{H}}^z_s(k) &=&\frac{J}{\varepsilon_k^2} (\delta_{s,0}\mathcal{Q}_2(k)-\gamma \delta_{s,1}\mathcal{Q}_8(k))\nn
\fl\qquad\quad\bar{{H}}^{xy}_{s, 1}(k) &=&\delta_{s,0}\mathcal{Q}_4(k)+\delta_{s,1} \mathcal{Q}_6(k)\nn
\fl\qquad\quad\bar{{H}}^{xy}_{s,2}(k) &=&\delta_{s,0} \mathcal{Q}_3(k)+\delta_{s,1} \mathcal{Q}_7(k)\nn
\fl\qquad\quad\bar{{H}}^{yx}_{s, 1}(k) &=&-\delta_{s,0} \mathcal{Q}_4(k)+\delta_{s,1}\mathcal{Q}_6(k)\nn
\fl\qquad\quad\bar{{H}}^{yx}_{s,2}(k) &=&- \delta_{s,0} \mathcal{Q}_3(k)+\delta_{s,1} \mathcal{Q}_7(k)\nn
\fl\qquad\quad\bar{{H}}^{xx}_{s,1}(k) &=&-\frac{J}{2 \varepsilon_k^2}\left((1+\gamma)+(1-\gamma)\cos k \right)(\delta_{s,0}\mathcal{Q}_1(k)-\delta_{s,1}\mathcal{Q}_5(k))\nn
\fl\qquad\quad\bar{{H}}^{xx}_{s,2}(k) &=&- \frac{J }{\varepsilon_k^2}\left[\gamma + (1-\gamma)(s+ \cos k) \right](\delta_{s,0}\mathcal{Q}_2(k)
+ \delta_{s,1}\mathcal{Q}_8(k))\nn
\fl\qquad\quad\bar{{H}}^{yy}_{s,1}(k) &=&-\frac{J }{2\varepsilon^2_k}\left((1-\gamma)+(1+\gamma)\cos k \right)(\delta_{s,0}\mathcal{Q}_1(k)+
\delta_{s,1}\mathcal{Q}_5(k))\nn
\fl\qquad\quad\bar{{H}}^{yy}_{s,2}(k) &=& \frac{J }{\varepsilon_k^2}\left( \gamma+s(1-\gamma)-(1+\gamma)(-1)^s\cos k\right)(\delta_{s,0}\mathcal{Q}_2(k)+
\delta_{s,1}\mathcal{Q}_8(k))\,.\label{eq:HQ8}
\eea
Here we expressed the results in terms of the symbols of the local charges of $H_{XY}$ \cite{F:super}
\begin{eqnarray}
 \mathcal{Q}_1(k)=\mathcal{I}^{+\rm (e)}_1(k)=\varepsilon_k\,\bigl[\sigma^xe^{i\frac{k}{2}\sigma^z}\bigr]\otimes\bigr[\sigma^y e^{i \theta_k\sigma^z}\bigr]
 \nn
\mathcal{Q}_2(k)=\mathcal{I}^{+\rm (o)}_1(k)=\cos(k/2) \varepsilon_k \,\1\otimes\bigl[\sigma^y e^{i \theta_k\sigma^z}\bigr]\nn
\mathcal{Q}_3(k)=\mathcal{I}^{-\rm (e)}_1(k)=\sin( k)\, \1\otimes\1 \nn
\mathcal{Q}_4(k)=\mathcal{I}^{-\rm (o)}_1(k)=\sin(k/2)\,\bigl[\sigma^xe^{i\frac{k}{2}\sigma^z}\bigr]\otimes\1\nn
\mathcal{Q}_5(k)=\mathcal{J}^{+\rm (e)}_1(k)=\varepsilon_k\, \bigl[\sigma^ye^{i\frac{k}{2}\sigma^z}\bigr]\otimes\bigl[\sigma^x e^{i \theta_k \sigma^z}\bigl]\nn
\mathcal{Q}_6(k)=\mathcal{J}^{+\rm (o)}_1(k)=\cos(k/2)\,\bigl[\sigma^ye^{i\frac{k}{2}\sigma^z}\bigr]\otimes\sigma^z\nn
\mathcal{Q}_7(k)=\mathcal{J}^{-\rm (e)}_1(k)=\sin(k)\,\sigma^z\otimes\sigma^z \nn
\mathcal{Q}_8(k)=\mathcal{J}^{-\rm (o)}_1(k)=\sin(k/2)\varepsilon_k\,\sigma^z\otimes\bigr[\sigma^x e^{i \theta_k \sigma^z}\bigl]\label{eq:C8}\,.
\end{eqnarray}
The first four symbols correspond to one-site shift invariant operators (the standard conservation laws of the quantum XY model), while the others change sign under a shift by one site.  
 
We remind the reader that from the symbol of an operator it is possible to infer its locality properties \cite{F:super}. In particular, a smooth symbol is associated with a quasi-local operator. If in addition the symbol has a finite number of nonzero Fourier components, as in  \eref{eq:C8}, the associated operator is local. Equations \eref{eq:HQ8} imply that  $\bar H^{xy}_s, \bar H^{yx}_s$ are local while $\bar H^{xx}_s, \bar H^{yy}_s, \bar H^{z}_s$ are quasi-local, thus the Hamiltonian \eref{eq:Hrel} is a member of the quasi-local extension of the class $\mathcal E$ studied in Section \ref{s:mf}.
 As pointed out in \Sref{s:mf}, we expect all the theorems of \Sref{s:mf}, in particular Corollary \ref{C:2}, to remain valid also for quasi-local operators. This guarantees the time evolution generated by \eref{eq:Hrel} to be \emph{locally} equivalent to the one generated by the following mean-field Hamiltonian 
\bea
\label{eq:meanfieldXYZ}
\fl \bar H_{MF}(T)= H_{XY}+2 J g \sum_s ((-1)^{s}+U)\frac{\braket{\bar{H}^z_s}_T}{L} \bar{H}^z_s\nn
+Jg\sum_{s,j} U^{j-1}\Bigl(\frac{\braket{\bar{H}^{xy}_{s j}}_T}{L} \bar{H}^{yx}_{s j}+\frac{\braket{\bar{H}^{yx}_{s j}}_T}{L} \bar{H}^{xy}_{s j}\Bigr)\nn
\qquad-Jg\sum_{s,j} U^{j-1}\Bigl(\frac{\braket{\bar{H}^{xx}_{s j}}_T}{L}\bar{H}^{yy}_{s j}+\frac{\braket{\bar{H}^{yy}_{s j}}_T}{L} \bar{H}^{xx}_{s j}\Bigr)+hg \bar{H}^z_0\, ,
\eea
where $\braket{\mathcal O}_T$ is the expectation value of the operator $\mathcal O$ in the mean-field description (\emph{cf.} \eref{eq:UVGGE})
\be
\label{eq:mfexpval}
\braket{\mathcal O}_T=\textrm{Tr}\left[ U_{\bar V}(T) \rho_{GGE} U^{\dag}_{\bar V}(T) \mathcal O\right]\,.
\ee
To determine the time evolution generated by $\bar H_{MF}(T)$ we need to solve the self-consistency conditions encoded in \eref{eq:meanfieldXYZ} and \eref{eq:mfexpval}. To this end, it is again convenient to exploit the representation in terms of symbols. Using \eref{eq:HQ8}, the symbol $\mathcal{H}_{MF}(k, T)$ of the time-dependent mean-field Hamiltonian can be written in terms of the symbols $\{\mathcal Q_\alpha(k),\alpha=1,\dots,8\}$, as follows
\footnote{From now on we set $J=1$.}
\be
\label{eq:MFO}
\mathcal{H}_{MF}(k, T)=-\mathcal{Q}_{1}(k)+g\mathcal{V}_{MF}(k, T)\,,
\ee
\be
\label{eq:VMFO}
\mathcal{V}_{MF}(k,T)=\frac{h}{\varepsilon_k^2}\mathcal{Q}_{2}(k)+ \sum_{\alpha=1}^{8}c_{\alpha}(k;\tilde y_{\alpha})\mathcal{Q}_{\alpha}(k)\,.
\ee
The coefficients are given by
\bea\label{eq:c8}
\fl \qquad  c_{1}(k; \tilde y_1)=-\frac{1+\cos k}{2\varepsilon^2_k}(\tilde y_{1}^{(0)}+\tilde y_{1}^{(1)})+\gamma^2\frac{1-\cos k}{2\varepsilon^2_k}(\tilde y_{1}^{(0)}-\tilde y_{1}^{(1)})\, \nn
\fl \qquad c_{2}(k; \tilde y_2)=2\frac{1+U}{ \varepsilon_k^2} \tilde y_{2}^{(0)}-2 U  \frac{\cos k}{\varepsilon_k^2}\tilde y_{2}^{(1)}
+2 U \gamma^2 \frac{1-\cos k}{\varepsilon_k^2}(\tilde y_{2}^{(0)}-\tilde y_{2}^{(1)})\, \nn
\fl \qquad c_{3}(k;\tilde y_{3})=-U(1+ \gamma^2) \tilde y_{3}^{(0)}-U(1-\gamma^2) \tilde y_{3}^{(1)}\, \nn
\fl \qquad c_{4}(k;\tilde y_{4})=-(1+ \gamma^2) \tilde y_{4}^{(0)}-(1-\gamma^2) \tilde y_{4}^{(1)}\, \nn
\fl \qquad c_{5}(k;\tilde y_5)= \frac{ 1+\cos k}{2\varepsilon^2_k} (\tilde y_{5}^{(0)}+\tilde y_{5}^{(1)})-\gamma^2 \frac{1-\cos k}{2\varepsilon^2_k} (\tilde y_{5}^{(0)}-\tilde y_{5}^{(1)})\, \nn
\fl \qquad c_{6}(k; \tilde y_{6})= (1+ \gamma^2) \tilde y_{6}^{(0)}+(1-\gamma^2) \tilde y_{6}^{(1)}\, \nn
\fl \qquad c_{7}(k; \tilde y_{7})= U(1+ \gamma^2) \tilde y_{7}^{(0)}+U(1-\gamma^2) \tilde y_{7}^{(1)}\, \nn
\fl \qquad c_{8}(k; \tilde y_{8})=2 \gamma^2\frac{U-1}{\varepsilon_k^2} \tilde y_{8}^{(0)}-2 \gamma^2 U \frac{\cos k}{ \varepsilon_k^2}  \tilde y_{8}^{(1)}+2 U \frac{1+ \cos k}{\varepsilon_k^2} (\tilde y_{8}^{(0)}+\tilde y_{8}^{(1)})\, , 
\eea
where we defined
\be
\tilde y_\alpha^{(\ell)}(T)=\int_{-\pi}^\pi\frac{\mathrm d p}{2\pi}\frac{\cos(\ell p)}{\varepsilon_k^2}y_\alpha(p,T)\, ,
\ee
and
\bea
\label{eq:defy}
y_\alpha(k,T)=\frac{1}{8}\textrm{Tr}\left[U_{\mathcal H_{MF}}(k,T) \Gamma_{GGE}(k)U^{\dag}_{\mathcal H_{MF}}(k,T) \mathcal{Q}_\alpha(k)\right]\, \nn
U_{\mathcal H_{MF}}(k,T)=\textrm{T}\exp\left[-i \int_0^{T}\textrm{d}s\,\mathcal V_{MF}(k,s)\right]\, .
\eea
By taking the first derivative of \eref{eq:defy} with respect to $T$ and using the (closed) commutator algebra of $\mathcal Q_\alpha(k)$ we get
\be
\label{eq:equations}
\fl\qquad\qquad \dot{y}_\alpha(k,T)=\frac{h}{\varepsilon_k^2}\sum_{\gamma =1}^{8}f^{2\alpha\gamma}_ky_{\gamma}(k,T)+\sum_{\beta,\gamma =1}^{8}c_{\beta}(k;\tilde y_\gamma)f^{\beta\alpha\gamma}_ky_{\gamma}(k,T)\,.
\ee
The nonzero structure constants $f^{\alpha\beta\gamma}_k$ that are not connected to one another by symmetry are given by 
\bea
\fl\quad f^{562}_k=f^{548}_k=2f^{647}_k=2 &\qquad f^{782}_k=f^{746}_k=2f^{845}_k=-2(1-\cos k)
\nn
\fl\quad f^{584}_k=f^{526}_k=2f^{827}_k=-2\varepsilon_k^2
&\qquad f^{728}_k=f^{764}_k=2f^{625}_k=2(1+\cos k)\nn
\eea
The others follow from $f^{\beta\alpha\gamma}_k=-f^{\alpha\beta\gamma}_k$. In particular $\mathcal Q_{1}(k)$ and $\mathcal Q_{3}(k)$ commute with all the other charges, so $y_1(k)$ and $y_3(k)$ are conserved and the system \eref{eq:equations} is reduced to 6 first order integro-differential equations that depend on a continuous variable $k$. 

The solution of  $\eref{eq:equations}$ entirely determines the time evolution generated by $\bar H_{\rm MF}$. Indeed, the expectation value of any local observable in the pre-relaxation limit can be computed using the Wick theorem with the correlation matrix
\bea
\label{eq:correlmatrix}
\fl
\qquad\qquad\Braket{\left(\begin{array}{c}
a_{2n-1}^x \\
a_{2n-1}^y \\
a_{2n}^x \\
a_{2n}^y
\end{array}\right)
\left(\begin{array}{cccc}
a_{2\ell-1}^x&a_{2\ell-1}^y&a_{2\ell}^x&a_{2\ell}^y
\end{array}\right)}=\nn
\qquad\qquad\qquad\delta_{\ell n} \1_4+\int_{-\pi}^\pi\frac{\mathrm d k}{2\pi} e^{-i(n-\ell)k}\sum_{i=1}^8  \frac{8y_i(k,T)}{\tr(\mathcal Q_i(k)^2)}\mathcal Q_i(k)\, .
\eea
This also means that the reduced density matrices of subsystems are gaussian at any time, so the  two assumptions \eref{ass:a} and \eref{ass:b} could be also reformulated as a single hypothesis of RDMs being gaussian.   

Equations \eref{eq:correlmatrix} and \eref{eq:equations} are the main results of this section: they allow us to compute the expectation values of local observables in the pre-relaxation limit of a weakly interacting model by solving a  nonlinear system of differential equations, which is rather easy from a numerical point of view.

\paragraph{Reflection symmetry}
The Hamiltonian \eref{eq:XYYU} is reflection symmetric, that is to say it is invariant under the transformation
\be
\sigma_\ell^\alpha\rightarrow \sigma_{s+L-\ell}^\alpha\qquad \alpha\in \{x,y,z\}\, ,
\ee
where $s$ is odd for reflections about a bond and even for those about a site. 

The reflection operator acts on the Majorana fermions as follows
\begin{eqnarray}
a_\ell^x\rightarrow i \Bigl(\prod_{j}\sigma_j^z\Bigr)a_{s+L-\ell}^y\nn
a_\ell^y\rightarrow -i \Bigl(\prod_{j}\sigma_j^z\Bigr)a_{s+L-\ell}^x\, .
\end{eqnarray}
Therefore the symbol $\mathcal H$ of a one-site shift invariant operator transforms as
\be
\mathcal H^{(1)}(k)\rightarrow\sigma^y\mathcal H^{(1)}(-k)\sigma^y\, ,
\ee
while for two-site shift invariant operators we find
\be
\mathcal H^{(2)}(k)\rightarrow\left\{
\begin{array}{ll}
\sigma^x\otimes \sigma^y\ \mathcal H^{(2)}(-k)\ \sigma^x\otimes \sigma^y&s \rm{\ odd}\\
e^{-i\frac{k}{2}\sigma^z}\otimes \sigma^y\ \mathcal H^{(2)}(-k)\ e^{i\frac{k}{2}\sigma^z}\otimes \sigma^y&s \rm{\ even}\, .
\end{array}\right.
\ee
The symbols \eref{eq:C8} of the conservation laws have the simple transformation rules
\begin{eqnarray}\label{eq:transf}
\mathcal Q_{1,2}(k)    &\longrightarrow \mathcal Q_{1,2}(k)\nn
\mathcal Q_{3,4}(k)    &\longrightarrow -\mathcal Q_{3,4}(k)\nn
\mathcal Q_{5,6}(k)&\longrightarrow -(-1)^sQ_{5,6}(k)\nn
\mathcal Q_{7,8}(k)&\longrightarrow (-1)^s\mathcal Q_{7,8}(k)\, .
\end{eqnarray}
Since a shift by one site is equivalent to a reflection about a bond followed by a reflection about a site, we recover the transformation rules pointed out 
below \eref{eq:C8}.

If the initial state is reflection symmetric about a \emph{bond}, $\mathcal Q_j(k)=0$ for $j=3,4,7,8$.
Thus, the system of equations \eref{eq:equations} can be reduced to
\bea\label{eq:simplesystem}
\fl\qquad\quad \dot{y}_2(k,T)&=-2 c_{5}(k;\tilde y_5)\varepsilon_k^2 y_{6}(k,T)+ c_{6}(k;\tilde y_6)(1+\cos k) y_{5}(k,T) \nn
\fl\qquad\quad \dot{y}_5(k,T)&=-2 c_{6}(k;\tilde y_6) y_{2}(k,T)+2 \Bigl(\frac{h}{\varepsilon_k^2}+c_{2}(k;\tilde y_2)\Bigr)\varepsilon_k^2 y_{6}(k,T)\nn
\fl\qquad\quad \dot{y}_6(k,T)&=2 c_{5}(k;\tilde y_5) y_{2}(k,T)-   \Bigl(\frac{h}{\varepsilon_k^2}+c_{2}(k;\tilde y_2)\Bigr)(1+\cos k) y_{5}(k,T)\, .
\eea
We numerically identified three different behaviours:
\begin{itemize}
\item \emph{Stationarity}: The expectation values of the observables remain equal to the initial values given by the unperturbed GGE (\fref{f:noevolution}). 
\item \emph{Local relaxation}: The observables relax to a different stationary value: one-site shift invariance is restored in some cases (\fref{f:relaxrest}) while remaining broken in others (\fref{f:relaxbrok}).  
\item \emph{Persistent oscillations}: The amplitude of the oscillations of the expectation values of the observables does not approach zero (\fref{f:persistent}).  
\end{itemize} 
We point out that, even when there is relaxation (at some intermediate times with $Jt\gg g^{-1}$), the stationary state is not thermal, being the local conservation laws of $H_{\rm XY}$ with symbol proportional to $\mathcal Q_1(k)$ and $\mathcal Q_3(k)$ (namely the charges that preserve non-abelian integrability) conserved in the pre-relaxation limit. 

\begin{figure}[tbp]
\begin{center}
\includegraphics[width=0.8\textwidth]{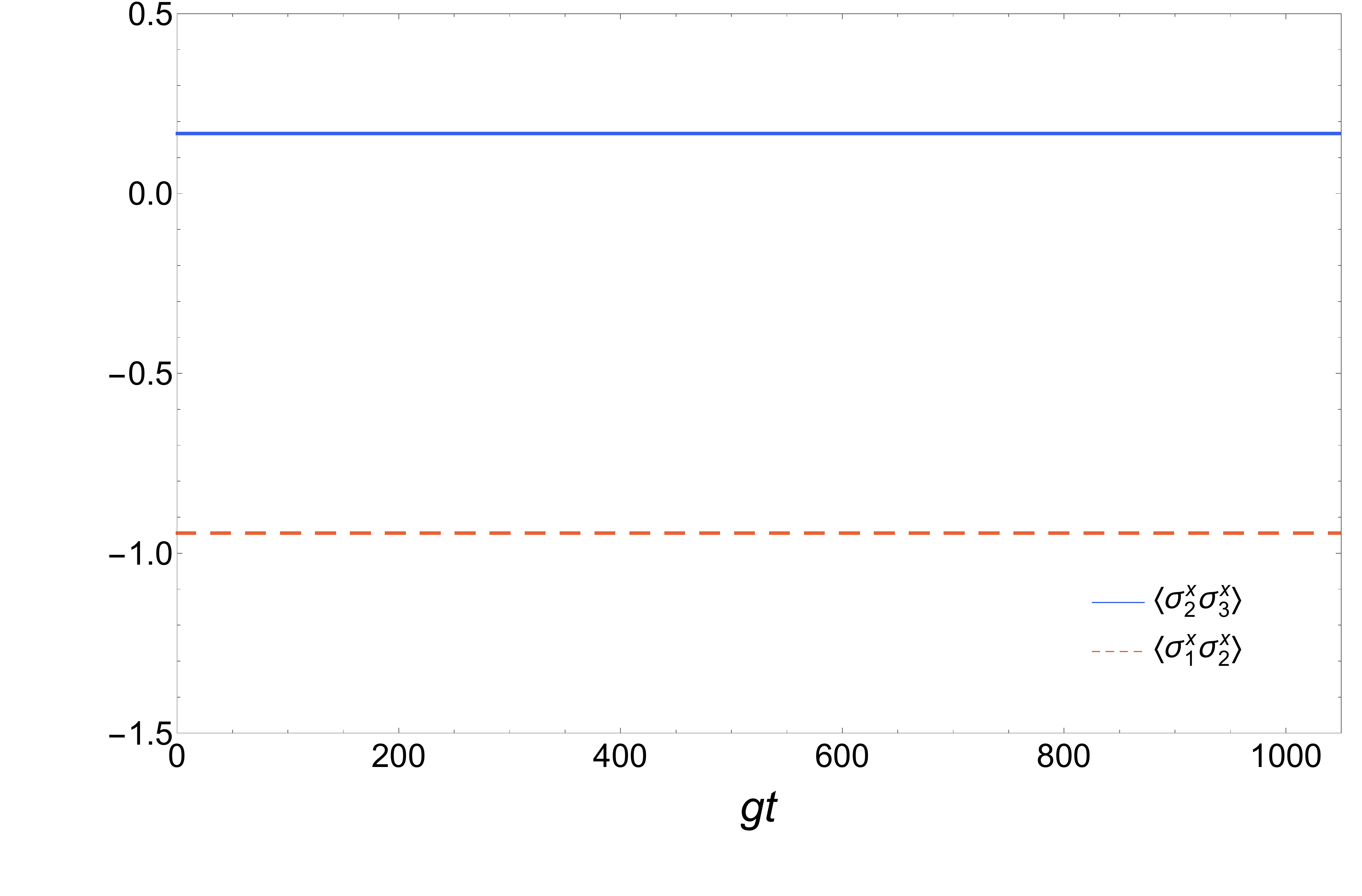}
\end{center}\caption{The time evolution of $\braket{\sigma^x_1\sigma^x_{2}}$ (red dashed) and $\braket{\sigma^x_2\sigma^x_{3}}$ (blue) after a quench from the state $\ket{\textrm{MG}}$ \eref{eq:MG} and Hamiltonian $H$ \eref{eq:XYYU} with $\gamma=2$, $h=0$, and $U=5$. The correlators are stationary.  We find stationary behaviour whenever the initial state is reflection symmetric, $y_2(k)=y_6(k)=0$ (\emph{cf}. \eref{eq:defy}), and $h=0$.}\label{f:noevolution}
\end{figure}

\begin{figure}[tbp]
\begin{center}
\includegraphics[width=0.8\textwidth]{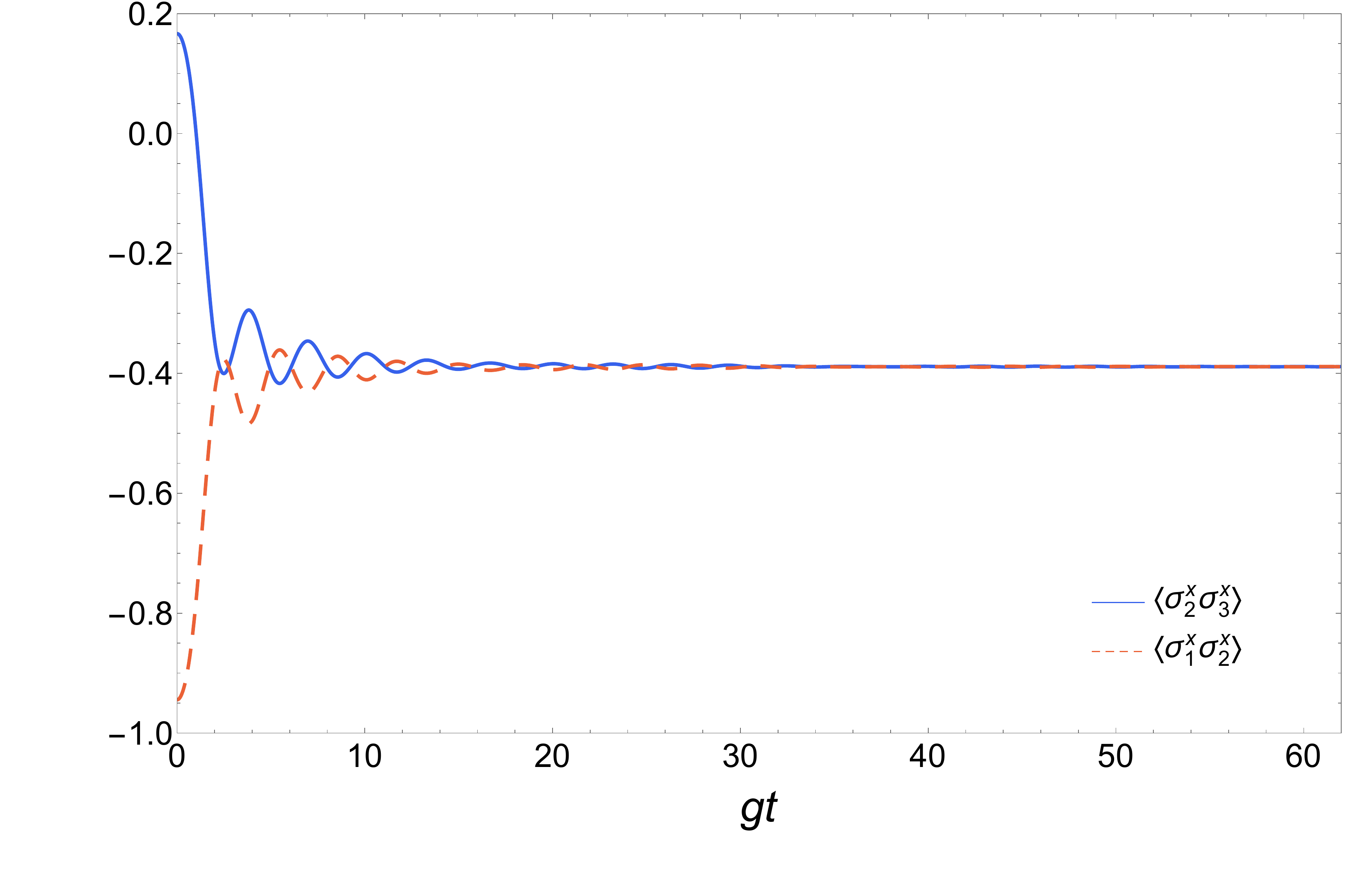}
\end{center}\caption{The time evolution of $\braket{\sigma^x_1\sigma^x_{2}}$ (red dashed) and $\braket{\sigma^x_2\sigma^x_{3}}$ (blue) after a quench from the state $\ket{\textrm{MG}}$ \eref{eq:MG} and Hamiltonian $H$ \eref{eq:XYYU} with $\gamma=2$, $h=1$ and $U=-2$. The correlators rapidly relax to the same stationary value, restoring translation invariance.  We verified relaxation up to $gt=1000$. }\label{f:relaxrest}
\end{figure}

\begin{figure}[tbp]
\begin{center}
\includegraphics[width=0.8\textwidth]{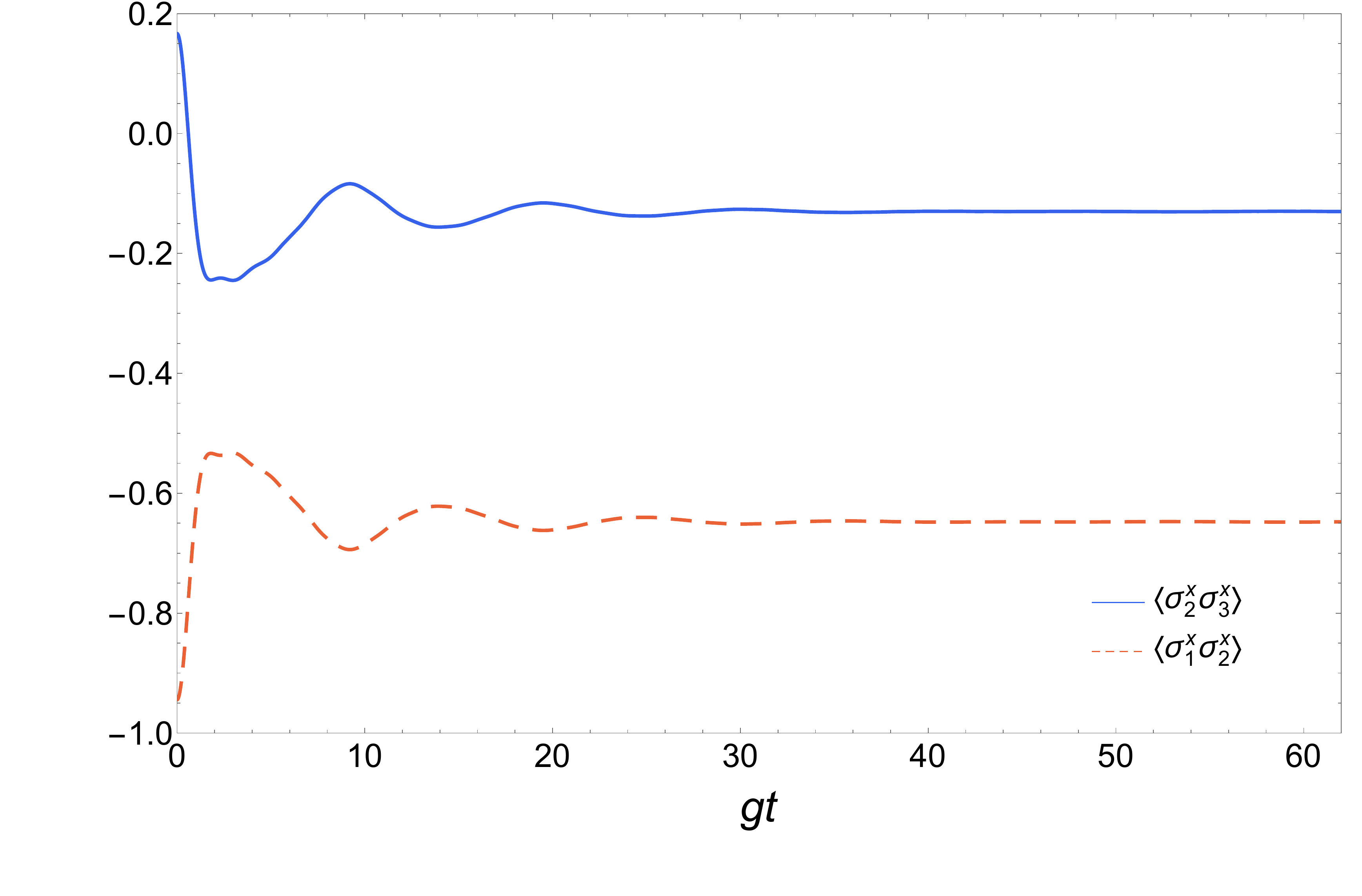}
\end{center}\caption{The time evolution of $\braket{\sigma^x_1\sigma^x_{2}}$ (red dashed) and $\braket{\sigma^x_2\sigma^x_{3}}$ (blue) after a quench from the state $\ket{\textrm{MG}}$  \eref{eq:MG} and Hamiltonian $H$ \eref{eq:XYYU} with $\gamma=2$, $h=2$ and $U=2$. The correlators rapidly relax to different stationary values. We verified relaxation up to $gt=1000$.}\label{f:relaxbrok}
\end{figure}

\begin{figure}[tbp]
\begin{center}
\includegraphics[width=0.8\textwidth]{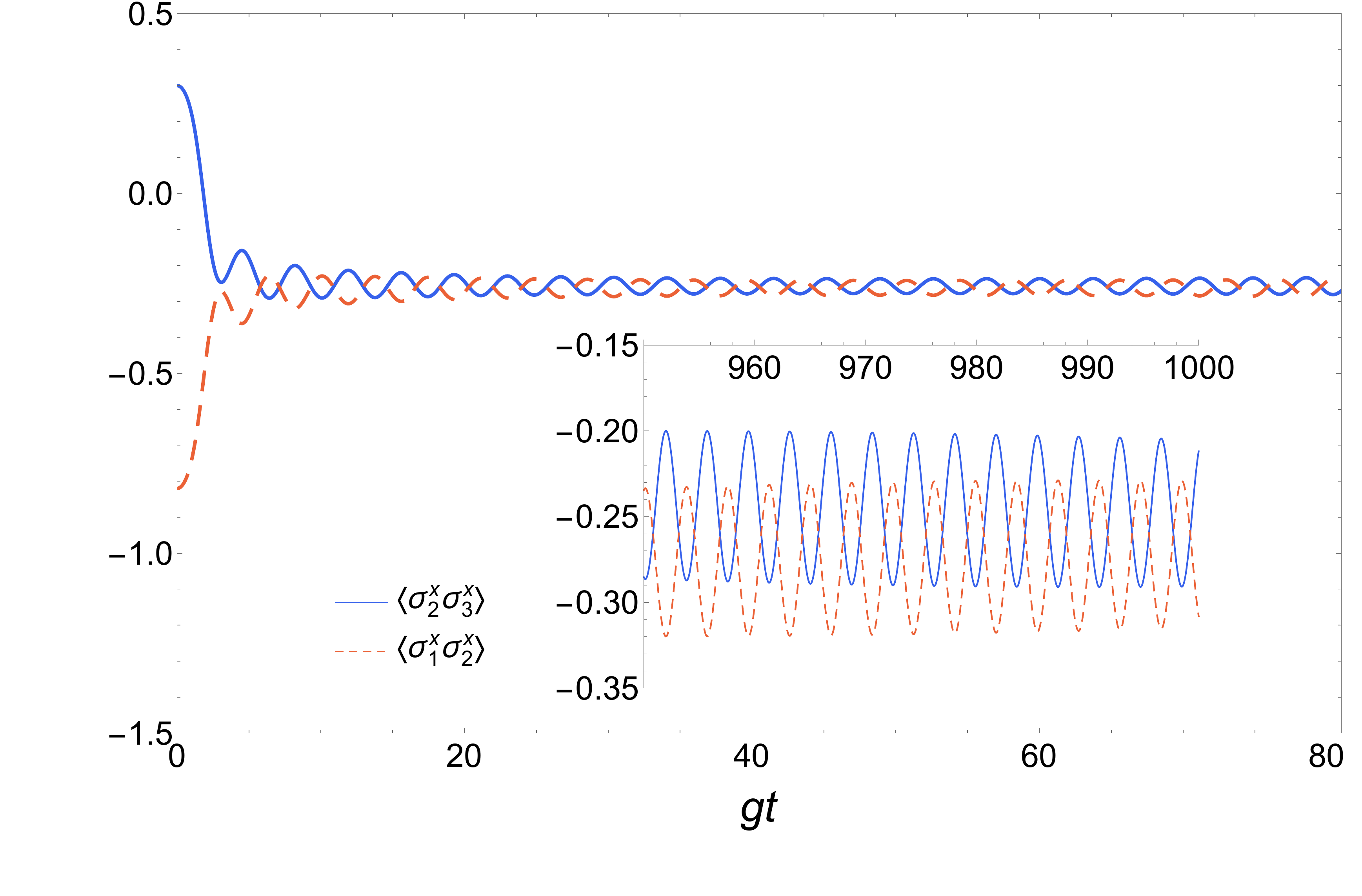}
\end{center}\caption{The time evolution of $\braket{\sigma^x_1\sigma^x_{2}}$ (red dashed) and $\braket{\sigma^x_2\sigma^x_{3}}$ (blue) after a quench from the state $\ket{\textrm{MG}}$ \eref{eq:MG} and Hamiltonian $H$ \eref{eq:XYYU} with $\gamma=4$, $h=1$ and $U=-2$. The correlators exhibit persistent oscillations on the time window explored. Inset: the amplitude of the oscillations is still unabated at $gt= 1000$.}\label{f:persistent}
\end{figure}

\subsection{Perturbations preserving integrability}\label{ss:int}
In this section we consider the case $U=\gamma=0$, in which $H$ \eref{eq:XYYU} is the Hamiltonian of the $XXZ$ spin-$\frac{1}{2}$ chain. Because of the $U(1)$ symmetry of rotations around $z$, there are many simplifications and the system of equations \eref{eq:simplesystem} can be rewritten as follows:
\bea\label{eq:simplesystemisot1}
\fl\qquad \frac{1}{2}\dot{y}_2^{[n]}(T)&=y_5^{[n]}(T) y_6^{[0]}(T)-y_6^{[n]}(T)y_5^{[0]}(T) \nn
\fl\qquad \frac{1}{2}\dot{y}_5^{[n]}(T)&=-(2y_2^{[n]}(T)+y_2^{[n-1]}(T)+y_2^{[n+1]}(T))y_6^{[0]}(T)+(h+4y_2^{[0]}(T))y_6^{[n]}(T)\nn
\fl\qquad \frac{1}{2}\dot{y}_6^{[n]}(T)&=(2y_2^{[n]}(T)+y_2^{[n-1]}(T)+y_2^{[n+1]}(T))y_5^{[0]}(T)-(h+4y_2^{[0]}(T))y_5^{[n]}(T)\, ,
\eea
where we defined
\be
\label{eq:ntransf}
\fl \quad y_2^{[n]}(T)=\int_{-\pi}^\pi\frac{\mathrm d p}{2\pi}\frac{\cos(n p)}{1+\cos p}y_2(p,T)\, ,\qquad y_{5,6}^{[n]}(T)=\int_{-\pi}^\pi\frac{\mathrm d p}{2\pi}\cos(n p)y_{5,6}(p,T)\, .
\ee
We notice that, despite the denominator, $y_2^{[n]}$ are expectation values of local operators, as well as $y_{5,6}^{[n]}$.
Since $S^z=\frac{1}{2}\sum_\ell\sigma_\ell^z$ commutes with the Hamiltonian, the 
dependence on $h$ is simple and, in particular, the expectation value of the one-site shift invariant conservation laws is independent of the magnetic field. This means that the functions $y_2^{[n]}$ are independent of $h$.

It is useful to rewrite the system for $n=0$. We find
\bea\label{eq:simplesystemisot20}
\fl\qquad\quad \dot{y}_2^{[0]}(T)&=0\nn
\fl\qquad\quad \dot{y}_5^{[0]}(T)&=2(h+2y_2^{[0]}(T)-2y_2^{[1]}(T))y_6^{[0]}(T)\nn
\fl\qquad\quad \dot{y}_6^{[0]}(T)&=-2(h+2y_2^{[0]}(T)-2y_2^{[1]}(T))y_5^{[0]}(T)\, . 
\eea
Inspecting the system we conclude that $y_2^{[0]}$ and $(y_5^{[0]})^2+(y_6^{[0]})^2$ are conserved. Moreover, the system \eref{eq:simplesystemisot20} can be directly solved, it yields
\bea\label{eq:sol0}
\fl\quad y_2^{[0]}(T)=\frac{\braket{S^z}}{2L}\equiv \frac{s^z}{2}\nn 
\fl\quad y_5^{[0]}(T)=y_5^{[0]}(0)\cos\Bigl(\int_0^T \mathrm d \tau\, (2h+4m(\tau))\Bigr)+y_6^{[0]}(0)\sin\Bigl(\int_0^T \mathrm d \tau\, (2h+4m(\tau))\Bigr)\nn
\fl\quad  y_6^{[0]}(T)=y_6^{[0]}(0)\cos\Bigl(\int_0^T \mathrm d \tau\, (2h+4m(\tau))\Bigr)-y_5^{[0]}(0)\sin\Bigl(\int_0^T \mathrm d \tau\, (2h+4m(\tau))\Bigr)\, .
\eea
Here we defined
\be\label{eq:mT}
\fl \quad m(T)\equiv \frac{s^z}{2}-y_2^{[1]}(T)=s^z-\frac{\braket{Q_2}_T}{L}=\frac{1}{4}\Braket{\sigma_\ell^z}+\frac{1}{8}\Braket{\sigma_{\ell-1}^x\sigma_{\ell}^z\sigma_{\ell+1}^x+\sigma_{\ell-1}^y\sigma_{\ell}^z\sigma_{\ell+1}^y}_T\, .
\ee
If both $y_5^{[0]}(T)$ and $y_6^{[0]}(T)$ are zero, \eref{eq:simplesystemisot1} has the solution 
\bea
y_2^{[n]}(T)=y_2^{[n]}(0)\nn
y_5^{[n]}(T)=y_5^{[n]}(0)\cos(2 (h+2s^z)T)+y_6^{[n]}(0)\sin(2 (h+2s^z) T)\nn
y_6^{[n]}(T)=y_6^{[n]}(0)\cos(2 (h+2s^z)T)-y_5^{[n]}(0)\sin(2 (h+2s^z) T)\, .
\eea
For $h\neq -2s^z$, local observables keep oscillating in time, otherwise, on the pre-relaxation timescale, the expectation values of local observables do not move from the values reached at times $1\ll Jt\ll g^{-1}$.  

More generally ($(y_5^{[0]})^2+(y_6^{[0]})^2\neq 0$), from \eref{eq:sol0} and \eref{eq:mT} we immediately infer that relaxation is possible only if 
\be\label{eq:hrelax}
\exists \lim_{T\rightarrow\infty}m(T)=-\frac{h}{2}\, ,
\ee
\be\label{eq:hrelax2}
 \exists \lim_{T\rightarrow \infty}\Bigl|\int_{0}^{T}\textrm{d}\tau\Bigl(m(\tau)+\frac{h}{2}\Bigr)\Bigr|<\infty \, .
\ee
We see that $m(T)$ could be interpreted as a sort of `induced magnetisation' that $h$ must compete with.  

The trivial dependence on $h$ is manifest choosing the variables
\bea\label{eq:simplesystemisot2}
Y_n=y_2^{[n]}(T)\nn
\Phi_n=y_5^{[n]}(T)y_5^{[0]}(T)+y_6^{[n]}(T)y_6^{[0]}(T)\, ,
\eea
which satisfy the following system of equations independent of $h$:
\bea
\label{eq:YPHI}
\fl\quad\qquad m(T)\equiv \frac{s^z}{2}-Y_1(T)\nn
\fl\quad\qquad\ddot Y_n(T)=-4\Phi_0(2Y_n(T)+Y_{n+1}(T)+Y_{n-1}(T))+8(s^z-m(T))\Phi_n(T)\nn
\fl\quad\qquad\dot \Phi_n(T)=-2(s^z-m(T))\dot Y_n(T)\, .
\eea
Here we omitted the time dependence in the conserved quantity $\Phi_0$. 
Since $\Phi_0\neq 0$ by assumption, the original variables are obtained from the inverse transformation
\bea\label{eq:invt}
y_5^{[n]}(T)=\frac{2y_5^{[0]}(T)\Phi_n(T)+y_6^{[0]}(T)\dot Y_n(T)}{2\Phi_0}\nn
y_6^{[n]}(T)=\frac{2y_6^{[0]}(T)\Phi_n(T)-y_5^{[0]}(T)\dot Y_n(T)}{2\Phi_0}\, ,
\eea
and \eref{eq:sol0}. Performing a qualitative analysis of the system \eref{eq:YPHI} we conclude that condition \eref{eq:hrelax} and $m(T)+{h}/{2}$ approaching 0 faster than $1/T$ imply relaxation of local degrees of freedom.
Therefore the variance \eref{eq:variance} of $m(T)$ is what in \Sref{s:summary} we called an `order parameter' for the transition between relaxation and oscillatory behaviour. 

As we will show in \Sref{s:Ising}, some aspects of the solutions of nonlinear systems like \eref{eq:YPHI} can be worked out analytically.  In the present context this would involve the study of quantum quenches from rather artificial initial states. Therefore we prefer to leave the entire discussion to \Sref{s:Ising}, where we will obtain a system of equations extremely similar to \eref{eq:YPHI}, with the advantage that the qualitative analysis can be carried out for more conventional initial states.

\subsection{Perturbations breaking integrability: linearisation}
In order to gain some insights into the time evolution under the Hamiltonian \eref{eq:XYYU} in the non-integrable case we focus on quantum quenches starting from the dimer product state
\be\label{eq:MG}
\ket{\textrm{MG}} =\frac{\ket{\uparrow\downarrow}-\ket{\downarrow\uparrow}}{2}\otimes\cdots\otimes\frac{\ket{\uparrow\downarrow}-\ket{\downarrow\uparrow}}{2}\,,
\ee 
which is the ground state of the Majumdar-Ghosh Hamiltonian
\be
H_0=\frac{J}{4}\sum_{\ell=1}^{L}\vec \sigma_{\ell}\cdot\vec \sigma_{\ell+1}+\frac{1}{2}\vec \sigma_{\ell}\cdot\vec \sigma_{\ell+2}\,.
\ee
Despite the model being interacting, \eref{eq:MG} is a two-site shift invariant Slater determinant, whose correlation matrix has the following symbol 
\be
\Gamma_{\rm MG}(k)=\sigma^x\otimes\sigma^y\, .
\ee
The initial conditions for $\{y_\alpha(k)\}$  \eref{eq:defy} are determined by the GGE correlation matrix for $g=0$. They can be obtained by expanding $\Gamma_{\rm MG}(k)$ in the base of the symbols  \eref{eq:C8} of the conserved charges of $H_{XY}$ (the remaining space is zeroed by the time evolution, as $1\ll Jt$, \emph{cf}. \sref{ss:t-dGGE})
\be
\Gamma(k;0)=\sum_{i=1}^8 \frac{\tr[\Gamma_{\rm MG}(k)\mathcal Q_i(k)]}{\tr[(\mathcal Q_i(k))^2]}\mathcal Q_i(k) \, .
\ee
We find
\be\label{eq:Gamma0}
\fl\quad \Gamma(k;0)=\frac{1+\cos k}{1+\cos k+\gamma^2(1-\cos k )}\mathcal Q_1(k)-\gamma \frac{1-\cos k}{1+\cos k+\gamma^2(1-\cos k )}\mathcal Q_5(k)\, .
\ee
The only nonzero initial conditions are given by (\emph{cf.} \eref{eq:defy})
\bea
y_1(k,0)=\frac{1+\cos k}{4}\, \nn
y_5(k,0)=-\gamma \frac{1-\cos k}{4}\, .\label{eq:intialconditionsMG2}
\eea
The initial state is reflection symmetric about a bond, so we can use the reduced system \eref{eq:simplesystem}.
Since  $y_5$ is the only nonzero initial condition that appears in \eref{eq:simplesystem}, for $\gamma=0$ (see previous section) the solution of the system of equations is independent of time, namely the pre-relaxation limit is trivial.

It is easy to see that also for $h=0$  system \eref{eq:simplesystem}\eref{eq:intialconditionsMG2} 
has a stationary solution. We therefore assume $\gamma, h\neq 0$. 
Since $c_j(k;y_j)$ are linear homogeneous functions of $y_j$, the magnetic field $h$ enters into the equations essentially as a scale factor. We rescale the variables as follows
\be
\label{eq:rescaling}
\tau=2 h T=2 h g t\qquad \epsilon=\frac{\gamma}{2h }\qquad z_j=\frac{2 y_j}{\gamma}\qquad \gamma_j=\frac{2 c_j}{\gamma}\, .
\ee
From \eref{eq:simplesystem} we then obtain
\bea\label{eq:linears}
\fl\qquad\qquad\partial_\tau z_2(k,\tau)&=&-\epsilon\, \varepsilon_k^2\gamma_5(k,\tau) z_6(k,\tau)+\epsilon\cos^2\frac{k}{2}\gamma_6(k,\tau) z_5(k,\tau)\nn
\fl\qquad\qquad\partial_\tau z_5(k,\tau)&=&-\epsilon\gamma_6(k,\tau) z_2(k,\tau)+(1+\epsilon\, \varepsilon_k^2\gamma_2(k,\tau))z_6(k,\tau)\nn
\fl\qquad\qquad\partial_\tau z_6(k,\tau)&=&\epsilon\gamma_5(k,\tau) z_2(k,\tau)-(1+\epsilon\, \varepsilon_k^2\gamma_2(k,\tau))\frac{\cos^2\frac{k}{2}}{ \varepsilon_k^2}z_5(k,\tau)\, ,
\eea
with the initial conditions
\be
z_2(k,0)=z_6(k,0)=0\qquad z_5(k,0)=-\sin^2\frac{k}{2}\, .
\ee
For generic $\epsilon$ the system of equations is not exactly solvable, but the limit of small $\epsilon$ allows a linear approximation. For not too large rescaled times 
(we'll come back to this point later)
the terms that are multiplied by $\epsilon$ in the last two equations can be neglected, while the functions that appear on the right hand side of the first equation can be computed at $O(\epsilon^0)$.  

For $z_5$ and $z_6$ we obtain the simple solution
\bea
\label{eq:linearisedsolution}
\fl\qquad\qquad z_5(k,\tau)\approx -\sin^2\frac{k}{2}\cos\Bigl(\frac{\cos \frac{k}{2}}{\varepsilon_k} \tau\Bigr)\nn
\fl\qquad\qquad z_6(k,\tau)\approx \sin^2\frac{k}{2}\frac{\cos \frac{k}{2}}{\varepsilon_k} \sin\Bigl(\frac{\cos \frac{k}{2}}{\varepsilon_k} \tau\Bigr)\, ,
\eea
while $z_2$ is a slightly more complicated function that involves integrals over the momentum of $z_{5,6}$, namely
\be
\label{eq:linearisedsolution2}
\fl\quad\quad z_2(k,\tau)\approx \epsilon\frac{\sin k\sin\frac{k}{2}}{2}  \int_{-\pi}^{\pi}\frac{\textrm{d}p}{2\pi}\sum_{\sigma=\pm}g_\sigma(k,p)f(\frac{\cos \frac{k}{2}}{\varepsilon_k}+\sigma\frac{\cos \frac{p}{2}}{\varepsilon_p};\tau)\,.
\ee
Here we defined
\be
\fl\quad\quad g_{\sigma}(k,p)\equiv\frac{\sin^2\frac{p}{2}}{\varepsilon_p}\Bigl[\frac{\cos^2\frac{k}{2}\cos^2\frac{p}{2}-\gamma^2\sin^2\frac{k}{2} \sin^2\frac{p}{2}}{\varepsilon_k\varepsilon_p}+\sigma\cos\frac{k}{2} \cos \frac{p}{2}\Bigr]\,,
\ee
\be
\fl\quad\quad f(x;\tau)\equiv\frac{1-\cos(x \tau)}{x}\, .
\ee
For small $\epsilon$ and given $\gamma$, $z_2(k,\tau)$ relaxes to the stationary value  
\be\label{eq:z2}
 \label{eq:stationary}
\fl\qquad\qquad z_2(k,\infty)=\frac{\epsilon\sin^2 k}{8 \gamma^2} \Bigl[1- 3\gamma^2+(2+4\gamma^2)\cos k+(1-\gamma^2)\cos2k\Bigr]\,.
\ee
Let us now estimate the time window in which the linear approximation is applicable. 
From \eref{eq:linearisedsolution} it follows that $\gamma_5$ and $\gamma_6$ decay to zero as $\tau^{-\frac{3}{2}}$. Instead, since $z_2$ approaches a nonzero stationary value, $\gamma_2$ is of the same order of $z_2$. This means that, as the time increases, the first term on the right hand side of the last two equations of \eref{eq:linears} becomes more and more negligible with respect to the other term multiplied by $\epsilon$. By neglecting the former we obtain essentially the same solution \eref{eq:linearisedsolution} as before, with the replacement   
\be
\tau\rightarrow\tau+\epsilon\ \varepsilon_k^2\int_0^\tau\mathrm d s \gamma_2(k,s)=\tau\left(1+\epsilon\ \varepsilon_k^2 \gamma_2(k,\infty)\right)+\dots\, .
\ee
Being $\gamma_2\sim O(\epsilon)$, after a rescaled time $\tau\sim \frac{1}{\epsilon^2}$, the correction to $z_{5,6}$ becomes comparable with the function itself. Assuming that the relevant part of the time evolution occurs within this time scale, the linear approximation is justified only if $|z_2|\ll 1$ (and $\epsilon\ll 1$).
For $\gamma<1/2$ we find $|z_2(k,\infty)|< \frac{\epsilon}{6\gamma^2}$, so we obtain the consistency condition
\be\label{eq:limitval}
\frac{1}{12h}\ll \gamma\ll 2h\, .
\ee
Figures \ref{f:q5} and \ref{f:q2} report a comparison between the solution of the linearised problem and the full numerical solution of system \eref{eq:simplesystem} for a set of parameters fulfilling \eref{eq:limitval}.

From the expressions \eref{eq:linearisedsolution} of $z_5(k,\tau), z_6(k,\tau)$ and \eref{eq:z2} of  $z_2(k,\tau)$, we can directly compute the time evolution of the expectation value of any local observable in the pre-relaxation limit. Indeed Corollary \ref{C:2} allows us to apply the Wick theorem at any time (in the limit under examination) and the correlation matrix is given by \eref{eq:correlmatrix}.

\begin{figure}[tbp]
\begin{center}
\includegraphics[width=0.8\textwidth]{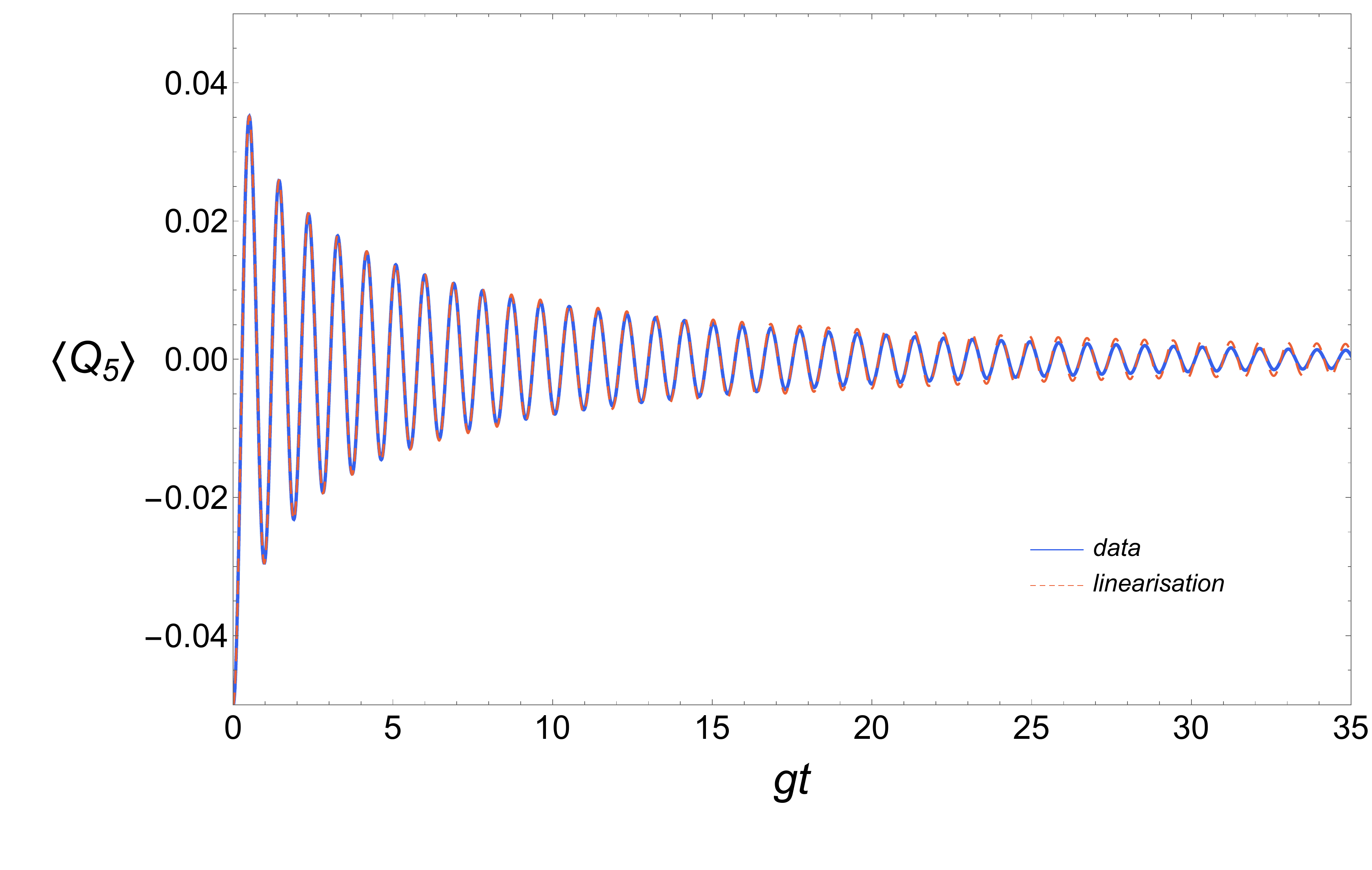}
\end{center}\caption{The time evolution of $\braket{Q_5}_t=\gamma\int\frac{{\rm d}k}{4\pi} z_5(k,t)$, where $Q_5$ is the conserved charge of $H_{XY}$ corresponding to the symbol $\mathcal Q_5(k)$ \eref{eq:C8} for a time evolution starting from the state  $\ket{MG}$ \eref{eq:MG}. The parameters of the Hamiltonian \eref{eq:XYYU} are $\gamma=0.2$, $h=3.5$ and $U=-1$, hence $\epsilon\approx0.029$ (\emph{cf.} \eref{eq:rescaling}) and  $\gamma$ fulfils the consistency condition \eref{eq:limitval} of the linearisation procedure. The analytical prediction of \eref{eq:linearisedsolution} (red dashed line) is in excellent agreement with the numerical data (blue line).}\label{f:q5}
\end{figure}

\begin{figure}[tbp]
\begin{center}
\includegraphics[width=0.8\textwidth]{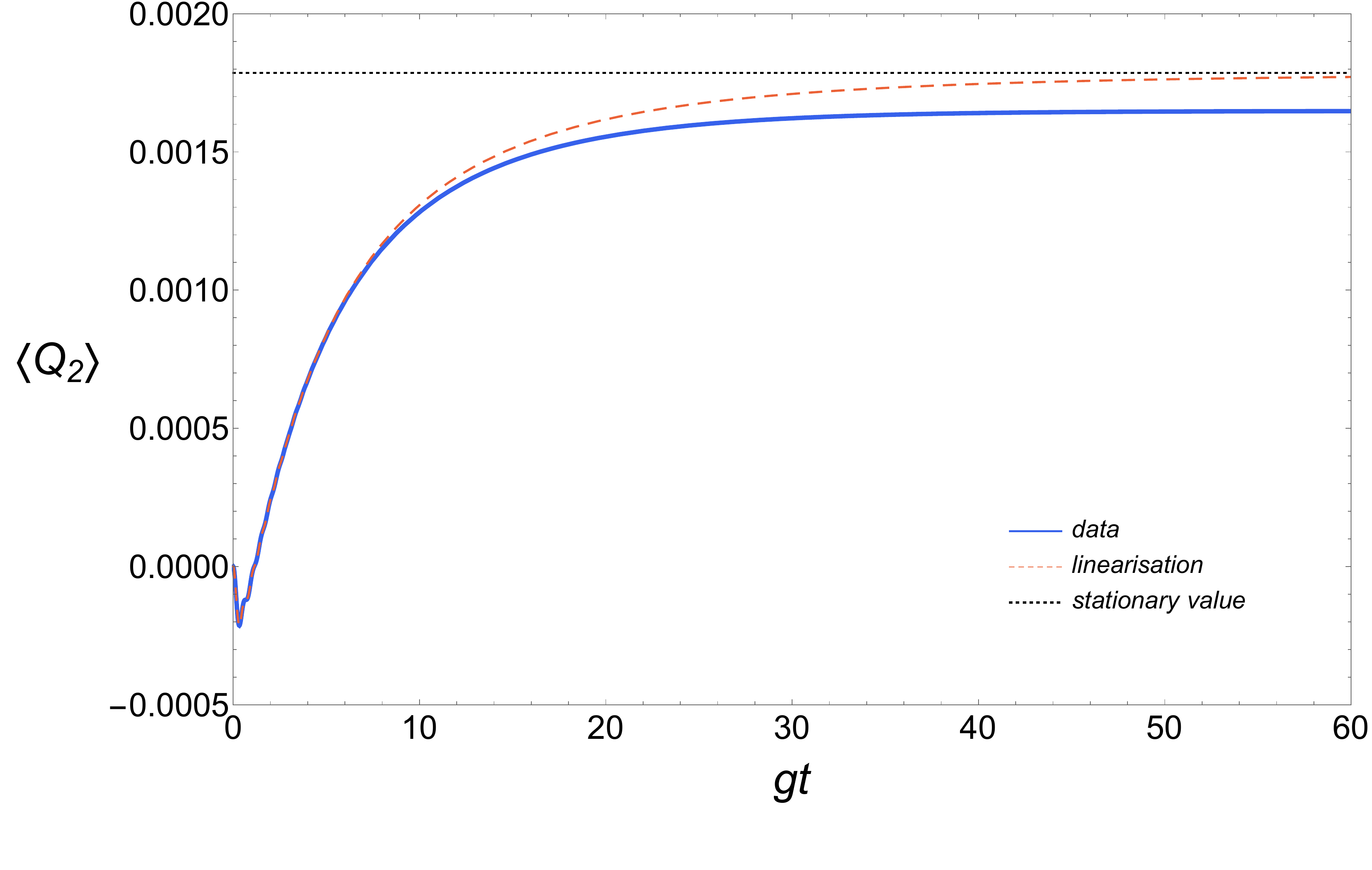}
\end{center}\caption{The time evolution of $\braket{Q_2}_t=\gamma\int\frac{{\rm d}k}{4\pi} z_2(k,t)$, where $Q_2$ is the conserved charge of $H_{XY}$ corresponding to the symbol $\mathcal Q_2(k)$ \eref{eq:C8} for a time evolution starting from the state  $\ket{MG}$ \eref{eq:MG}. The parameters of the Hamiltonian \eref{eq:XYYU} are $\gamma=0.2$, $h=3.5$ and $U=-1$, hence $\epsilon\approx0.029$ (\emph{cf.} \eref{eq:rescaling}) and  $\gamma$ fulfils the consistency condition \eref{eq:limitval} of the linearisation procedure. The analytical prediction of \eref{eq:linearisedsolution} (red dashed line) is in fairly good agreement with the numerical data (blue line). The stationary value produced by the solution of the linearised problem (black dotted line) is $\braket{Q_2} =\frac{(1- 5\gamma^2)}{128h} $ (\emph{cf.} \eref{eq:stationary}).}\label{f:q2}
\end{figure}

Any integral involving $z_5$ and $z_6$ approaches zero and $z_2$ becomes independent of time even if not integrated. Therefore, in the limit \eref{eq:limitval} and for large times the expectation value of any local observable relaxes to a stationary value that can be described by the correlation matrix with symbol
\be\label{eq:inflin}
\lim_{\tau\rightarrow\infty^{(*)}}\Gamma(k;\tau)= \frac{\gamma z_1(k,\infty)}{\varepsilon_k^2}\mathcal Q_1(k)+\frac{\gamma z_2(k,\infty)}{\varepsilon_k^2\cos^2(k/2)}\mathcal Q_2(k)\, ,
\ee
where the infinite time limit $\tau\rightarrow\infty^{(*)}$ must be understood within the limits of validity of the linear approximation.

We point out that one-site shift invariance is restored, indeed  the only contributions to the correlation matrix at infinite times arise from $\mathcal Q_1(k)$ and $\mathcal Q_2(k)$, which are symbols of one-site shift invariant operators. The manifestly one-site shift invariant expression of the correlation matrix in the limit \eref{eq:inflin} reads
\be
\fl
\qquad\qquad\lim_{\tau\rightarrow\infty^{(*)}}\Braket{\left(\begin{array}{c}
a_{n}^x \\
a_{n}^y 
\end{array}\right)
\left(\begin{array}{cccc}
a_{\ell}^x&a_{\ell}^y
\end{array}\right)}= \delta_{\ell n} \1_2+\int_{-\pi}^\pi\frac{\mathrm d k}{2\pi} e^{-i(n-\ell)k}\,\Gamma^{(1)}(k)\, ,
\ee
with
\bea
\fl \Gamma^{(1)}(k)=\frac{2h(1+\cos 2k) +\sin^2 k\cos k(1- 3\gamma^2+(2+4\gamma^2)\cos 2k+(1-\gamma^2)\cos 4k)}{2 h(1+\cos 2 k+\gamma^2(1-\cos 2k))}\nn
\qquad\qquad\qquad\qquad\qquad\qquad\times(\cos k\  \sigma^y-\gamma\sin k\ \sigma^x)\, .
\eea

\subsection{Remarks on the late time dynamics} We notice that, despite one-site shift invariance being restored in \eref{eq:inflin}, the asymptotic value is not given by the average over a shift of the expectation value of the operator in the GGE of the unperturbed model. Indeed the one-site shift average of \eref{eq:Gamma0} is proportional to $\mathcal Q_1(k)$ (\emph{cf}.  \eref{eq:transf}) but the symbol of the large time correlation matrix \eref{eq:inflin} has also a term proportional to $\mathcal Q_2(k)$.
Consequently, the shift-averaged stationary values can \emph{not} be recovered from those in the limit of small perturbation $g\rightarrow 0$. For example we have
\bea
\lim_{1\ll t\ll \frac{1}{g}}\Braket{\frac{\sigma_{2\ell-1}^z+\sigma_{2\ell}^z}{2}}=O(g)\label{eq:zerom}\\
\lim_{2 h g t\rightarrow\infty^{(*)}}\Braket{\sigma_{\ell}^z}
=-\frac{\gamma^2}{16 h}\frac{3+\gamma}{(1+\gamma)^3}+O(g)\label{eq:nonzerom}\, ,
\eea
where we highlighted that there are $O(g)$ corrections. Besides this particular quench in a non-integrable model, similar issues arise also in the integrable case, where it is generally believed that at infinite time after the quench the expectation values can be computed in a GGE constructed with the (quasi-)local conservation laws of the model. 
In the rest of the section we propose a reasoning that relates this kind of deviations to possible atypical properties of the model. 

We notice that at times $1\ll t\ll\frac{1}{g}$, shift-symmetrised expectation values of quadratic operators, \emph{e.g.}  \eref{eq:zerom}, can be obtained from the GGE of the unperturbed model constructed with only the local translation invariant conservation laws. This is because the limit $1\ll t\ll\frac{1}{g}$ with $g\rightarrow 0$ can be described by the GGE of the unperturbed model, but 
the conservation laws that break translation invariance are odd under a shift by one site, so they can in fact be neglected (this equivalence breaks down for operators that consist of the product of more than two Jordan-Wigner fermions). 

For nonzero $g$ non-abelian integrability is supposed to break down and the relevant charges are generally assumed to be one-site shift invariant and in involution with one another. 
We now speculate about the stationary state in the limit of small $g$ if the perturbation does not break integrability. 
Let the (quasi-)local conservation laws of the interacting (integrable) model be in a smooth one-to-one correspondence with the local one-site shift invariant conservation laws for $g=0$. In the limit of small $g$, the stationary state should locally approach the GGE constructed with the local one-site shift invariant conservation laws of the unperturbed model.  
The expectation values of shift-symmetrised quadratic operators  at times $1\ll t\ll\frac{1}{g}$ are compatible with such a one-site shift invariant GGE. However, discrepancies like that between \eref{eq:zerom} and  \eref{eq:nonzerom} show that at larger times there is a time window in which the expectation values approach a different value. 
Our assumption of regularity of the conservation laws as a function of $g$ requires that at even larger times the expectation values should eventually relax to the same values they had in the earliest plateau.
This is clearly possible, but an infinite number of operators displaying a similar behaviour is rather surprising. This makes us wonder whether the hypothesis of regularity could break down, that is to say there are (quasi-)local conservation laws for $g\neq 0$ that become nonlocal when $g=0$ (\emph{e.g.} their typical range could be singular as $g\rightarrow 0$) or, \emph{vice versa}, some one-site shift invariant conservation laws of the unperturbed Hamiltonian do not have analogues at nonzero $g$.

From this point of view, discrepancies like that between \eref{eq:zerom} and \eref{eq:nonzerom}  could be indications that in the XYZ model there might be quasi-local conservation laws that do not behave well in the limit $g\rightarrow 0$.
This scenario becomes even more plausible if one takes into account the issues in the construction of the GGE in XXZ spin-$\frac{1}{2}$ chains that were recently pointed out \cite{QAXXZ-14,PMWKZT-14,A:bound}.

\subsection{Summary}
In this section we considered the pre-relaxation limit in an XY spin-$\frac{1}{2}$ chain perturbed by interacting operators. 

As an example of pre-relaxation in integrable models we investigated the XXZ spin-$\frac{1}{2}$ chain. The model has $U(1)$ symmetry, manifested by the conservation of the spin in the $z$-direction. Consequently, the external magnetic field (along $z$) generally produces oscillatory behaviour in local observables. In fact we showed that there is a specific (generally nonzero) value of the magnetic field for which the time evolution in the pre-relaxation limit may end up in a second plateau.

For non-integrable perturbations we exhibited examples of the typical time evolution of the expectation values of local observables in the pre-relaxation limit. We also described a linearisation scheme that allowed us to predict the time evolution of the dimer product state \eref{eq:MG} when the Hamiltonian parameters satisfy particular conditions. In that limit we found local relaxation to a one-site shift invariant state. In order to characterise the crossover between persistent oscillatory behaviour and relaxation, one should go beyond that linearisation scheme.

In the next section we will consider the model \eref{eq:TFICNI}, which has several dynamical aspects in common with the Hamiltonian \eref{eq:Hrel}, especially in the integrable case $\gamma=U=0$ considered in \Sref{ss:int}.
Specifically, we will present a method that resembles the linearisation considered in this section but that allows us to extract some information about the `quench dephasing diagram' of the model.

\section{An exactly solvable model }\label{s:Ising} %

The results of Section \ref{s:int} are a compelling motivation for the study of nonlocal Hamiltonians of the form \eref{eq:opform}. 
We now go beyond that rigid derivation: we skip the formal steps that relate \eref{eq:opform} to a pre-relaxation limit and start off directly with a Hamiltonian of the form \eref{eq:opform}.
We then query whether such models with (apparently) non-integrable long-range interactions could display thermal-like behaviour at late times after a quench. Specifically, we consider the Hamiltonian 
\be\label{eq:TFICNI}
H(\tilde g,\lambda)=-\sum_\ell^L(\sigma_\ell^x\sigma_{\ell+1}^x+\tilde g\sigma_\ell^z)+\frac{\lambda}{L}\Bigl(\sum_\ell^L\sigma_\ell^z\Bigr)^2\, .
\ee
This  has been recently proposed as a convenient model to investigate pre-thermalisation issues~\cite{IsingNI}. In fact, in order to recover some temporal cluster decomposition properties, the authors of \cite{IsingNI} considered a slightly modified version of \eref{eq:TFICNI}, where, in the term proportional to $\lambda$, $\sum_\ell^L\sigma_\ell^z$ was replaced by $\sum_\ell^L\sigma_\ell^z-\overline{\sum_\ell^L\sigma_\ell^z}$, the latter being its time average for $\lambda=0$.  The time average is a simple quadratic quasi-local operator (\emph{cf.} the first equation of \eref{eq:HQ8}), so our formalism could be readily applied. However, for the sake of simplicity,
we consider \eref{eq:opform} and show later that a redefinition of $\tilde g$ is sufficient to recover the results shown in \cite{IsingNI}.

From Corollary \ref{C:2}, in the thermodynamic limit $L\rightarrow\infty$ the time evolution under \eref{eq:TFICNI} is locally equivalent to the time evolution under the time-dependent mean-field Hamiltonian
\be\label{eq:HMFIS}
H_{\rm MF}^{\Psi_0}(t)=-\sum_\ell(\sigma_\ell^x\sigma_{\ell+1}^x+h(t)\sigma_\ell^z)\, ,
\ee
with $h(t)$ the solution of the self-consistent equation
\be\label{eq:hm}
\fl\qquad h(t)=\tilde g-2\lambda \braket{\Psi_0|{\rm T^\dag}\exp\Bigl(i\int_0^t\mathrm d\tau H_{\rm MF}^{\Psi_0}(\tau)\Bigr)\sigma_\ell^z\ {\rm T}\exp\Bigl(-i\int_0^t\mathrm d\tau H_{\rm MF}^{\Psi_0}(\tau\Bigr)|\Psi_0}\, .
\ee
Here we used translation invariance to replace $\frac{1}{L}\sum_\ell\sigma_\ell^z$ by the local operator $\sigma_\ell^z$ (which removes the nasty dependence on the chain length $L$). 
For Slater determinant initial states the expectation value can be conveniently written in terms of the symbol $\Gamma(k)$ of the initial correlation matrix (see Appendix~\ref{a:free}). In particular we find
\be\label{eq:ht}
h(t)=\tilde g-\lambda\int_{-\pi}^\pi\frac{\mathrm d k}{2\pi} y_k(t)\, ,
\ee
with
\be
y_k(t)=\tr\Bigl[U_{{\mathcal H}_{\rm MF}}(k,t)\Gamma(k)U^\dag_{{\mathcal H}_{\rm MF}}(k,t) \sigma^y\Bigr]\, ,
\ee
\be
U_{{\mathcal H}_{\rm MF}}(k,t)={\rm T}\exp\Bigl[2i\int_0^t\mathrm d\tau \sigma^y(h(\tau)-e^{i k \sigma^z})\Bigr]\,.
\ee
One can easily verify that $y_k$ is the solution of 
\be\label{eq:system}
\begin{array}{l}
y_k^{''}(t)=4(h(t)-\cos k)\phi_k(t)+16(h(t)\cos k  -1)y_k(t)\\
\phi_k^{'}(t)=-4h(t)y_k^{'}(t)\, ,
\end{array}
\ee
with
\be
\phi_k(t)=-4\tr\Bigl[U_{{\mathcal H}_{\rm MF}}(k,t)\Gamma(k)U^\dag_{{\mathcal H}_{\rm MF}}(k,t) \sigma^ye^{i k \sigma^z}\Bigr]\, .
\ee

Equations \eref{eq:ht} and \eref{eq:system} could be solved by discretising the momenta $k$; however working in real space is more transparent.  
We therefore introduce the (real space) Fourier coefficients 
\be\label{eq:Fcoeff}
\tilde y_n=\int_{-\pi}^\pi\frac{\mathrm d k}{2\pi}\cos(n k)y_k(t)\qquad \tilde\phi_n=\int_{-\pi}^\pi\frac{\mathrm d k}{2\pi}\cos(n k)\phi_k(t)\, ,
\ee
which are the expectation values of local operators with range $n$ and $n+1$ respectively.  
The system of equations \eref{eq:system} can then be recast as follows
\be\label{eq:systemF}
\fl\qquad\qquad\begin{array}{ll}
h=g-\lambda \tilde y_0\\
\tilde y_0''=4 h \tilde\phi_0-4\tilde\phi_{1}+16 h\tilde y_{1}-16\tilde y_{0}\\
\tilde y_n^{''}=4 h \tilde\phi_n-2(\tilde\phi_{n+1}+\tilde\phi_{n-1})+8 h (\tilde y_{n+1}+\tilde y_{n-1})-16\tilde y_{n}& (n>0)\\
\tilde \phi_n^{'}=-4h \tilde y_n'&  (n\geq 0)\, . \\
\end{array}
\ee
The similarity with the system of equations \eref{eq:YPHI} for the pre-relaxation  limit in XXZ spin chains is remarkable, although the meaning of the variables is  different.

For sufficiently large $n$ (in order to be outside of the (deformed) light cone), both $\tilde y_n$ and $\tilde\phi_n$ are small (in the noncritical case they are exponentially small in $n$). Thus, the error originated from truncating the system of equations to $n\leq N$ decays very fast to zero with $N$ and \eref{eq:systemF} can be conveniently reduced to a finite system of differential equations. This is the regularisation that we used in our numerical investigations\footnote{From a numerical point of view, this allows us to avoid integrating $y_k$ at each time step of the Runge-Kutta algorithm used for the resolution of \eref{eq:systemF}, at the cost of complicating the initial conditions. This is convenient because generally the number of time steps is much larger than $N$.}. 

The system of equations has at least one integral of motion, namely the energy per unit length $\varepsilon=\braket{\Psi_0|H|\Psi_0}/L$. This can be written as follows
\be\label{eq:energy}
\varepsilon=\frac{h^2-\tilde g^2}{4\lambda}-\frac{1}{8}\tilde\phi_0\, .
\ee
In addition, \eref{eq:TFICNI} has infinite noninteracting conservation laws that are odd under reflection symmetry \cite{F:pair,FE:RDM}. 
The Dzyaloshinskii-Moriya interaction
\be
H_{\rm D-M}=\sum_{\ell}\sigma_\ell^x\sigma_{\ell+1}^y-\sigma_\ell^y\sigma_{\ell+1}^x
\ee
is one of them. 
For generic initial states, this is sufficient to rule out thermalisation. 
However, we embrace the point of view of \cite{IsingNI} and wonder whether at infinite time after a quench from a reflection symmetric state some form of thermalisation arises. 

\subsection{To relax or not to relax}\label{ss:RONR} %

\begin{figure}[tbp]
\begin{center}
\includegraphics[width=0.8\textwidth]{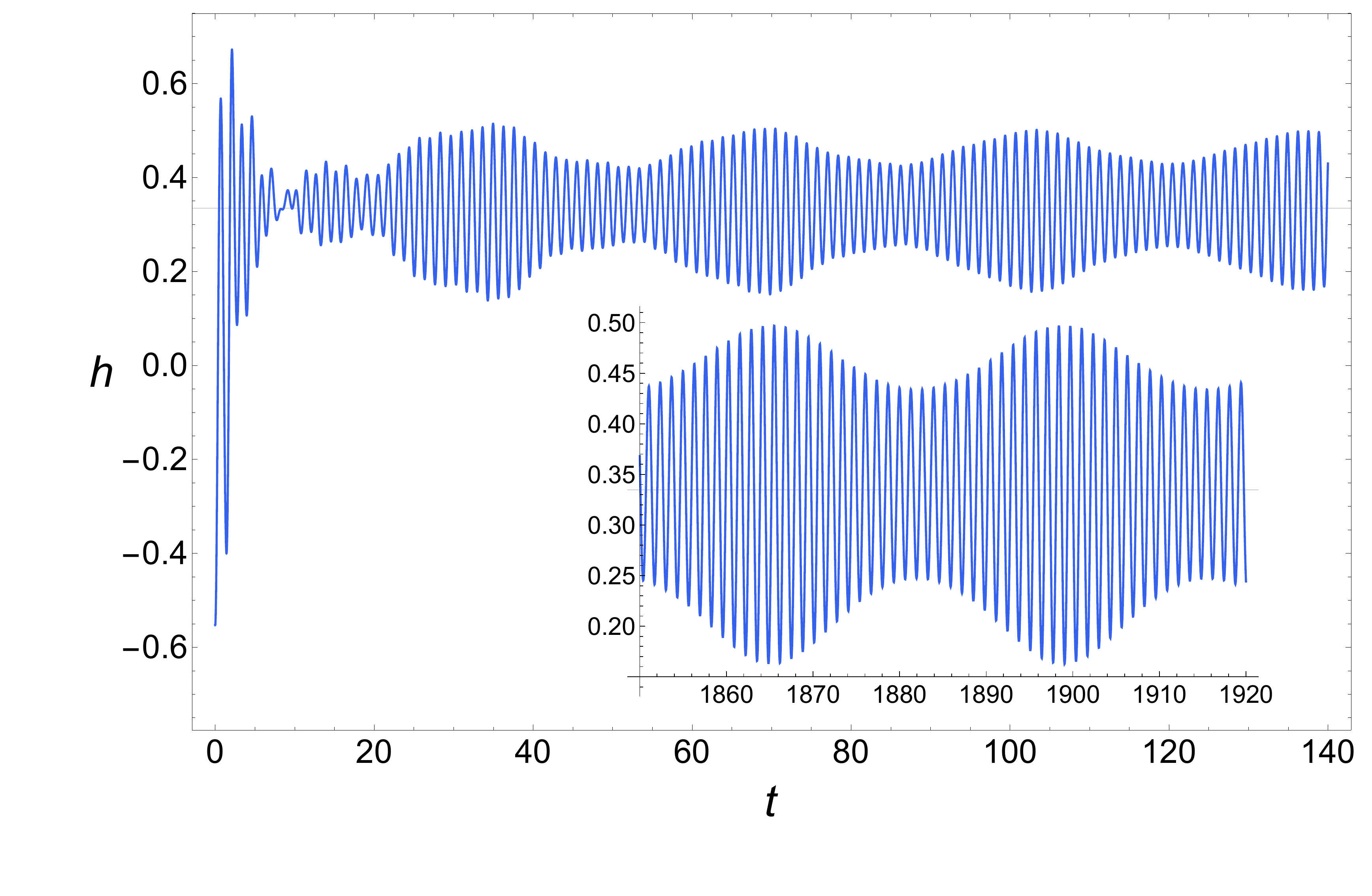}
\end{center}\caption{The time-dependent magnetic field $h(t)$ of the mean-field Hamiltonian \eref{eq:HMFIS} (essentially, the magnetisation along $z$, \emph{cf}. \eref{eq:hm}) after a quench from  the ground state of the TFIC \eref{eq:TFIC}, with  $g_0=1.5$, and Hamiltonian $H(\tilde g,\lambda)$ \eref{eq:TFICNI}, with $\tilde g=0.5$ and $\lambda=0.6$. 
The system does not seem to relax, indeed also at very large times (inset) there are (rather regular) persistent oscillations.
}\label{f:hosc}
\end{figure}

At late times the mean-field Hamiltonian \eref{eq:HMFIS} can result in two distinct behaviours: either the dynamics is governed by a (time independent) TFIC Hamiltonian (viz.  $\exists \lim_{t\rightarrow\infty} h(t)$), or there is no relaxation (viz. $\cancel \exists \lim_{t\rightarrow\infty} h(t)$).  

In the following we will focus on quantum quenches from the ground state of the TFIC Hamiltonian
\be\label{eq:TFIC}
H_0=-\sum_\ell(\sigma_\ell^x\sigma_{\ell+1}^x+g_0\sigma_\ell^z)\, .
\ee

A numerical analysis suggests that for generic values of $g_0$, $\tilde g$ and $\lambda$, the system of equation \eref{eq:system} does not always describe a relaxation process. We indeed found cases in which $h(t)$ oscillates with almost constant amplitude (see Figure \ref{f:hosc}).    

If there are (finite) regions of the parameter space associated with relaxation and regions that are instead characterised by persistent oscillations, some quantities shall behave non-analytically at the boundaries of the regions. 
For example, the \emph{relaxation parameter}
\be\label{eq:DeltahT}
\Delta h_T=\sqrt{\frac{1}{T}\int_{T}^{2T}\mathrm d \tau h^2(\tau)-\Bigl(\frac{1}{T}\int_{T}^{2T}\mathrm d \tau h(\tau)\Bigr)^2}
\ee
approaches zero as $T\rightarrow\infty $ in the regions of (local) relaxation.  On the other hand, if there are (sufficiently regular) persistent oscillations, $\Delta h_T$ remains nonzero for arbitrarily large $T$. 

\begin{figure}[tbp]
\begin{center}
\includegraphics[width=0.8\textwidth]{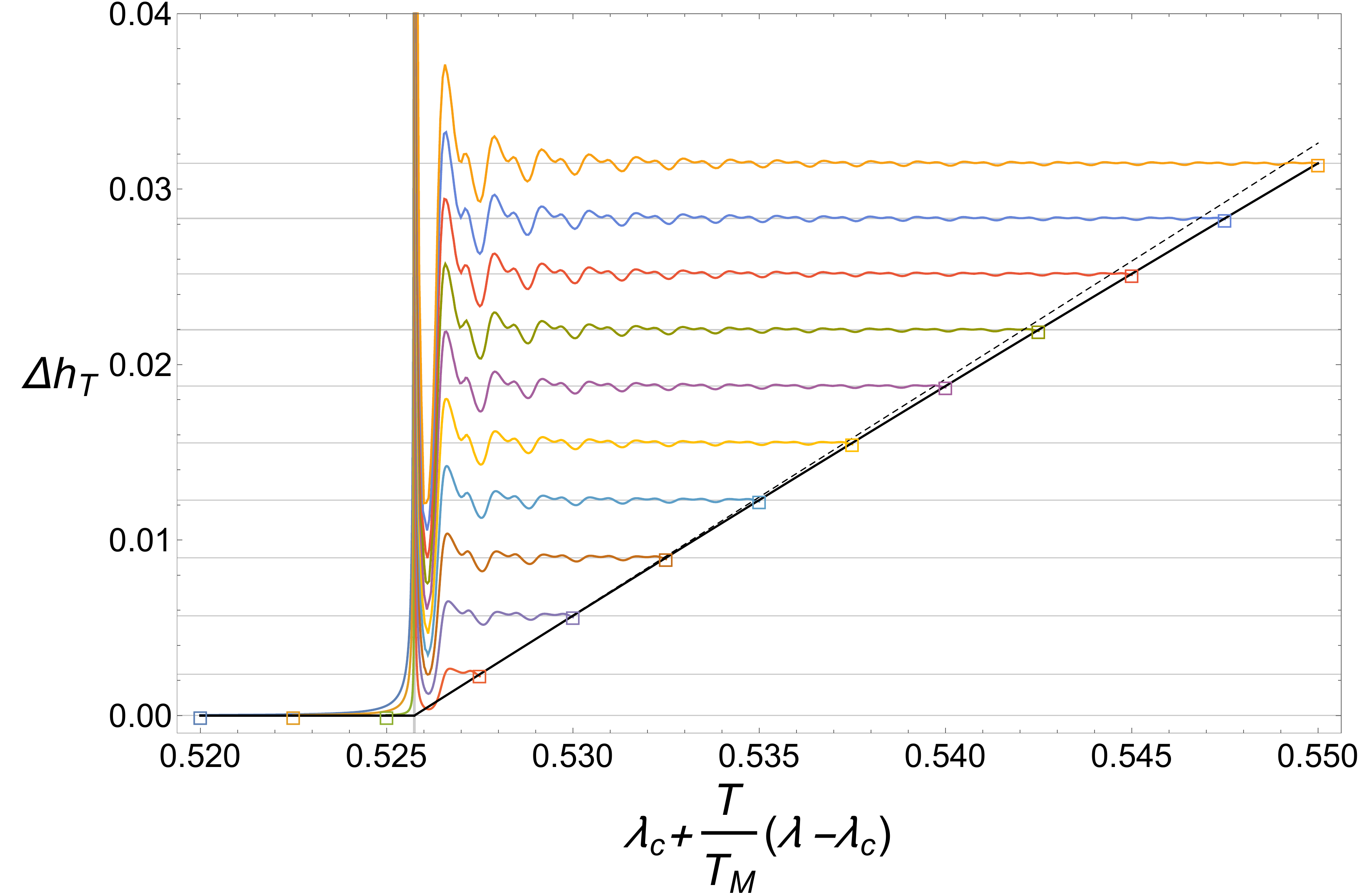}
\end{center}\caption{The relaxation parameter $\Delta h_T$ \eref{eq:DeltahT} for quenches from the ground state of the TFIC \eref{eq:TFIC}, with $g_0=1.5$, and Hamiltonian $H(\tilde g,\lambda)$ \eref{eq:TFICNI}, with $\tilde g=0.5$, as a function of the rescaled time $\lambda_c+\frac{T}{T_M}(\lambda-\lambda_c)$ for various values of $\lambda$. The maximal time considered for the time average is $T_M=1920$; at $T_M$ the abscissa is exactly equal to $\lambda$, which can therefore be identified with the abscissa of the open squares. The critical value $\lambda_c\approx 0.5257$ (grey vertical line, where all curves collapse) has been estimated by a parabolic fit (black solid line) of $\Delta h_T$ (open squares) for large $T$ as a function of $\lambda$. The dashed line is the linear term of the fit.}\label{f:relax}
\end{figure}

We  numerically analysed the region of the parameter space in which there is no relaxation (at least apparently). This generally happens for sufficiently large $\lambda$ (see Figures \ref{f:diagram} and \ref{f:hosc}). In the vicinity of (a numerical estimation of) $\lambda_c$, $\Delta h_T$ is nicely fitted by a line (\emph{cf.} Figure \ref{f:relax})
\be
\Delta h_T= \alpha (\lambda-\lambda_c)+O((\lambda-\lambda_c)^2) \qquad\lambda>\lambda_c\, .
\ee
Since $\Delta h_T$ is positive, linear behaviour is indicative of discontinuous derivative at $\lambda=\lambda_c$. Indeed, for $\lambda<\lambda_c$, $\Delta h_T$ is compatible with zero (see \emph{e.g.} the three solutions with $\lambda<\lambda_c$ in Figure \ref{f:relax}).  

\subsubsection{Small quench.}  %
In the case $g_0= h(0)$ the solution of \eref{eq:systemF} is independent of time (as a consequence of Corollary \ref{C:3}). This can be turned into a condition on the parameter $\tilde g$ of the Hamiltonian  as follows 
\be \label{eq:noquench1}
\tilde g=\bar g(g_0,\lambda)\, ,
\ee
where we defined
\be\label{eq:noquench}
\bar g(g_0,\lambda)=g_0+2\lambda\int_{-\pi}^\pi\frac{\mathrm d k}{2\pi}\frac{g_0-\cos k}{\sqrt{1+g_0^2-2 g_0\cos k}}\, .
\ee  
The initial state $\ket{\psi_{g_0}}$  corresponding to this quench is then an effective eigenstate of $H(\tilde g,\lambda)$ \eref{eq:TFICNI}.
In addition, in Appendix \ref{a:GS} we show that one of the solutions of \eref{eq:noquench1} corresponds to the state  $\ket{\psi_{g_0}}$ with minimal energy among Slater determinants. As a matter of fact, the numerical analysis indicates that $\ket{\psi_{g_0}}$ is the true ground state of \eref{eq:TFICNI}.  We can therefore use $\ket{\psi_{g_0}}$ as a reference state to define the limit of \emph{small quench}.

In this limit, the transition relaxation/no-relaxation can be understood more clearly. Indeed, choosing the parameters such that \eref{eq:noquench1} is approximately satisfied,   both $|\tilde y_n(t)-\tilde y_n(0)| $ and $|\tilde \phi_n(t)-\tilde \phi_n(0)| $ turn out to be small at any time. We can therefore linearise the system of equations \eref{eq:systemF}  isolating a time independent contribution from $\tilde y_n$ and $\tilde\phi_n$:
\be
\begin{array}{l}
h=g_0+(\tilde g- \bar g(g_0,\lambda)- \lambda\delta y_0)\\
\tilde y_n=\bar y_n+\delta y_n\\
\tilde \phi_n=\bar \phi_n+\delta \phi_n\, ,\\
\end{array}
\ee
where variables with a bar on top are expectation values calculated in $\ket{\psi_{g_0}}$. 
We then obtain
\be\label{eq:linsys}
\fl\qquad\qquad \delta y_n''\approx -16(1+g_0^2)\delta y_n+16 g_0(\delta  y_{n+1}+\delta  y_{n-1})-\lambda  \bar v_n\delta y_0+(\tilde g-\bar g)  \bar v_n\, ,
\ee
where 
\be
\bar v_n=4(\bar\phi_n+2\bar y_{n-1}+2\bar y_{n+1})\, .
\ee
The system of equations \eref{eq:linsys} can be readily solved. 
Since for quenches from the ground states of TFIC Hamiltonians $\tilde y_n'(0)=0$, we find
\be\label{eq:sollin}
\delta y_n(t)\approx (\tilde g-\bar g)\sum_j k_j \Bigl[1-\cos(\sqrt{a_j}t)\Bigr] w_{n, j}
\ee
where $a_j$ and $w_{n, j}$ are the eigenvalues and the components of the (right) eigenvectors (at fixed $j$) of 
\be\label{eq:A}
A_{\ell n}=16(1+g_0^2)\delta_{\ell n}-16 g_0\delta_{|\ell- n|,1}-16g_0\delta_{\ell 0}\delta_{n 1}
+\lambda \bar v_\ell \delta_{n 0}
\ee
and $k_j$ are given by
\be\label{eq:kappa}
\vec k=(AW)^{-1}\vec{\bar v}\qquad ([\vec k]_j=k_j\, ,\quad [W]_{nj}=w_{n,j}\, ,\quad [\vec{\bar v}]_j= \bar v_j)\, .
\ee
For $\lambda=0$, $A$ can be diagonalised in momentum space; 
the eigenvalues are given by
\be\label{eq:contsp}
a_j=16(1+g_0^2-2 g_0\cos k_j)\, ,
\ee
and, in the limit $N\rightarrow\infty$ ($N$ is our regularisation parameter), the spectrum becomes continuous (the density of $k_j$ is uniform in $(0,\pi)$) and the eigenvectors unbounded. It is not difficult to show that the rank-1 perturbation $\lambda \bar v_\ell \delta_{n 0}$ does not change the continuous part of the spectrum, which is still described by \eref{eq:contsp} and, as $N\rightarrow\infty$, the density of $k_j$ remains uniform. 

For sufficiently small $\lambda$ the spectrum is continuous. Given that $k_j w_{n,j}$ is a smooth function of $a_j$ for a fix $n$ (as we numerically checked),  we can apply the Riemann-Lebesgue lemma to extract the large time behaviour of \eref{eq:sollin}. We find that it relaxes to 
\be
\tilde y_n(\infty)\approx \bar y_n+(\tilde g-\bar g)\sum_j [A^{-1}]_{n j}\bar v_j\,.
\ee
The power-law corrections to this result can be obtained by stationary phase approximations.  

On the other hand, when $\pm \lambda$ exceeds some critical value $\pm \lambda_c^{\pm}$, $A$'s spectrum develops an isolated eigenvalue and the corresponding eigenvector is bounded, whereas the continuous part of the spectrum is still described by \eref{eq:contsp}. 
The main consequence is that the oscillations associated with the isolated eigenvalue do not cancel by dephasing mechanisms and local degrees of freedom keep oscillating at arbitrarily large times. 

\subsubsection{The bound state.} %

The isolated eigenvalue $a_0$ can be easily worked out in momentum space
\be\label{eq:boundw}
\fl\qquad\sum_nA_{\ell n}w_{n,0}=a_0 w_{\ell,0}\Rightarrow [a_0-16(1+g_0^2-2g_0\cos k)]w(k;0)=\lambda \bar v(k)w_{0,0}\, ,
\ee
with $w(k;j)=w_{0,j}+2\sum_{n>0}\cos(n k)w_{n,j}$ and $\bar v(k)=\bar v_0+2\sum_{n>0}\cos(n k)\bar v_n$. 
Isolated eigenvalues are such that \mbox{$\omega=\sqrt{a_0}/4$} is outside of the continuous band, given by the image of $\sqrt{1+g_0^2-2g_0\cos k}$. We are therefore allowed to write
\be
\frac{w(k;0)}{w_{0,0}}=\frac{\lambda}{16} \frac{\bar v(k)}{\omega^2-1-g_0^2+2g_0\cos k}\, ,
\ee
which, integrated over $k$, gives the condition
\be\label{eq:osc}
1=\frac{\lambda}{16}\int_{-\pi}^\pi\frac{\mathrm d k}{2\pi}\frac{\bar v(k)}{\omega^2-1-g_0^2+2g_0\cos k}\, .
\ee
The critical values $\lambda_c^{\pm}$ at the boundaries of the relaxation region correspond to the extrema of the continuous band where $\omega=1+\mathrm{sgn}(\lambda)|g_0|$. Using the explicit form of $\bar v_n$ we then find
\be\label{eq:lambdac}
\lambda_c^{\pm}=\left(\int_{-\pi}^\pi\frac{\mathrm d k}{2\pi}\frac{\sin^2 k}{(\mathrm{sgn}(\lambda)|g_0|+g_0\cos k)\sqrt{1+g_0^2-2g_0\cos k}}\right)^{-1}
\ee
and for $\lambda>\lambda_c^+$ or $\lambda<\lambda_c^-$ equation \eref{eq:osc} can be rewritten as
\be\label{eq:oscifreq}
\fl\qquad\qquad\int_{-\pi}^\pi\frac{\mathrm d k}{2\pi}\frac{\sin^2 k}{(\omega^2-1-g_0^2+2g_0\cos k)\sqrt{1+g_0^2-2g_0\cos k}}=\frac{1}{2\lambda}\, .
\ee
Taking the derivative of this expression with respect to $\omega$ we obtain 
\be
\frac{\rm{d}\lambda}{\rm{d}\omega}>0\qquad0\leq\omega< |1-|g_0||\, \lor\, \omega> 1+|g_0| \, ;
\ee
this fact, together with the observation that $\lambda>0$ if  $\omega> 1+|g_0|$ and $\lambda<0$ if  $0\leq\omega<|1-|g_0||$, implies that $\lambda$ is an increasing function of $\omega$  (in the allowed dominion), therefore it is injective. This means that \emph{there can not be more than one isolated eigenvalue} for a fixed $\lambda$.  

The right eigenvector corresponding to the isolated eigenvalue is given by
\be
\fl\qquad\qquad w_{n,0}\propto \int_{-\pi}^\pi\frac{\mathrm d k}{2\pi}\frac{\sin^2 k\cos(n k)}{(\omega^2-1-g_0^2+2g_0\cos k)\sqrt{1+g_0^2-2g_0\cos k}}\, ,
\ee
which decays exponentially with $n$, confirming that it is a bound state. 
Analogously, the left eigenvector $w^L_{0,n}$ reads as (the factor in front of the integral is due to $A$'s asymmetry)
\bea
w^L_{0,n}\propto (2-\delta_{n 0})\int_{-\pi}^\pi\frac{\mathrm d k}{2\pi}\frac{\cos(n k)}{\omega^2-1-g_0^2+2g_0\cos k}\nn
\qquad\qquad\qquad=\frac{(2-\delta_{n 0})(-1)^n e^{-n\theta}}{\sqrt{(1+g_0^2-\omega^2)^2-4 g_0^2}}\qquad (\lambda>0)\, ,
\eea
where
\be\label{eq:theta}
\theta=\mathrm{arccosh}\Bigl(\frac{\omega^2-1-g_0^2}{2g_0}\Bigr)\, .
\ee
Using \eref{eq:sollin} and \eref{eq:kappa} we can compute the entire contribution of the isolated eigenvalue to the solution of \eref{eq:linsys}:
\be\label{eq:iseig}
(\tilde g-\bar g)\frac{\vec {\bar v}\cdot \vec w_0^L\  \vec w_0}{\vec w_0^L\cdot \vec w_0}\frac{1-\cos(4\omega t)}{16\omega^2}\, .
\ee

\begin{figure}[tbp]
\begin{center}
\includegraphics[width=0.8\textwidth]{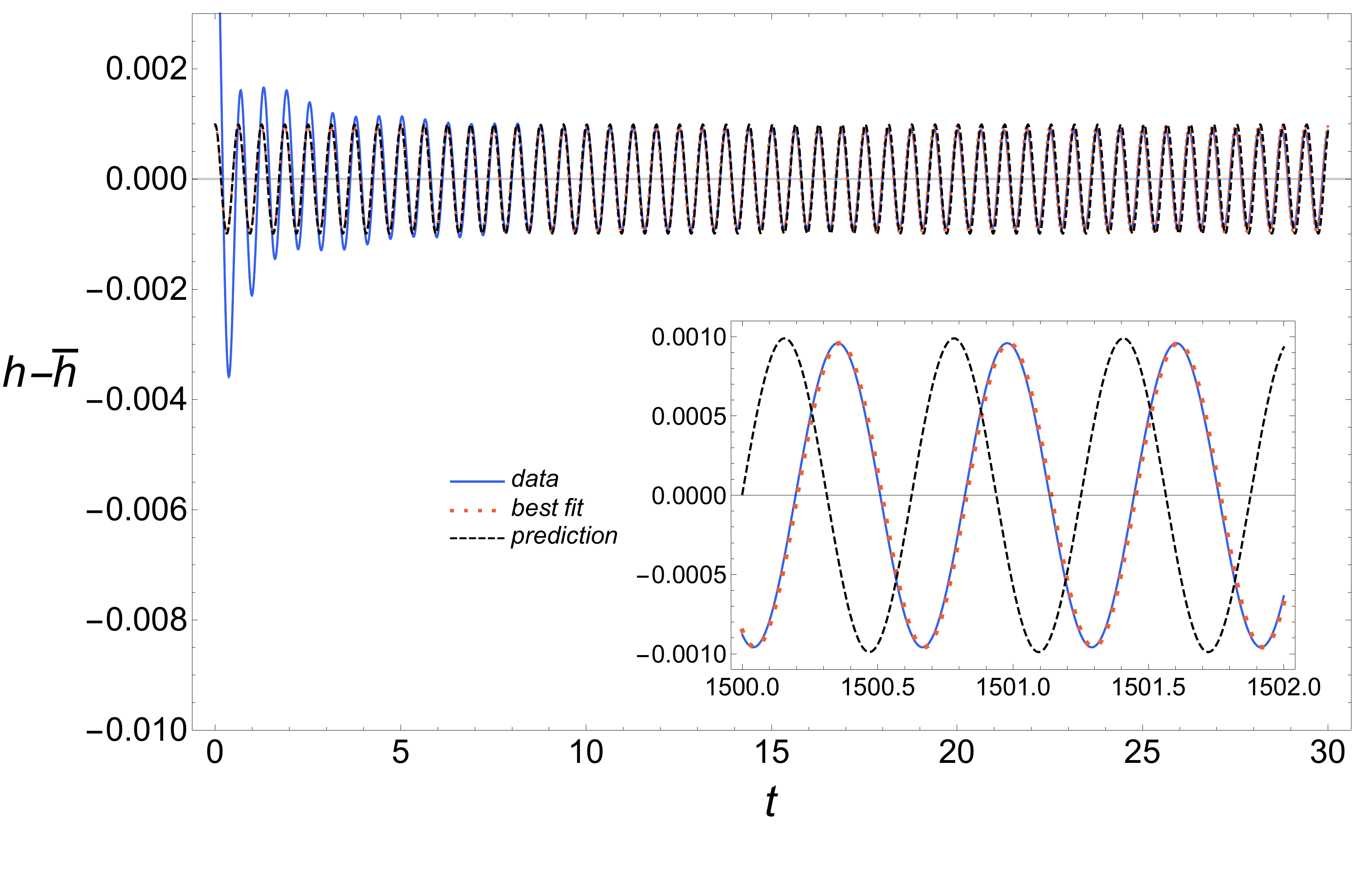}
\end{center}\caption{The time-dependent magnetic field $h(t)$ (minus its long time average $\bar h$) of the mean-field Hamiltonian \eref{eq:HMFIS} after a small quench from  the ground state of the TFIC \eref{eq:TFIC}, with  $g_0=1.5$, and Hamiltonian $H(\tilde g,\lambda)$ \eref{eq:TFICNI}, with $\tilde g=7.6513$ and $\lambda=3.5$. The agreement with a function of the form $h(t)\sim \bar h+c\cos(\mathcal E t+\varphi)$ becomes excellent at large times. 
The parameters of the fit are given by $c\approx 0.000958[0.000989]$, $\mathcal E \approx 10.0507[10.0353]$, $\varphi=0.0019[0]$, where in square brackets we reported the corresponding values based on the prediction \eref{eq:iseig}. Despite the parameters do not differ much from the (asymptotic) prediction, at large times (inset) the corrections to the frequency have conspicuous effects.
 }\label{f:smallquench}
\end{figure}

We stress that, within the linear approximation, the oscillation frequency is independent of $\tilde g$.
Figure \ref{f:smallquench} shows that the most important correction to \eref{eq:iseig} lies precisely in the frequency, essentially because it is multiplied by the time, which  has to be large for the subleading (time dependent) contributions to be negligible. However, in not-too-large time windows, the numerical data are in excellent agreement with \eref{eq:iseig} (the expression must be modified including a corrective phase shift if $(\tilde g -\bar g) J t$ is not small).

We point out that there is a third relevant point $\lambda_*<\lambda_c^-(<0)$ at which $\omega=0$:
\be
\fl\qquad\qquad\lambda_*=-\Bigl(\int_{-\pi}^\pi\frac{\mathrm d k}{\pi}\frac{\sin^2 k}{(1+g_0^2-2g_0\cos k)^{3/2}}\Bigr)^{-1}\, .
\ee
 For $\lambda<\lambda_*$, the isolated eigenvalue of $A$ becomes negative (\emph{i.e.} $\omega$ becomes purely imaginary). This would result in an exponential growth of \eref{eq:sollin}, which  after some finite time would no longer be consistent with the linearisation procedure. 
However, we point out that, for small quenches in which the energy is close to the ground state one, $\lambda$ is always larger than $\lambda_\ast$. 
Indeed, if we assume $\lambda\leq \lambda_*(g_0)$, 
it turns out that the ground state of \eref{eq:TFIC} is not equivalent to the ground state of $H(\bar  g(g_0,\lambda),\lambda)$, with $\bar g(g_0,\lambda)$ given by \eref{eq:noquench}. The latter state is instead associated with the ground state of the TFIC Hamiltonian \eref{eq:TFIC} with magnetic field $g_0'\neq g_0$ such that $\bar g(g_0',\lambda)=\bar g(g_0,\lambda)$ and $\lambda_\ast(g_0')<\lambda$. In the quench dephasing diagram of Figure \ref{f:diagram} we can indeed easily identify $\lambda_c^\pm$ (the `critical' curves for $\lambda$ positive and negative, respectively) but there is no trace of $\lambda_\ast$.

\paragraph{Interpretation.}   %

The bound state of the matrix $A$ may be put in relation with the existence of localised excitations in \eref{eq:TFICNI}. 
We emphasise that this is not an ab initio calculation but rather a physical picture that explains the observations. 

Since the time evolution of the expectation value of any local observable in $\ket{\psi_{g_0}}$ is stationary, we can assume that, to all intents and purposes, $\ket{\psi_{g_0}}$ is an eigenstate of $H(\bar g(g_0,\lambda),\lambda)$. 
We now consider the limit of small quench and assume $|\lambda|>|\lambda_c|$.
From \eref{eq:sollin} it follows that the projection on the bound state of $y_n''(t)+4i \omega y_n'(t)$ is proportional to an oscillating phase
\be\label{eq:guess}
\sum_n w^L_{0,n}(y_n''(t)+4i \omega y_n'(t))\propto e^{4i \omega t}\, .
\ee
The left hand side can be written as the expectation value $\braket{\psi_{\tilde g_0}(t)|B^\dag_{g_0}|\psi_{\tilde g_0}(t)}$ of a noninteracting operator with symbol 
\be
\sin k\frac{(h(t)-\cos(k))\sigma^x-\sin k\sigma^y+i \omega\sigma^z}{\omega^2-1-g_0^2+2g_0\cos k}\, ,
\ee
where $\tilde g_0$ approximately satisfyes \eref{eq:noquench1}. 
In the no-quench limit the mean-field parameter is constant $h(t)\rightarrow g_0$ (and $\tilde g_0\rightarrow g_0$), so  such operator can be written as
\be\label{eq:exit}
B_{g_0}^\dag \sim\sum_{\ell, n}(\begin{array}{cc}
a_\ell^x&a_\ell^y
\end{array})[\mathcal B_{g_0}^\dag]_{\ell-n}\Bigl(\begin{array}{c}
a_n^x\\
a_n^y
\end{array}\Bigr)\, ,
\ee
where
\be
[\mathcal B_{g_0}^\dag]_{\ell}=\int_{-\pi}^\pi\frac{\mathrm d k}{2\pi} e^{i \ell k} \sin k\frac{(g_0-\cos(k))\sigma^x-\sin k\sigma^y+i \omega\sigma^z}{\omega^2-1-g_0^2+2g_0\cos k}
\ee
and $a_\ell^\alpha$ are the Majorana fermions \eref{eq:Maj}. In \eref{eq:exit} we left out the normalisation. Importantly, $[\mathcal B_{g_0}^\dag]_{\ell}\sim e^{-|\ell|\theta}$, with $\theta$ defined in \eref{eq:theta}; 
we have therefore found a quasi-local operator whose expectation value approximately oscillates in time as in  \eref{eq:guess}. As a consequence, $\mathcal B_{g_0}^\dag$ acts like an excitation over the initial state. Indeed we have 
\be\label{eq:Lehmann}
\fl\qquad\qquad\braket{\psi_{\tilde g_0}(t)|B^\dag_{g_0}|\psi_{\tilde g_0}(t)}=\sum_{\ell, n}\braket{\psi_{\tilde g_0}|\ell}\braket{n|\psi_{\tilde g_0}}e^{i (E_\ell-E_n)}[B^\dag_{g_0}]_{\ell n}\sim e^{4i \omega t}
\ee
and if the state $\ket{\psi_{\tilde g_0}}$ has a sufficiently general representation (\emph{i.e.} the overlaps with the eigenstates of $H(\tilde g,\lambda)$ are generally nonzero), \eref{eq:Lehmann} tells us that $B^\dag_{g_0}$ connects only states with energy difference equal to $4\omega$.

In conclusion, the bound state of \eref{eq:A} seems to be a manifestation of a localised excitation of $H(\tilde g,\lambda)$.

\subsubsection{Remark on quenches from the ordered phase.}
In the previous section we ignored a subtlety that in principle could have invalidated part of the discussion (and part of the diagram in \Fref{f:diagram}). The mapping to a mean-field Hamiltonian relies on cluster decomposition properties but the ground state of the TFIC Hamiltonian in the ferromagnetic phase is the superposition of two Slater determinants \cite{CEF}, that separately do not possess cluster decomposition properties. Nevertheless, the discussion (and in turn \Fref{f:diagram}) remains correct also in this problematic case, at least for the operators commuting with $\prod_\ell\sigma_\ell^z$. This can be seen as follows:
\begin{enumerate}
\item The mean-field mapping is exact for the true ground state, which breaks the spin-flip symmetry realised by $\prod_\ell\sigma_\ell^z$. Thus \eref{eq:hm} is valid.
\item Using that $\sigma^z$ is a quadratic operator in the Jordan-Wigner fermions, in \eref{eq:hm} the ground state can be replaced by one of the two Slater determinants, therefore $h(t)$ is still solution of \eref{eq:system}. Analogously, the expectation value of any operator commuting with $\prod_\ell\sigma_\ell^z$ can be found replacing the initial state with one of the two Slater determinants and then using Wick theorem.
\item For any given $t$, the expectation value of operators $\mathcal O^{o}(\ell)$ that anticommute with $\prod_\ell\sigma_\ell^z$, like the order parameter, can be obtained from the large $r$ limit of  $\braket{\Psi_0| U^{\dag}_{\rm MF}(t) \mathcal O^o(\ell) \mathcal O^o(\ell+r)U_{\rm MF}(t)|\Psi_0}$, using cluster decomposition properties. 
\end{enumerate}
However, using similar general arguments we are not able to exclude that the expectation values of the odd operators might keep oscillating also when all the even operators relax. There is indeed a subtle problem of limits that comes from the trick of Point (iii). Nevertheless, in the cases considered we have never encountered this situation, suggesting that for the model under examination such complications do not arise. Here we provide a heuristic argument. When the limit of infinite time for $h(t)$ exists, at sufficiently long times the dynamics is essentially the same as for a quantum quench in the TFIC from a certain state with cluster decomposition properties. In the latter situation we can apply (a direct generalisation of) the results of \cite{FE:RDM}, which showed that the expectation values of odd operators decay to zero. This suggests that the diagram of \Fref{f:diagram} could be valid also taking into account the odd operators.

\subsection{Relaxation properties}\label{ss:relprop} %

We now focus on the regions (of the parameter space) in which the limit $\lim_{t\rightarrow\infty}h(t)\equiv h_\infty$ exists and investigate more closely the relaxation properties. 
Even beyond the linear approximation \eref{eq:sollin}, we can still guess an asymptotic form for the solution $y_k$ of \eref{eq:system}
\be\label{eq:asymp}
y_k\rightarrow \zeta_k+(A_k e^{2i\varepsilon_k t}+\mathrm{h.c.})\, ,
\ee
where $\varepsilon_k\sim2\sqrt{1+h_\infty^2-2h_\infty \cos k}$.
Equation \eref{eq:asymp} is compatible with the relaxation of $h(t)$ with corrections $O(t^{-1/2-j})$ (stationary phase approximation), where $j$ is an integer that depends on the behaviour of $A_k$ around the extremal points of the dispersion relation ($k=0\vee \pi$). 
Our numerical analysis for several quenches from TFIC initial states is compatible with $j=1$. This is not surprising since the same exponent governs the late time behaviour of $\braket{\sigma^z}$ after quenches in the TFIC \cite{CEF}. Thus, we conjecture the large time expansion
\be\label{eq:asympt}
\fl\qquad\quad h(t)\sim h_\infty+\frac{A_0\cos(4(1-h_\infty) t+\varphi_0)+A_\pi\cos(4(1+h_\infty) t+\varphi_\pi)}{t^{3/2}}\, ,
\ee
and similar behaviours for $\tilde y_n(t)$. 
Remarkably, the (leading) oscillatory frequency is only determined by $h_\infty$.  In Figure \ref{f:hrel}, \eref{eq:asympt} is compared against numerical data for a quench that leads to relaxation.  

\begin{figure}[tbp]
\begin{center}
\includegraphics[width=0.8\textwidth]{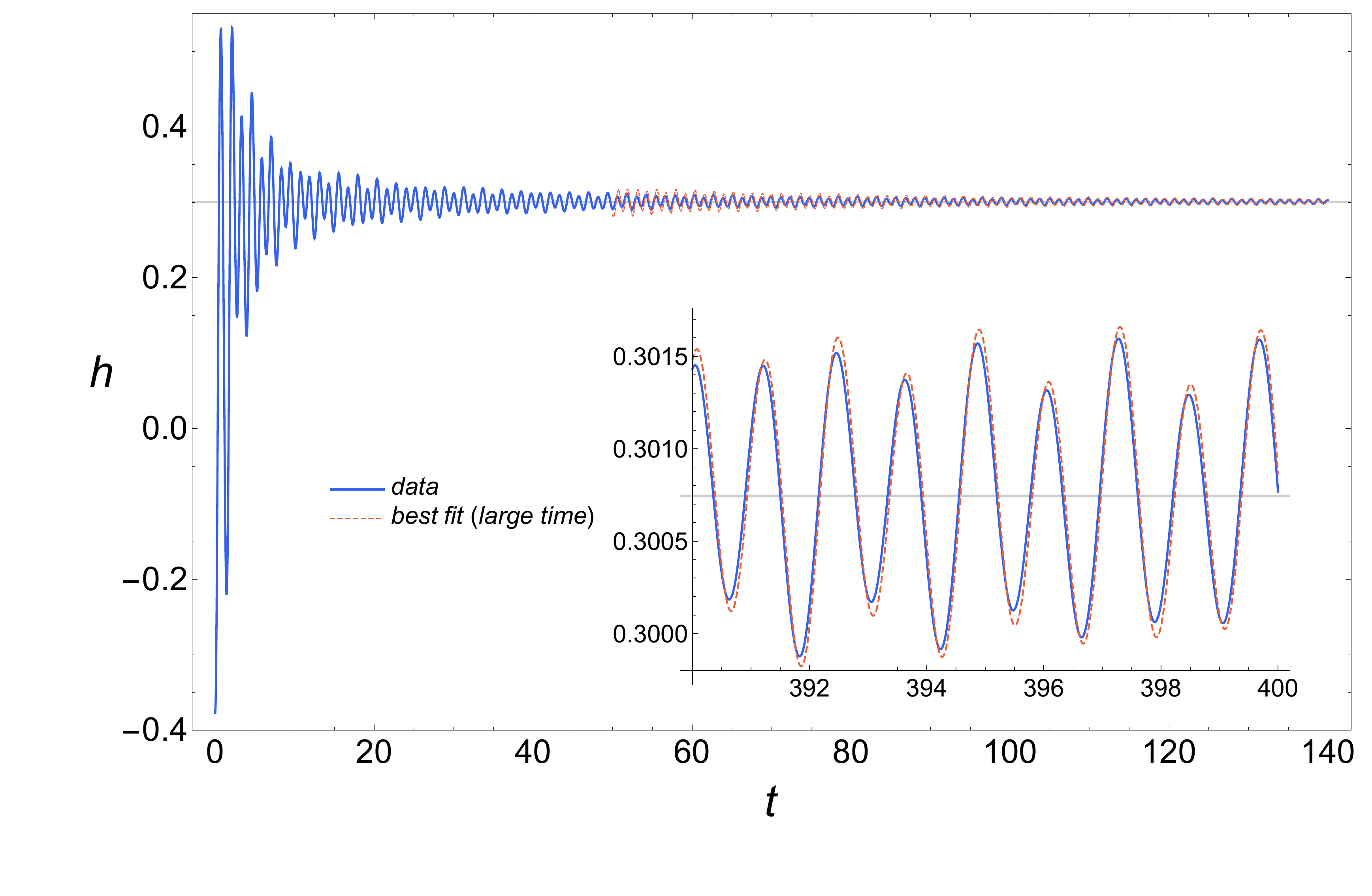}
\end{center}\caption{The time-dependent magnetic field $h(t)$ of the mean-field Hamiltonian \eref{eq:HMFIS} after a quench from  the ground state of the TFIC \eref{eq:TFIC}, with  $g_0=1.5$, and Hamiltonian $H(\tilde g,\lambda)$ \eref{eq:TFICNI}, with $\tilde g=0.5$ and $\lambda=0.5$. 
The dashed red line is the asymptotic prediction \eref{eq:asympt} with the coefficients estimated by fitting the data at very large times $t>1800$. The inset shows the goodness of the prediction in some intermediate time window. The horizontal grey line corresponds to the value of $h_\infty$ extracted from the fit.
}\label{f:hrel}
\end{figure}

Since for asymptotically large times the time evolution is equivalent to the one generated by the TFIC Hamiltonian
\be\label{eq:Hf}
H_f=-\sum_\ell(\sigma_\ell^x\sigma_{\ell+1}^x+h_\infty\sigma_\ell^z)\, ,
\ee
at late times local observables are described by a generalised Gibbs ensemble constructed with the local conservation laws of $H_f$.
However, the Lagrange multipliers can not be simply fixed by computing the corresponding integrals of motion at the initial time, as they are in fact conserved only at asymptotically large times.
This is an example of a stationary state written in terms of operators commuting with one another but not with the Hamiltonian. 
In fact, this is not the first time that such an unusual description emerges: in \cite{E:preT} the pre-thermalisation plateau was described by a GGE constructed with operators in involution that are however not conserved at the perturbative order that was worked out.

As it will be clarified in the next section, such a stationary state coincides with the pre-thermalisation plateau of \cite{IsingNI}. In the thermodynamic limit this is just the stationary state that emerges at \emph{infinite} time after the quench.
Indeed, in the regions in which there is relaxation, we have not found any indication of pre-relaxation/pre-thermalisation behaviour. 
Furthermore, `relaxation' is not synonym of `thermalisation'. 
Indeed it is not difficult to show that at late times the system still retains infinite information about the initial state. To this aim, as initial state we choose the ground state of the local Hamiltonian
\be\label{eq:Hi}
H[\{a\}]=\sum_j^n a_j H_j\, ,\qquad\qquad a_j\in\mathbb{R}\, ,
\ee
where $n$ is finite and $H_j$ are the most local reflection symmetric conservation laws of a TFIC model with Hamiltonian $H_{0}$.
Such state can be easily constructed \cite{AFC:exc} and, using the notations of Appendix \ref{a:GS}, is completely characterised by the function 
\be\label{eq:mform}
m(k)=-\mathrm{sgn}(\sum_j^n a_j \cos(j k))\, ,
\ee
which is equal to $1$ if and only if the excitation $\alpha^\dag_k$  of $H_0$ is present in the state. 
It is important to note that different characteristic functions $m(k)$ correspond to locally inequivalent states.
Thermal-like behaviour would imply that the only information about $m(k)$ that is retained at late times is the corresponding energy and magnetisation. 
We now consider the special cases in which the initial magnetisation is such that $H_0$ is the mean-field Hamiltonian at the initial time. In this way, by Corollary \ref{C:3}, the expectation value of local observables is independent of time and the late time stationary state is equivalent to the initial state. The only scenario compatible with thermal-like behaviour is that each distinct function $m(k)$ of the form \eref{eq:mform} corresponds to a distinct pair $\{h_m,\varepsilon_m\}$, where $h_m$ is the parameter of the late time mean-field Hamiltonian and $\varepsilon_m$ the energy.
The self-consistent conditions behind the no-quench limit are worked out in Appendix \ref{a:GS} and are given by
\be\label{eqe:hm}
h_m=g+2\lambda\int_{-\pi}^\pi\frac{\mathrm d k}{2\pi}m(k)\frac{h_m-\cos k}{\sqrt{1+h_m^2-2h_m\cos k}}\, .
\ee
\be\label{eqe:energyex}
\varepsilon_m=\frac{h_m^2-g^2}{4\lambda}-\int_{-\pi}^\pi\frac{\mathrm d k}{2\pi} m(k)\frac{h_m\cos k-1}{\sqrt{1+h_m^2-2h_m\cos k}}\, .
\ee
It is easy to see that the same pair $\{h_m,\epsilon_m\}$ is associated with infinitely many functions $m(k)$ of the form \eref{eq:mform} (it is enough to choose $n=3$ in \eref{eq:mform} to find (infinite) examples).  
Thus, the non-equilibrium time evolution under \eref{eq:TFICNI} does not generally result in thermalisation.
 
We conclude the analysis of \eref{eq:TFICNI} considering the expectation value $\braket{Q_n}_t$ of the local conservation laws of $H_f$ \eref{eq:Hf}, which can be written as follows \cite{FE:RDM}
\be
{\frac{Q_n}{L}}=\int_{-\pi}^\pi\frac{\mathrm d k}{2\pi} \cos(n k)\varepsilon_k\Bigl(\alpha^\dag_k\alpha^{\phantom \dag}_k-\frac{1}{2}\Bigr)\, ,
\ee 
where $\varepsilon_k$ is the dispersion relation of $H_f$ and $\alpha_k$ are noninteracting fermions that diagonalise $H_f$.  
Using free-fermion techniques we obtain 
\be\label{eq:Qcharge}
\braket{\frac{Q_n}{L}}_t=-\frac{h_\infty}{2}\tilde y_n-\frac{1}{8}\tilde\phi_n\, .
\ee
In particular, the expectation value of $H_f\equiv Q_0$ can be written as follows
\be
\braket{\frac{H_f}{L}}_t=\varepsilon-\frac{(h(t)-h_\infty)^2-(h_\infty-\tilde g)^2}{4\lambda}\, ,
\ee
from which it is clear that the relaxation exponent of $\braket{H_f}_t$ is twice the exponent of $h$ (which in the cases that we investigated is $3/2$, \emph{cf.} \eref{eq:asympt}). 
This result can be easily generalised to any local conservation law of $H_f$.
By taking the time derivative of \eref{eq:Qcharge}, the last equation of \eref{eq:systemF} implies
\be
\braket{\frac{Q_n}{L}}_t'=\frac{h-h_\infty}{2}y_n'\, .
\ee 
Since both $h-h_\infty$ and $y_n'$ are expected to decay as $t^{-3/2}$ (with oscillatory factors like in \eref{eq:asympt}), we immediately obtain
\be\label{eq:corrQ}
\braket{\frac{Q_n}{L}}_\infty-\braket{\frac{Q_n}{L}}_t\sim O(t^{-3})\, .
\ee

\subsubsection{Comparison with \cite{IsingNI}.}
It is not a coincidence that the same relaxation exponents were found in \cite{IsingNI} in a perturbative framework. 
In order to understand the relation between the two models we must come back to the modified version of Hamiltonian considered in \cite{IsingNI}, \emph{i.e.} \eref{eq:TFICNIS}. 
The mean-field Hamiltonian for that precise model reads
\be\label{eq:MFS}
\fl\qquad H_{\rm MF}(t)=-\sum_\ell\sigma_\ell^x\sigma_{\ell+1}^x-\texttt{g}\sum_\ell\sigma_\ell^z+2\lambda\braket{\sigma_\ell^z-\bar\sigma_\ell^z}_{t, \rm MF}\Bigl(\sum_\ell\sigma_\ell^z-\overline{\sum_\ell\sigma_\ell^z}\Bigr)\, ,
\ee
where the time average $\overline{\phantom{(}\!\!\!\cdots\phantom{)}\!\!\!}$ is taken with respect to the Hamiltonian with $\lambda=0$. It is convenient to introduce the auxiliary Hamiltonian
\be\label{eq:MFST}
\fl\qquad \tilde{H}_{\rm MF}(t)=-\sum_\ell\sigma_\ell^x\sigma_{\ell+1}^x-\texttt{g}\sum_\ell\sigma_\ell^z+2\lambda\braket{\sigma_\ell^z-\tilde\sigma_\ell^z}_{t, \rm MF}\Bigl(\sum_\ell\sigma_\ell^z-\widetilde{\sum_\ell\sigma_\ell^z}\Bigr)\, ,
\ee
where the time average $\widetilde{\phantom{(}\!\!\!\cdots\phantom{)}\!\!\!}$ is now taken with respect to  \mbox{$\tilde{H}_{f}\equiv\lim_{t\rightarrow\infty}\tilde{H}_{MF}(t)$}, limit which is assumed to exist. It is now simple to prove that $\tilde H_f=H(\texttt{g},0)$, where $H(\texttt{g},\lambda)$ is given in eq. \eref{eq:TFICNIS}.
To this aim let us consider the expectation value of $\sigma^z$ evolving via $\tilde{H}_{\rm MF}(t)$; it fulfils   
\bea
\fl\lim_{t\rightarrow\infty}\braket{\tilde \sigma_\ell^z}_{t,\rm MF} &=&\lim_{t\rightarrow\infty}\lim_{T\rightarrow\infty}\frac{1}{T}\int_{0}^{T}\textrm{d}s\, \braket{\Psi_0|\tilde{U}^{\dag}_{MF}(t)e^{i \tilde{H}_{f} s }\sigma_\ell^z e^{-i \tilde{H}_{f} s }\tilde{U}_{MF}(t)|\Psi_0} \nn
&=&\lim_{t\rightarrow\infty}\lim_{T\rightarrow\infty}\frac{1}{T}\int_{0}^{T}\textrm{d}s\, \braket{\Psi_0|\tilde{U}^{\dag}_{MF}(t+s) \sigma_\ell^z \tilde{U}_{MF}(t+s)|\Psi_0}  \nn
&=&\lim_{t\rightarrow\infty} \braket{ \sigma_\ell^z}_{t,\rm MF}= \braket{ \sigma^z}_{\infty,\rm MF}\,,
\eea
where $\tilde{U}_{MF}(t)$ is the time evolution operator constructed with $\tilde{H}_{MF}(t)$ and in the second step we replaced $e^{-i \tilde{H}_{f} s} \tilde{U}_{MF}(t)$  with $\tilde{U}_{MF}(t+s)$, as it is legitimate at late times. From this it follows $\tilde H_f=H(\texttt{g},0)$ and, in turn, the equivalence between \eref{eq:MFS} and \eref{eq:MFST} $H_{\rm MF}(t)=\tilde H_{\rm MF}(t)$. Importantly, this means that \eref{eq:MFS} and \eref{eq:TFIC} have the same infinite time limit if we set $h_\infty=\texttt{g}$.

As a matter of fact, the equivalence between \eref{eq:MFS} and \eref{eq:TFIC} is not restricted to infinite times.
The mean-field time evolution operator for \eref{eq:MFS} can indeed be written as follows
\be\label{eq:UMfactor}
\fl \qquad U_{\rm MF}(t)= e^{4i \lambda\int_0^t\mathrm d s\braket{\delta \sigma^z}_{s, \rm MF}\overline{S^z}}\mathrm{T}\exp\Bigl(-i\int_0^t\mathrm d s\ H_f+4\lambda\braket{\delta \sigma^z}_{s, \rm MF}S^z(s)\Bigr)\, .
\ee
Here $S^z=\frac{1}{2}\sum_\ell\sigma_\ell^z$, $\delta\sigma^z=\sigma^z-\bar\sigma^z$ and
\be
S^z(s)=e^{-4i \lambda\int_0^s\mathrm d s'\braket{\delta \sigma^z}_{s', \rm MF}\overline{S^z}}S^ze^{4i \lambda\int_0^s\mathrm d s'\braket{\delta \sigma^z}_{s', \rm MF}\overline{S^z}}\, .
\ee
If the magnetisation relaxes faster than $1/t^{1+\alpha}$, with $\alpha>0$, the operator at the exponent of the first term of \eref{eq:UMfactor} is a bounded function of the time, so that exponential can be safely expanded. In fact the entire term can be  neglected (it gives corrections $o(\lambda)$).  
The same holds true in $S^z(s)$, indeed the finiteness of $\sup_{s}\lambda\int_0^s\mathrm d s'\braket{\delta \sigma^z}_{s', \rm MF}$ guarantees that the series expansion of the exponentials in $S^z(s)$ can be truncated for any $s$ with an error that goes to zero as $\lambda\rightarrow 0$. By considering the first terms of the expansion one immediately realizes that the correction is $o(\lambda)$ and approaches zero for large $s$ as $1/s^{\alpha}$. Putting everything together, replacing $S^z(s)$ by $S^z$ in \eref{eq:UMfactor} introduces an error $o(\lambda)$, independently of the time.
Therefore we obtain
\be\label{eq:UMfactor1}
 U_{\rm MF}(t)\sim \mathrm{T}\exp\Bigl(-i\int_0^t\mathrm d s\ H_f+4\lambda\braket{\delta \sigma^z}_{s, \rm MF}S^z\Bigr)+o(\lambda)\, .
\ee
We can also replace $\braket{\delta \sigma^z}_{s, \rm MF}$ with $\braket{\sigma^z}_{s, \rm MF}-\braket{\sigma^z}_{\infty, \rm MF}$: the difference between the two terms is $o(\lambda^0)$ and approaches zero at large times at least as $1/t^{1+\alpha}$. 
In this way we have reduced the Hamiltonian \eref{eq:MFS} to \eref{eq:TFICNI},
provided that the condition  (see first equation of \eref{eq:systemF}) $\tilde g= \texttt{g}+4\lambda  m^z_\infty$ is satisfied. 
With this choice we recover the perturbative results of \cite{IsingNI}, \emph{e.g.} the relaxation exponents \eref{eq:corrQ}.  \emph{A posteriori} we note that the large time behaviour of $\braket{\sigma_\ell^z}$ under \eref{eq:TFICNI} (\emph{cf}. \eref{eq:asympt}) is sufficiently fast ($\alpha=1/2$) to justify the approximation of \eref{eq:UMfactor} by \eref{eq:UMfactor1}. 
We also checked that the mean-field solution of \eref{eq:TFICNI} is perfectly compatible with the results shown in Figure 1 of \cite{IsingNI} (see \Fref{f:mirror59}).

\begin{figure}[tbp]
\begin{center}
\includegraphics[width=0.8\textwidth]{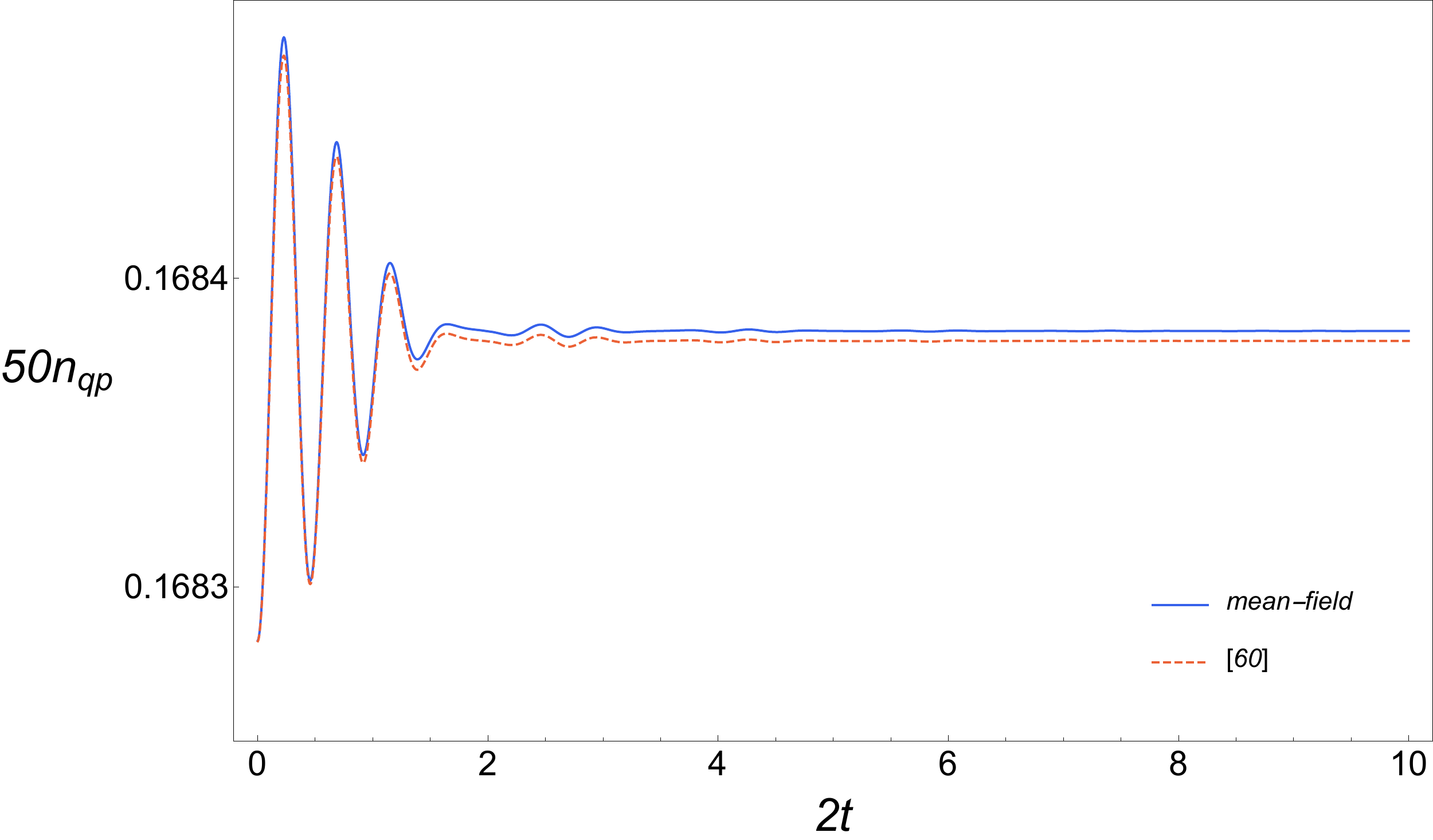}
\end{center}\caption{The time evolution of the number of quasiparticles $n_{qp}=\int_{-\pi}^\pi\frac{\mathrm d k}{2\pi} \alpha^\dag_k\alpha_k$ that diagonalise the late-time mean-field Hamiltonian for a quench with $g_0=8$, $\lambda=0.05$ and $\tilde g=3.597274\dots$. The parameters are chosen to reproduce the first figure of \cite{IsingNI} (dashed orange line). In particular, $\tilde g$ is such that the mean-field parameter $h(t)$ in the limit of infinite time approaches the value $\texttt{g}=3.5$, considered in \cite{IsingNI} (see the main text). The timescale and $\lambda$ differ from \cite{IsingNI} because of two small typos (the dispersion relation was unintentionally halved and the right hand side of Equation (3) of \cite{IsingNI} should have been multiplied by $4$). The (tiny) discrepancy is compatible with higher order corrections in $\lambda$.} \label{f:mirror59}
\end{figure}

Finally, we point out that  \cite{IsingNI} introduced the term $\overline{\sum_\ell\sigma_\ell^z}$ to fix some conditions in the long-time limit, where \eref{eq:TFICNI} and the Hamiltonian of \cite{IsingNI} turn out to be equivalent.
\Eref{eq:TFICNI} is therefore a sensible replacement for the Hamiltonian of \cite{IsingNI}. 

\subsection{Generalisations}
\begin{figure}[tbp]
\begin{center}
\includegraphics[width=0.8\textwidth]{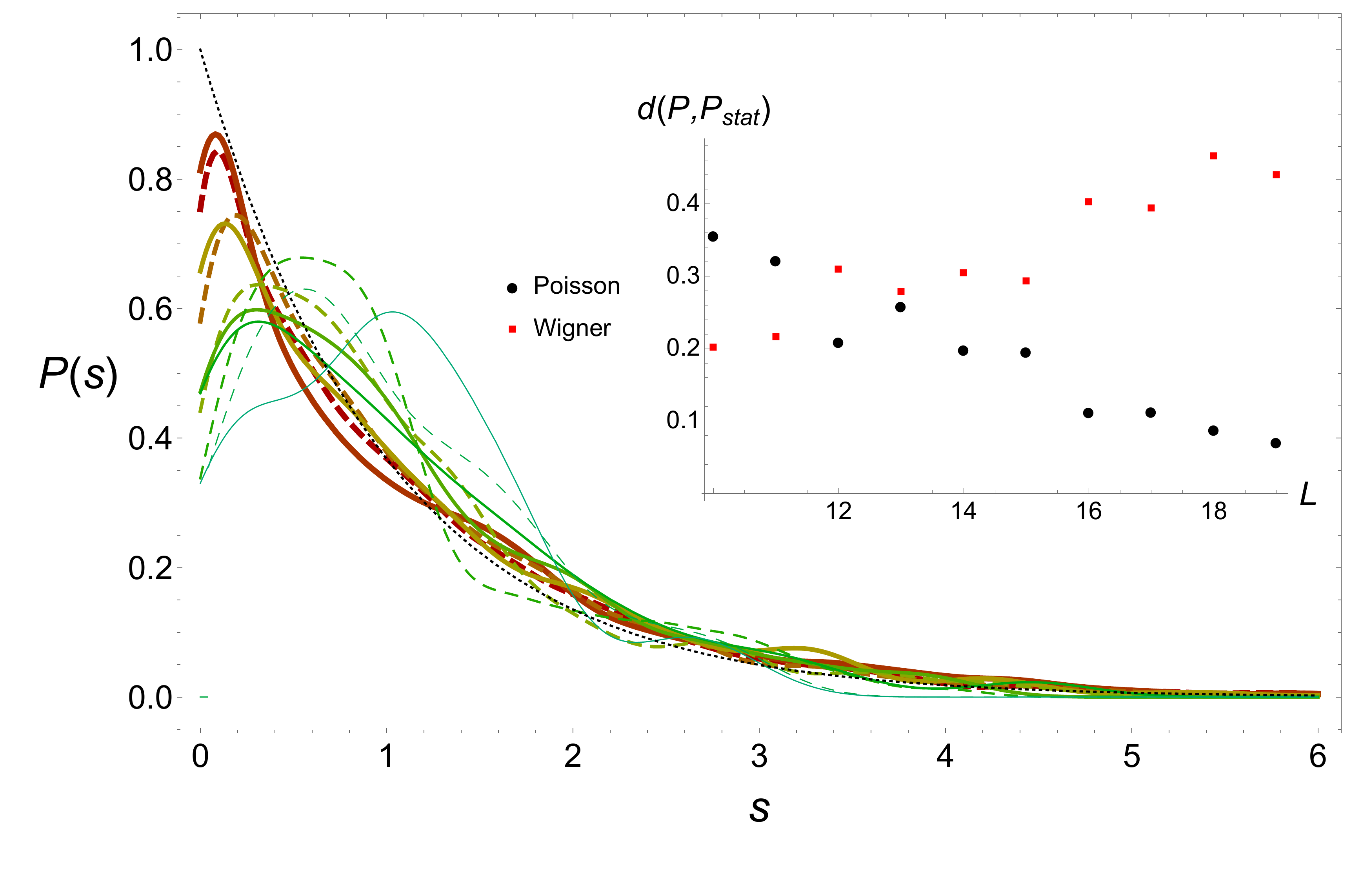}
\end{center}\caption{The nearest neighbour spacing distribution $P(s)$ of the Hamiltonian \eref{eq:TFICNI} with $\tilde g=0.5$ and $\lambda=0.5$ in the reflection symmetric sector of the zero momentum subspace with spin-flip parity $\prod_\ell\sigma_\ell^z$ equal to 1 for various chain sizes ($L=10\div 19$).  
As the size is increased the colours vary from green to brown and the lines become thicker. Dashed lines correspond to odd sizes. The dotted black line is the exponential distribution (Poisson statistics).  In the inset the distance $d(P,P_{stat})=\sqrt{\int_0^\infty(P(s)-P_{stat}(s))^2}$ from Poisson ($P_{stat}(s)=e^{-s}$) and Wigner ($P_{stat}(s)=\frac{\pi}{2} s e^{-\frac{\pi}{4}s^2}$).  
}\label{f:lsTFICNI}
\end{figure}

\begin{figure}[tbp]
\begin{center}
\includegraphics[width=0.8\textwidth]{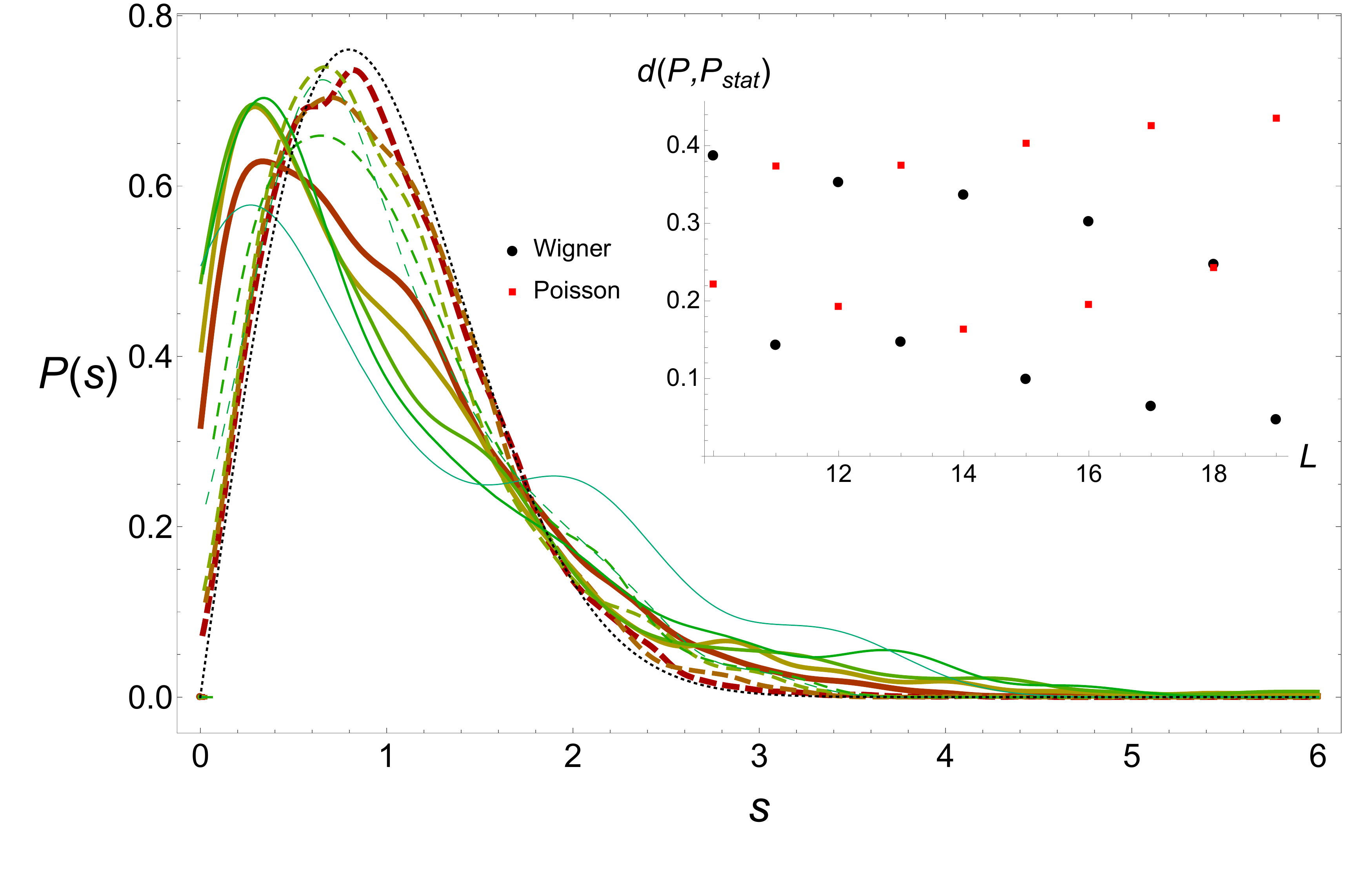}
\end{center}\caption{The nearest neighbour spacing distribution $P(s)$ of the Hamiltonian \eref{eq:HWig} with $\gamma=0.25$ and $\lambda=0.5$ in the reflection symmetric sector of the zero momentum subspace with spin-flip parity $\prod_\ell\sigma_\ell^z$ equal to 1 for various chain sizes ($L=10\div 19$).  
As the size is increased the colours vary from green to brown and the lines become thicker. Dashed lines correspond to odd sizes. The dotted black line is the Wigner distribution.  In the inset the distance $d(P,P_{stat})=\sqrt{\int_0^\infty(P(s)-P_{stat}(s))^2}$ from Poisson ($P_{stat}(s)=e^{-s}$) and Wigner ($P_{stat}(s)=\frac{\pi}{2} s e^{-\frac{\pi}{4}s^2}$). The distribution for odd chains converges rather quickly to Wigner. 
}\label{f:lsXYZNI}
\end{figure}

The construction of (low-entangled) stationary states that we proposed in the previous section (\eref{eq:Hi} and below) can be applied also to other Hamiltonians of the form \eref{eq:opform} if the corresponding mean-field Hamiltonian is integrable. In those cases we can rule out thermalisation if the self-consistent problem satisfied by the stationary solutions at fixed energy and mean-field parameters has more than one solution.

When the mean-field Hamiltonian is noninteracting, following the lines of the proof sketched for the nonlocal generalisation of the Ising model in \Sref{ss:relprop}, one can generally show that the solution is not unique. 

It is also reasonable to expect that, also in the presence of interactions, the finite number of constraints given by the energy conservation and the  late time values of the mean-field coupling constants could not reduce the parameter space of the initial Hamiltonian \eref{eq:Hi} to a single point. We indeed believe that thermalisation is unlikely to emerge if the mean-field Hamiltonian describes an integrable model at any time. 

Nevertheless, the interacting case exhibits counterintuitive behaviours, for example in the energy level-spacing statistics. Generally integrable models exhibit Poisson statistics, whereas generic models follow a Wigner distribution \cite{{BT:stat}, {GMGW:stat}}. There are many exceptions to this rule \cite{CM:rem}, however the nearest neighbour spacing distribution is probably the most reliable numerical check of integrability. 

In Figure \ref{f:lsTFICNI} the level spacing distribution is shown for various chain sizes for the Hamiltonian \eref{eq:TFICNI} with $\lambda=0.5$ and $\tilde g=0.5$ in the reflection symmetric sector of the zero momentum subspace with spin-flip parity $\prod_\ell\sigma_\ell^z$ equal to 1. The numerical data suggest that in the thermodynamic limit the curves collapse to an exponential distribution (Poisson statistics). This is consistent with our observation that at arbitrarily large times after a quantum quench the system keeps retaining infinite information about the initial state.

A completely different scenario appears for the Hamiltonian
\be\label{eq:HWig}
H=-\sum_{\ell}\Bigl(\frac{1+\gamma}{4}\sigma_\ell^x\sigma_{\ell+1}^x+\frac{1-\gamma}{4}\sigma_\ell^y\sigma_{\ell+1}^y\Bigr)+\frac{\lambda}{4 L}\Bigl(\sum_{\ell}\sigma_\ell^z\sigma_{\ell+1}^z\Bigr)^2\, .
\ee 
The corresponding mean-field Hamiltonian for a given one-site shift invariant initial state $\ket{\Psi_0}$ is given by
\be\label{eq:HWigMF}
\fl \quad H_{\rm MF}^{\Psi_0}(t)=-\sum_{\ell}\Bigl(\frac{1+\gamma}{4}\sigma_\ell^x\sigma_{\ell+1}^x+\frac{1-\gamma}{4}\sigma_\ell^y\sigma_{\ell+1}^y\Bigr)+\frac{\lambda}{2}\braket{\Psi_t|\sigma_\ell^z\sigma_{\ell+1}^z|\Psi_t}\sum_{\ell}\sigma_\ell^z\sigma_{\ell+1}^z\, .
\ee 
At fixed time, this describes an XYZ model, which is known to be integrable for any choice of the coupling constants. Therefore, assuming relaxation, the stationary properties of local observables should be described by a GGE constructed with the conservation laws of an XYZ model. 
We considered $\gamma=0.25$ and $\lambda=0.5$.
We found large finite size corrections in the level-spacing statistics (in the same sector as before) and, in particular, a remarkable even-odd parity effect (see Figure \ref{f:lsXYZNI}). 
However, it seems that increasing $L$ the curves collapse to a Wigner distribution, as it commonly happens in non-integrable models. 

This might appear in contradiction  with our conjecture that thermalisation should not be expected when the mean-field Hamiltonian is interacting and integrable.  
In fact, we have not taken into account that at late times the mean-field parameters are fixed.
In the previous section we ruled out thermalisation by constructing an infinite family of stationary states with the same energy and the same mean-field parameters. If this is possible, then we should find a signature of the huge degeneracy in the level-spacing statistics by restricting the space to the excited states that lie in some shell with mean-field parameters almost fixed. 

In the restricted space our preliminary analysis is indeed compatible with Poisson statistics also for the Hamiltonian \eref{eq:HWig}. However, our data turn out to be compatible with Poisson statistics even if the mean-field Hamiltonian does not describe an integrable model.
This is in contrast to our expectations that in generic systems there should not be more than a few parameters that characterise the stationary state, namely the energy and, at worst, the mean-field coupling constants at infinite time after the quench.

Our interpretation of these contradictory results is that we did not investigate sufficiently large chains, so our analysis of the energy-level statistics in the restricted space is not sufficiently indicative. We are confident that a more accurate analysis will show a different behaviour in the non-integrable case. 

Finally, we point out that the situation is trickier when there are isolated points in the parameter space of the mean-field coupling constants that correspond to integrable models. For example, it is not clear to us whether or not we should expect thermalisation when at asymptotically large times the coupling constants of the mean-field Hamiltonian match the integrability points. 

\subsection{Summary}

We showed that the time evolution under \eref{eq:TFICNI} has a quite rich phenomenology, including both cases of relaxation and cases of persistent oscillatory behaviour.
In the limit of small quench the latter has been interpreted as the effect of localised excitations that appear (or become relevant) when the Hamiltonian parameters cross some ``critical'' line. 

In addition, we confirmed the perturbative results of \cite{IsingNI} in a non-perturbative setup. Our analysis excludes that in the thermodynamic limit the late time behaviour of local observables could be described by a thermal-like ensemble.

More generally, we provided some argument that suggests that thermalisation is unlikely to emerge if the mean-field Hamiltonian describes an integrable model at any time after the quench. 

We also proposed a numerical check of thermalisation based on the analysis of the energy-level statistics on some restricted space, however our preliminary analysis on small chains ($L<20$) was not sufficient to discriminate between the cases in which we expect thermal-like behaviour and the cases in which instead also at late times  infinite information about the initial state is retained. 

\section{Conclusions}\label{s:conclusions}  %

Pre-relaxation is a dynamical phenomenon that arises when small perturbations break symmetries that affect the late time behaviour of local observables. 
When the perturbation breaks (abelian) integrability, this is usually called pre-thermalisation, which is generally thought as a two-step process in which local observables experience virtual relaxation before approaching thermal-like expectation values. However the relaxation process can also be more complicated, following many steps of quasi-stationary behaviour. This happens in particular when the model is close to a non-abelian integrable point. In order to extract the pre-relaxation behaviour one must therefore identify the correct time scale of the phenomenon. 

We have considered the problem of pre-relaxation after quantum quenches in weakly interacting models, starting from initial states with cluster decomposition properties. We focussed on the particular situation in which the unperturbed (one-site shift invariant) Hamiltonian has a non-abelian set of local conservation laws that break one-site shift invariance. In particular we considered interacting perturbations to the XY spin-$\frac{1}{2}$ chain and investigated both integrable extensions, like the Heisenberg XYZ model, and the effects of perturbations that break integrability.

We identified the inverse perturbation strength as a relevant time scale of pre-relaxation and studied the dynamics of local observables at times proportional to it. 

Despite the model being interacting, the noninteracting structure, remnant of the unperturbed Hamiltonian and manifested in the Wick theorem, survives the pre-relaxation limit. However interactions do affect the dynamics by introducing a nontrivial time dependence in the effective noninteracting Hamiltonian that generates the time evolution. The most striking effect is probably that, even if local degrees of freedom approach stationary values, these can not be generally predicted without following the entire dynamics.

We have shown how to recast the non-equilibrium problem into a system of nonlinear differential equations involving expectation values of quasi-local operators. 
The system of equations has qualitatively distinct solutions, which vary from trivial stationarity to persistent oscillatory behaviour over the entire time window considered. 
We have not found any relevant difference between integrable and non-integrable perturbations, suggesting that the scenario of thermalisation in generic models arises at much larger times. 

For the very nature of the local conservation laws of the XY model, in order to have a nontrivial time evolution the initial state must break one-site shift invariance. 
For a particular initial state of that kind we considered a limit in which the equations can be linearised and exhibited the analytic solution, in which one-site shift invariance is eventually restored. 
The regime worked out analytically shows quite clearly the \emph{importance of cluster decomposition} in the non-equilibrium problem. While, as mentioned above, the pre-relaxation limit is trivial for one-site shift invariant states, a shift symmetrisation of the two-site shift invariant initial state has a nontrivial time evolution. This is because cluster decomposition has been lost with the symmetrisation. It is important to take into account such aspect when analytic predictions of late time stationary behaviour are compared with numerical data at times in which one-site shift invariance is not yet restored.  

The crossover between oscillatory behaviour and relaxation is quite interesting \emph{per se}. This has been the main motivation for the analysis of a simplified model that shares most of the formal aspects with the effective description of pre-relaxation in the perturbed XY model, but that, in fact, has not been derived from a pre-relaxation  limit. 
We considered a transverse-field Ising chain with an additional nonlocal interaction proportional to the magnetisation squared per unit length. This model was already studied in \cite{IsingNI} in the framework of a perturbation theory. We used some general properties, proven for Hamiltonians of that form, to obtain nonperturbative results and showed that in the thermodynamic limit subsystems retain infinite information about the initial state, whatever large the time is. 
This is not in disagreement with \cite{IsingNI}, where thermalisation was conjectured for time averages in finite systems: for this model the diagonal ensemble could not be locally equivalent to the stationary state that emerges in the thermodynamic limit when the quench parameters are compatible with relaxation.     

We showed that the late time behaviour in the thermodynamic limit (which corresponds to the `pre-thermal' behaviour of \cite{IsingNI}) can not always be described by a stationary state. In the parameter space there are indeed `critical lines' that separate relaxation from persistent oscillatory behaviour. We defined a limit of small quench and, in that limit, exhibited the analytic expressions for such critical lines. The appearance of oscillatory behaviour has been interpreted as a consequence of the emergence of localised excitations.

We also discussed the generalisations to other Hamiltonians in which some terms have the form of interactions with macroscopic observables, like the magnetisation squared per unit length of the model above. In particular, we ruled out thermal-like behaviour in a large class of models of that kind. 

Finally, we would like to stress that our description of the pre-relaxation limit is based on a few hypotheses. In particular, we neglected some ``anomalous terms'', proving only the self-consistency of the conjecture. Some preliminary checks against iTEBD simulations are confirming the validity of the assumptions\cite{CF:entgr}; however, we leave a more rigorous analysis of the regimes of validity of our approximations to future research. 

\ack
We thank Mario Collura, Fabian Essler and Andrea Gambassi for useful discussions and Matteo Marcuzzi for having provided us with the data for \fref{f:mirror59}. This work was supported by the EPSRC under grants EP/I032487/1 (BB) and EP/J014885/1 (MF)
and by the LabEx ENS-ICFP: ANR-10-LABX-0010/ANR-10-IDEX-0001-02 PSL* (MF).

\appendix
\section*{APPENDICES}

\section{Free-fermion relations}\label{a:free}%

We briefly summarise some useful relations valid in noninteracting models. Additional details can be found in \cite{F:super} and \cite{FC:disjoint}.

Let us consider a local one-site shift invariant spin-chain Hamiltonian $H$ that is mapped to noninteracting fermions by the Jordan-Wigner transformation ($\{a_\ell^\alpha,a_n^\beta\}=2\delta_{\alpha \beta}\delta_{\ell n}$)
\be
a_\ell^\alpha=\Bigl(\prod_{j<\ell}\sigma_j^z\Bigr) \sigma_\ell^\alpha\qquad \alpha\in\{x,y\}\, .
\ee 
Up to boundary terms, $H$ reads as
\be\label{a:eq:free}
\fl\qquad \qquad H\sim\frac{1}{4}\sum_{\ell, m}^{L/n}(\begin{array}{ccccc}
a_{n\ell-n+1}^x&a_{n\ell-n+1}^y&\dots&a_{n\ell}^x&a_{n\ell}^y
\end{array})[\mathcal H^{(n)}]_{\ell m}\left(\begin{array}{c}
a_{n m-n+1}^x\\a_{n m-n+1}^y\\\vdots\\a_{n m}^x\\a_{n m}^y
\end{array}\right)\, ,
\ee
where $L$ is the chain length and $\mathcal H$ is a block-circulant\footnote{A circulant matrix $M$ is a Toeplitz matrix ($M_{nm}\equiv M_{n-m}$ ) in which any row is a right cyclic shift of the row above.} matrix
\be
[\mathcal H^{(n)}]_{\ell m}= \frac{n}{L}\sum_{k}e^{-i(m-\ell) k}\mathcal H^{(n)}(k)\qquad e^{i L k/n}=1\, ,
\ee
where
\be
 {\mathcal H^{(n)}}^\dag(k)=\mathcal H^{(n)}(k)\,,\qquad {\mathcal H^{(n)}}^T(k)=-\mathcal H^{(n)}(-k)\, .
\ee
The index $n$ is a divisor of $L$ (but in the thermodynamic limit any positive integer is allowed).  
We call the $2n$-by-$2n$ matrix $\mathcal H^{(n)}(k)$ the $n$-site representation of the  \emph{symbol}, which completely characterises the block-circulant matrix. In the following we will refer to $\mathcal H^{(n)}(k)$ also as the symbol of $H$.  

For given $n$, the following properties hold:
\begin{enumerate}
\item Any function of block-circulant matrices is a block circulant matrix, with symbol equal to the function of the respective symbols
\be
f(\mathcal A, \mathcal B, \dots)\rightarrow f(\mathcal A(k),\mathcal B(k),\dots)\, .
\ee 
\item Let $A$ and $B$ as in \eref{a:eq:free}. Their commutator $[A,B]$ has the form \eref{a:eq:free}, with symbol equal to the commutators of the symbols
\be
[A,B]\rightarrow[\mathcal A(k),\mathcal B(k)]
\ee
\item The time evolution in the Heisenberg picture under \eref{a:eq:free}  of a noninteracting operator $A$ of the same form is noninteracting, with symbol 
\be
\label{timeevolution}
e^{i H t} A e^{-i H t}\rightarrow e^{i \mathcal H(k) t} A(k)e^{-i \mathcal H(k) t}
\ee
\item (Wick theorem) The expectation value of any operator in a thermal state of  \eref{a:eq:free} (and in any Slater determinant state) can be expressed in terms of the correlation matrix
\be
\fl\qquad { \Gamma^{(n)}_{\ell m}}=\delta_{\ell m}\mathrm I-\braket{\left(\begin{array}{c}
a_{n m-n+1}^x\\a_{n m-n+1}^y\\\vdots\\a_{n m}^x\\a_{n m}^y
\end{array}\right)(\begin{array}{ccccc}
a_{n\ell-n+1}^x&a_{n\ell-n+1}^y&\dots&a_{n\ell}^x&a_{n\ell}^y
\end{array})}\, .
\ee
\item For a thermal state with inverse temperature $\beta$ the correlation matrix is given by
\be
\Gamma_\beta=-\tanh\bigl(\frac{\beta}{2}\mathcal H\bigr)\, .
\ee 
Therefore, the ``thermal'' ground state (which, in the presence of degeneracies, is equivalent to the incoherent superposition of the states) has correlation matrix
\be
\Gamma_\infty=-\mathrm{sgn}(\mathcal H)\, .
\ee
\item If $\mathcal H$ is block-circulant, $\Gamma$ is block-circulant as well, so it is completely characterised by its symbol $\Gamma(k)$. In particular, the symbol of the thermal ground state reads
\be
\Gamma_\infty(k)=-\mathrm{sgn}(\mathcal H(k))\, .
\ee
\item The correlation matrix of a Slater-determinant state that time evolves under the noninteracting Hamiltonian \eref{a:eq:free} has the symbol
\be
\Gamma(k;t)=e^{-i\mathcal H(k) t}\Gamma(k;0)e^{i\mathcal H(k) t}\, .
\ee
\item In the thermodynamic limit ($L\rightarrow\infty$), the expectation value of an operator $A$ of the form \eref{a:eq:free} in a Slater-determinant state with correlation matrix $\Gamma$ can be written as follows
\be
\lim_{L\rightarrow\infty }\frac{1}{L}\braket{A}=\frac{1}{4 n}\int_{-\pi}^{\pi}\frac{\mathrm d k}{2\pi}\tr[\Gamma(k)\mathcal A(k)]\, .
\ee
\end{enumerate}

\section{Time averages of interacting operators}\label{a:Wick} %

In this appendix we show the validity of Property \ref{L:Wick}. The Property can be more easily proven for Jordan-Wigner fermions 
\be
c_j=\frac{1}{2}(a_j^x-i a_j^y)\, ;
\ee
the relation for the Majorana fermions $a^x_j, a_j^y$ will then follow by linearity.

In order to proceed it is convenient to introduce the following notation 
\be
\mathbf{c}_{j}^{\alpha}(t)\equiv
\Biggl\{
\begin{array}{ll}
c^{\dag}_j(t) &\alpha=+\\
c_j(t) &\alpha=-\,.
\end{array}
\ee  
The relation between $c^\dag,c$ and the Bogoliubov fermions $b^\dag(k),b(k)$ that diagonalise the unperturbed (noninteracting) Hamiltonian $H_{XY}$ \eref{eq:XY} can be written as
\be
\label{eq:Bog}
\mathbf{c}^{\alpha}_{j}(t)=\frac{1}{\sqrt{L}}\sum_{k}\sum_{\beta=\pm} e^{i \alpha j k}U(k)^{\alpha}_{\beta} \mathbf{b}_\beta(k)e^{i \beta \varepsilon_k t}\,.
\ee
Here $U(k)$ is the $2\times2$ matrix defining the Bogoliubov transformation, $\varepsilon_k$ is the dispersion relation in the 1-site shift invariant representation
\be\label{a:disp}
\varepsilon_k=J\sqrt{\cos^2 k+\gamma^2\sin^2 k}
\ee
and we set
\be
 \mathbf{b}_\beta(k)\equiv
 \Biggl\{
\begin{array}{ll}
b^{\dag}(k) &\beta=+\\
b(-k)               &\beta=-\,.
\end{array}
 \ee
The relation  \eref{eq:Wick} is then equivalent to 
\be
 \label{eq:WickJWls}
\frac{1}{L}\overline{\sum_j \mathbf{c}_{j+n1}^{\alpha_1} \mathbf{c}_{j+n2}^{\alpha_2} \mathbf{c}_{j+n3}^{\alpha_3} \mathbf{c}_{j+n4}^{\alpha_4}}= \mathcal F_{\{n\}}^{\{\alpha\}}+ \mathcal A_{\{n\}}^{\{\alpha\}}\, ,
\ee
where ${\mathcal F_{\{n\}}^{\{\alpha\}}}$ is a factorised term
\be\label{eq:WickJWrs}
{\mathcal F_{\{n\}}^{\{\alpha\}}}=\sum_{s=0}^1\underbrace{\mathbf{c}_{n_1}^{\alpha_1} \mathbf{c}_{n_2}^{\alpha_2}}_{s}\underbrace{\mathbf{c}_{n_3}^{\alpha_3} \mathbf{c}_{n_4}^{\alpha_4}}_{s}-\underbrace{\mathbf{c}_{n_1}^{\alpha_1} \mathbf{c}_{n_3}^{\alpha_3}}_{s}\underbrace{\mathbf{c}_{n_2}^{\alpha_2} \mathbf{c}_{n_4}^{\alpha_4}}_{s}+\underbrace{\mathbf{c}_{n_2}^{\alpha_2} \mathbf{c}_{n_3}^{\alpha_3}}_{s}\underbrace{\mathbf{c}_{n_1}^{\alpha_1} \mathbf{c}_{n_4}^{\alpha_4}}_{s}
\ee
\be
\label{eq:2p}
\underbrace{\mathbf{c}_{n_1}^{\alpha} \mathbf{c}_{n_2}^{\beta}}_s= \overline{ \frac{1}{L}\sum_\ell (-1)^{s\ell}\mathbf{c}_{\ell+n_1}^{\alpha} \mathbf{c}_{\ell+n_2}^{\beta }}
\ee
and $ \mathcal A_{\{n\}}^{\{\alpha\}}$ is the remaining contribution.
Using \eref{eq:Bog} we can explicitly carry out the time average and the sum over $\ell$ in \eref{eq:2p}. We obtain 
\be
\label{eq:2pmomentumrep}
\underbrace{\mathbf{c}_{n_1}^{\alpha} \mathbf{c}_{n_2}^{\beta}}_s= \frac{1}{L}\sum_{k} e^{- i \alpha (n_2-n_1) k}e^{i n_2 s\pi}U(k)^{\alpha}_{\gamma}U({\alpha\beta}\bar k_s)^{\beta}_{\bar\gamma}\mathbf{b}_\gamma(k)\mathbf{b}_{\bar\gamma}({\alpha\beta}\bar k_s)\, ,
\ee
where $\bar \alpha = -\alpha$ and we defined
\be
k_s = k +\pi s\,.
\ee
Analogously, \eref{eq:WickJWls} reads as
\bea
\label{eq:averaged}
\fl\qquad\frac{1}{L}\overline{\sum_j \mathbf{c}_{j+n_1}^{\alpha_1} \mathbf{c}_{j+n_2}^{\alpha_2} \mathbf{c}_{j+n_3}^{\alpha_3} \mathbf{c}_{j+n_4}^{\alpha_4}}=\nn
\fl\qquad=\frac{1}{L^2}\sum_{\{k_i\}} \prod_{j=1}^4\left( e^{i \alpha_j n_j k_j} U(k_j)^{\alpha_j}_{\beta_j}\right)\mathbf{b}_{\beta_1}(k_1)\mathbf{b}_{\beta_2}(k_2)\mathbf{b}_{\beta_3}(k_3)\mathbf{b}_{\beta_4}(k_4)\nn
\fl\qquad\qquad\qquad\qquad\times\delta_{\alpha_1 k_1+\alpha_2 k_2+\alpha_3 k_3+\alpha_4 k_4}\delta_{\beta_1\varepsilon_{k_1}+\beta_2\varepsilon_{k_2}+\beta_3\varepsilon_{k_3}+\beta_4\varepsilon_{k_4}}\,.
\eea
In order to compute the sums over the momenta it is necessary to solve the constraints given by the delta functions, \emph{i.e.} 
\bea
\alpha_1 k_1+\alpha_2 k_2+\alpha_3 k_3+\alpha_4 k_4=0\nn
\beta_1\varepsilon_{k_1}+\beta_2\varepsilon_{k_2}+\beta_3\varepsilon_{k_3}+\beta_4\varepsilon_{k_4}=0\,.\label{eq:EnergyConstraint}
\eea
Some of the solutions to these equations can be found by requiring the terms of \eref{eq:EnergyConstraint}  to cancel in pairs.
This would give 
\be
\label{eq:solution}
\fl\qquad\qquad \delta_{\alpha_1 k_1+\alpha_2 k_2+\alpha_3 k_3+\alpha_4 k_4}\delta_{\beta_1\varepsilon_{k_1}+\beta_2\varepsilon_{k_2}+\beta_3\varepsilon_{k_3}+\beta_4\varepsilon_{k_4}}=\sum_{s=0,1}\Delta^s_1+\Delta^s_2+\Delta^s_3\,,
\ee
with 
\bea
\Delta^s_1=\delta_{\beta_1,\bar\beta_2}\delta_{\beta_3,\bar\beta_4}\delta_{\bar\alpha_1 k_1,\alpha_1 k_{2,s}}\delta_{\bar\alpha_3 k_3,\alpha_4 k_{4,s}}\\
\Delta^s_2=\delta_{\beta_1,\bar\beta_3}\delta_{\beta_2,\bar\beta_4}\delta_{\bar\alpha_1 k_1,\alpha_3 k_{3,s}}\delta_{\bar\alpha_2 k_2,\alpha_4 k_{4,s}}\\
\Delta^s_3=\delta_{\beta_1,\bar\beta_4}\delta_{\beta_2,\bar\beta_3}\delta_{\bar\alpha_1 k_1,\alpha_4 k_{4,s}}\delta_{\bar\alpha_2 k_2,\alpha_3 k_{3,s}}\,.
\eea
For a generic dispersion relation it is reasonable to expect these solutions to be the only possible. For the specific dispersion considered, \eref{a:disp}, 
equations \eref{eq:EnergyConstraint} admit other solutions in the thermodynamic limit. We call these \emph{anomalous solutions} because they strongly depend on the precise form of the dispersion relation. We now show that the term $ \mathcal A_{\{n\}}^{\{\alpha\}}$ in \eref{eq:WickJWls} is exactly the contribution arising from these solutions, \emph{i.e.} 
\bea
\fl \qquad \mathcal F_{\{n\}}^{\{\alpha\}}=\frac{1}{L^2}\sum_{k_1, k_2, k_3, k_4} \sum_{s=0,1} \prod_{j=1}^4\left( e^{i \alpha_j n_j k_j} U(k_j)^{\alpha_j}_{\beta_j}\right)\mathbf{b}_{\beta_1}(k_1)\mathbf{b}_{\beta_2}(k_2)\mathbf{b}_{\beta_3}(k_3)\mathbf{b}_{\beta_4}(k_4)\nn
\qquad\qquad\qquad\qquad\qquad\qquad \times\{\Delta^s_1+\Delta^s_2+\Delta^s_3\}\,.
\eea
We stress that the operator  $ \mathcal A_{\{n\}}^{\{\alpha\}}$ will be nonzero only in the thermodynamic limit. 

We consider for example the term containing $\Delta_2$. We have 
\bea\label{a:F}
\fl\frac{1}{L^2}\sum_{k_1, k_2, k_3, k_4} \sum_{s=0,1}\prod_{j=1}^4\left( e^{i \alpha_j n_j k_j} U(k_j)^{\alpha_j}_{\beta_j}\right)
\mathbf{b}_{\beta_1}(k_1)\mathbf{b}_{\beta_2}(k_2)\mathbf{b}_{\beta_3}(k_3)\mathbf{b}_{\beta_4}(k_4)\Delta^s_2\,\nn
\fl\quad=\frac{1}{L^2}\sum_{p,q}\sum_{s=0,1}e^{i \alpha_1 (n_1-n_3) p + i \alpha_2 (n_2-n_4) q}e^{i n_3 s \pi}e^{i n_4 s \pi}U(p)^{\alpha_1}_{\beta_1}U(q)^{\alpha_2}_{\beta_2}U({\alpha_1\alpha_3}\bar p_{s})^{\alpha_3}_{\bar\beta_1}U({\alpha_2\alpha_4}\bar q_{s})^{\alpha_4}_{\bar\beta_2}\nn
\fl\qquad\qquad\qquad\qquad\qquad\times\mathbf{b}_{\beta_1}(p)\mathbf{b}_{\beta_2}(q)\mathbf{b}_{\bar\beta_1}({\alpha_1\alpha_3}\bar p_{s})\mathbf{b}_{\bar\beta_2}({\alpha_2\alpha_4} \bar q_{s})\nn
\fl \quad =-\sum_{s=0,1}\underbrace{\mathbf{c}_{n_1}^{\alpha_1} \mathbf{c}_{n_3}^{\alpha_3}}_s\underbrace{\mathbf{c}_{n_2}^{\alpha_2} \mathbf{c}_{n_4}^{\alpha_4}}_s+\,O\bigl(L^{-1}\bigr)\, ,
\eea
where we used the commutation relations of the $\{\mathbf{b}_{\beta}(k)\}$ in the last step.
Although the terms $\mathcal F_{\{n\}}^{\{\alpha\}}$ and $\mathcal A_{\{n\}}^{\{\alpha\}}$ are in fact multiplied by $L$ in the time average of \eref{eq:V}, the possible corrections $O(L^{-1})$ in \eref{a:F} (which would result in corrections $O(L^0)$ in the effective Hamiltonian) are locally irrelevant, because their density approaches zero in the thermodynamic limit.   

We obtain analogous results for $\Delta_1$ and $\Delta_3$, that is to say  \eref{eq:WickJWls}.

\begin{remark}
We point out that for other dispersion relations (still with the properties $\varepsilon_k=\varepsilon_{k+\pi}$ and $\varepsilon_k\neq\varepsilon_{k+\pi/n}$ for generic $k$ and $n>1$) the anomalous terms could be factorised as well. 
Generally in such situations the factors have a very simple time dependence, \emph{e.g.} a single oscillation frequency.  
As a consequence, relaxation is ruled out. 
\end{remark} 

\section{Towards a mean-field description}\label{a:MF}
In this appendix we prove the Lemmas of Section \ref{s:mf}.

\begin{lemma}\label{L:opnorm}
If $\mathcal O\in \mathcal E$, the operator norm (\emph{i.e.} the maximal eigenvalue in absolute value) of $\mathcal O/L$ is bounded.
\end{lemma}

\begin{proof}
The proof is straightforward. Let us expand $\mathcal O/L$ as in \eref{eq:opform}:
\be
\frac{\mathcal O}{L}=\sum_{j=1}^N \frac{1}{L^{n_j}}\mathcal O_1^{(j)}\cdots \mathcal O_{n_j}^{(j)}
\ee
where $\mathcal O_m^{(j)}$ have local densities, that is to say, they can be written as follows
\be
\mathcal O_m^{(j)}=\sum_\ell o_{m;\ell}^{(j)} \, ,
\ee
with $o_{m;\ell}^{(j)}$ local operators. We immediately find the chain of inequalities
\be
\fl \parallel \frac{\mathcal O}{L}\parallel\leq \sum_{j=1}^N\frac{1}{L^{n_j}} \parallel \mathcal O^{(j)}_1\cdots \mathcal O^{(j)}_{n_j} \parallel\leq\sum_{j=1}^N \parallel \frac{\mathcal O_1^{(j)}}{L}\parallel\cdots \parallel \frac{\mathcal O_{n_j}^{(j)}}{L}\parallel\leq  \sum_{j=1}^N (\max_{m,\ell}\parallel o_{m;\ell}^{(j)} \parallel)^{n_j}\, .
\ee 
The right hand side is clearly $O(L^0)$ because $N$ and $n_j$ are finite by definition and $o_{m;\ell}^{(j)}$ are local.
\end{proof}

\begin{lemma}\label{L:algebra}
If $\mathcal O,\tilde{ \mathcal O}\in \mathcal E$, then $[\mathcal O,\tilde{\mathcal O}]\in \mathcal E$ as well. 
\end{lemma}

\begin{proof}
Without loss of generality, we can restrict to two single terms of the expansions \eref{eq:opform} of $\mathcal O$ and $\tilde{\mathcal O}$. The commutator of the two terms reads as
\begin{eqnarray}
\fl\quad [\frac{1}{L^{n_i-1}}\mathcal O_1^{(i)}\cdots \mathcal O_{n_i}^{(i)},\frac{1}{L^{n_j-1}}\tilde{\mathcal O}_1^{(j)}\cdots \tilde{\mathcal O}_{n_j}^{(j)}]=\nonumber \\
\fl\qquad\frac{1}{L^{n_j+n_j-2}}\sum_{k,p}\mathcal O_1^{(i)}\cdots \mathcal O_{k-1}^{(i)}\tilde{\mathcal O}_1^{(j)}\cdots \tilde{\mathcal O}_{p-1}^{(j)}[\mathcal O_k^{(i)},\tilde{\mathcal O}_p^{(j)}]\mathcal O_{k+1}^{(i)} \cdots\mathcal O_{n_i}^{(i)}\tilde{\mathcal O}_{p+1}^{(j)} \cdots\tilde{\mathcal O}_{n_j}^{(j)}\, .
\end{eqnarray}
Since $[\mathcal O_k^{(i)},\tilde{\mathcal O}_p^{(j)}]$ have local densities (the commutator of two local operators is nonzero only if there is a region on which they both act nontrivially; in addition, its range is smaller than the sum of the ranges of the two operators), the number of extensive operators exceeds by one the exponent of $1/L$. Thus, $[\mathcal O,\tilde{\mathcal O}]\in \mathcal E$.
\end{proof}

\begin{lemma}\label{a:L:1}
(viz. Lemma \ref{L:1}) Let $\mathcal O\in \mathcal E$ and $\ket{\Psi}$ a state with cluster decomposition properties. The expectation value of $\mathcal O/L$ in $\ket{\Psi}$ can be reduced to the expectation values of the local translation invariant operators it consists of:
\be\label{eq:fact}
\lim_{L\rightarrow\infty }\braket{\Psi|\frac{H_1}{L}\cdots \frac{H_n}{L}|\Psi}=\lim_{L\rightarrow\infty }\prod_j \frac{\braket{\Psi|H_j|\Psi}}{L}\, .
\ee
\end{lemma}

\begin{proof}
Let us consider a term \eref{eq:form} of the expansion \eref{eq:opform}. Its expectation value (per unit of length) is given by
\be\label{eq:locexp}
\braket{\Psi|\frac{H_1}{L}\cdots \frac{H_n}{L}|\Psi}=\frac{1}{L^n}\sum_{\ell_1,\dots,\ell_n}\braket{\Psi|h_{1,\ell_1}\cdots h_{n,\ell_n}|\Psi}\, ,
\ee 
where $h_{j,\ell_j}$ are local operators acting nontrivially only around $\ell_j$ and such that 
\be
H_j=\sum_\ell h_{j,\ell_j}\, .
\ee
By cluster decomposition we have
\be\label{eq:rhs}
\fl \sum_{\ell_1,\dots,\ell_n\atop |\ell_j-\ell_{j'}|>\xi\gg 1\ (\forall j\neq j')}\braket{\Psi|\frac{h_{1,\ell_1}}{L}\cdots \frac{h_{n,\ell_n}}{L}|\Psi}= \!\!\!\!\!\!\!\!\!\!\!\!\sum_{\ell_1,\dots,\ell_n\atop |\ell_j-\ell_{j'}|>\xi\gg 1\ (\forall j\neq j')}\!\!\!\!\!\!\!\!\!\!\!\!\frac{\braket{\Psi|h_{1,\ell_1}|\Psi}}{L}\cdots\frac{\braket{\Psi| h_{n,\ell_n}|\Psi}}{L}+f(\xi,L)\, ,
\ee
where
\be
\lim_{\xi\rightarrow\infty}\lim_{L\rightarrow\infty}f(\xi,L)=0\, .
\ee
The difference between \eref{eq:locexp} and the left hand side of \eref{eq:rhs} can be bounded from above as follows
\be
\fl\qquad\Bigl|\sum_{\ell_1,\dots,\ell_n\atop |\ell_j-\ell_{j'}|\leq \xi\ (\exists j\neq j')}\braket{\Psi|\frac{h_{1,\ell_1}}{L}\cdots \frac{h_{n,\ell_n}}{L}|\Psi}\Bigr|\leq {n \choose 2}\frac{\xi}{L} \max_{\{\ell\}}\Bigl|\braket{\Psi|h_{1,\ell_1}\cdots h_{n,\ell_n}|\Psi}\Bigr| \rightarrow 0\, .
\ee
Analogously
\be
\fl\qquad\Bigl|\sum_{\ell_1,\dots,\ell_n\atop |\ell_j-\ell_{j'}|\leq \xi\ (\exists j\neq j')}\!\!\!\!\!\!\!\!\!\!\!\!\frac{\braket{\Psi|h_{1,\ell_1}|\Psi}}{L}\cdots\frac{\braket{\Psi| h_{n,\ell_n}|\Psi}}{L}\Bigr|\leq {n \choose 2} \frac{\xi}{L} \max_{\{\ell\}}\Bigl|\prod_{j=1}^n\braket{\Psi|h_{j,\ell_j}|\Psi}\Bigr| \rightarrow 0\ ,
\ee
so that 
\be
\Bigl|\braket{\Psi|\frac{H_1}{L}\cdots \frac{H_n}{L}|\Psi}-\prod_j \frac{\braket{\Psi|H_j|\Psi}}{L}\Bigr|\leq |f(\xi,L)|+O(1/L)
\ee
Being $\xi$ arbitrary, we can take the limit $\lim_{\xi\rightarrow\infty}\lim_{L\rightarrow\infty}$, obtaining \eref{eq:fact}.
\end{proof}

\begin{lemma}\label{a:L:2}
If the state $\ket{\Psi_0}$ has cluster decomposition properties and $\mathcal O\in\mathcal E$, the mean-field Hamiltonian defined in \sref{s:mf} satisfies the following identity:
\be\label{eq:lemma2}
\fl \qquad\lim_{L\rightarrow\infty}\braket{\Psi_0|\bar U^\dag(t)[H^{\Psi_0}_{\rm MF}(t),\frac{\mathcal O}{L}]\bar U(t)|\Psi_0}=\lim_{L\rightarrow\infty}\braket{\Psi_0|\bar U^\dag(t)[H,\frac{\mathcal O}{L}]\bar U(t)|\Psi_0}\, ,
\ee
where $\bar U$ was defined in \eref{eq:U}.
\end{lemma}

\begin{proof}
Let us consider a generic term \eref{eq:form} of the expansion \eref{eq:opform} of $H$
\be
\tilde H=\frac{1}{L^{n-1}}H_1\cdots H_n\, .
\ee 
The corresponding term \eref{eq:tran} of the mean-field Hamiltonian \eref{eq:tran} is given by
\be
\tilde H_{\rm MF}^{\Psi_0}(t)=\sum_{\ell=1}^n \prod_{j\neq \ell} \frac{\braket{\Psi_0|\bar U^\dag(t)H_j\bar U(t)|\Psi_0}}{L}H_\ell\, .
\ee
By taking the commutators with $\mathcal O$ we find
\begin{eqnarray}\label{Htilde}
[\tilde H,\frac{\mathcal O}{L}]=\sum_{\ell=1}^n \prod_{j=1}^{\ell-1} \frac{H_j}{L}\frac{[H_{\ell},\mathcal O]}{L}\prod_{j=\ell+1}^n \frac{H_j}{L}\\
\label{HtildeMF}{}[\tilde H_{\rm MF}^{\Psi_0}(t),\frac{\mathcal O}{L}]=\sum_{\ell=1}^n \prod_{j\neq \ell} \frac{\braket{\Psi_0|\bar U^\dag(t)H_j\bar U(t)|\Psi_0}}{L}\frac{[H_\ell,\mathcal O]}{L}\, .
\end{eqnarray}
Because $\ket{\Psi_0}$ has cluster decomposition properties and the mean-field Hamiltonian is local at any time, the state $\bar U(t)\ket{\Psi_0}$ has cluster decomposition properties as well (the only difference with respect to $\ket{\Psi_0}$ is that the function $f$ of \eref{eq:rhs} is now time dependent).
Finally, by Lemma \ref{a:L:1}, in the thermodynamic limit the expectation values of \eref{Htilde} and \eref{HtildeMF} in the state $\bar U(t)\ket{\Psi_0}$ are identical, that is to say \eref{eq:lemma2}.
\end{proof} 

\begin{lemma}\label{a:L:3}
Let $\ket{\Psi_0}$ be a translation invariant state with cluster decomposition properties and $H,\mathcal O\in \mathcal E$. The time derivatives of the expectation value of $\mathcal O/L$ in the state evolving with $H_{\rm MF}^{\Psi_0}(t)$ fulfil
\be
\label{a:eq:derivs}
\fl\quad\frac{\textrm{d}^{\,n}}{\textrm{dt}^{\,n}}\lim_{L\rightarrow\infty}\braket{\Psi_0|\bar U^\dag(t)\frac{\mathcal O}{L}\bar U(t) |\Psi_0}=i^n\lim_{L\rightarrow\infty}\braket{\Psi_0|\bar U^\dag(t)\underbrace{\bigl[H,[H,...[H}_{n},\frac{\mathcal O}{L}\bigr]...]]\bar U(t) |\Psi_0}
\ee
\end{lemma}
\begin{proof}
We proceed by induction. First of all we see that for $n=0$ the property is trivially satisfied; let then the property be true for $n$, we have 
\bea
\label{a:eq:derivs2}
\fl\frac{\textrm{d}^{\,n+1}}{{\textrm{d}t}^{\,n+1}}\lim_{L\rightarrow\infty}\braket{\Psi_0|\bar U^\dag(t)\frac{\mathcal O}{L}\bar U(t) |\Psi_0}=i^n\frac{\textrm{d}}{\textrm{d}t}\lim_{L\rightarrow\infty}\braket{\Psi_0|\bar U^\dag(t)\underbrace{\bigl[H,[H,...[H}_{n},\frac{\mathcal O}{L}\bigr]...]]\bar U(t) |\Psi_0}\nn
\fl\qquad\qquad\qquad\qquad\qquad\,\,\, = i^{n+1}\lim_{L\rightarrow\infty}\braket{\Psi_0|\bar U^\dag(t)[H_{MF}(t),\underbrace{\bigl[H,[H,...[H}_{n},\frac{\mathcal O}{L}\bigr]...]]\bar U(t) |\Psi_0}\nn
\fl\qquad\qquad\qquad\qquad\qquad\,\,\, = i^{n+1}\lim_{L\rightarrow\infty}\braket{\Psi_0|\bar U^\dag(t)\underbrace{\bigl[H,[H,...[H}_{n},\frac{\mathcal O}{L}\bigr]...]]\bar U(t) |\Psi_0}\,.
\eea
In the second step we used Lemma \ref{L:algebra} and Lemma \ref{a:L:2}. This concludes the proof.
\end{proof}
\begin{lemma}\label{a:T:1}
(viz. Lemma \ref{T:1})
Let $\ket{\Psi_0}$ be a translation invariant state with cluster decomposition properties and $H,\mathcal O\in \mathcal E$. 
Let the expectation value of $\mathcal O$ in the state that time evolves with $H_{\rm MF}^{\Psi_0}(t)$ be an analytic function of $t$ in the strip $|\mathrm{Im}[t]|<r$, with $r$ a nonzero constant.
In the thermodynamic limit, the time evolution with $H$ can be replaced by the time evolution with the mean-field Hamiltonian:
\be\label{eq:rep}
\lim_{L\rightarrow\infty}\braket{\Psi_0|e^{i H t}\frac{\mathcal O}{L} e^{-i H t}|\Psi_0}=\lim_{L\rightarrow\infty}\braket{\Psi_0|\bar U^\dag(t)\frac{\mathcal O}{L}\bar U(t) |\Psi_0}\, .
\ee

\end{lemma}

\begin{proof}
We define
\be
f(t,s)=\lim_{L\rightarrow\infty}\braket{\Psi_0|\bar U^\dag(t)e^{i H s}\frac{\mathcal O}{L}e^{-i H s}\bar U(t)|\Psi_0}\, .
\ee
By Lemma \ref{a:L:3} we have
\be\label{eq:TaylorC}
\frac{\partial^n}{\partial t^n} f(t,0)=\frac{\partial^n}{\partial s^n}\Bigr|_{s=0}f(t,s)\, ,
\ee
indeed 
\be
\frac{\partial^n}{\partial t^n}\Bigr|_{t=0} e^{i H t}\frac{\mathcal O}{L} e^{-i H t}=i^n \underbrace{\bigl[H,[H,...[H}_{n},\frac{\mathcal O}{L}\bigr]...]]\, .
\ee
By assumption, $f(t,0)$ (which corresponds to the time evolution with the mean-filed Hamiltonian) is analytic in the strip $|\mathrm{Im}[t]|<r$, so the convergence radius of the Taylor expansion at $t=0$ is larger than or equal to $r$.
Thus we have
\be
\label{eq:step0}
\fl\quad\, f(\tau,0)=\sum_n \frac{\tau^n}{n!}\frac{\partial^n}{\partial t^n}\Bigr|_{t=0}f(t,0)=\sum_n \frac{\tau^n}{n!}\frac{\partial^n}{\partial t^n}\Bigr|_{t=0}f(0,t)=f(0,\tau)\,,\qquad |\tau|<r\, .
\ee
Let us call $t_\ast$ a time such that $f(t,0)=f(0,t)$ for any $0\leq t< t_\ast$.  
As before, the function $f(t+\tau,0)$ is analytic in the strip $|\mathrm{Im}[\tau]|<r$, so we have
\be
\fl\quad f(t+\tau,0)=\sum_n \frac{\tau^n}{n!}\frac{\partial^n}{\partial t^n}f(t,0)=\sum_n \frac{\tau^n}{n!}\frac{\partial^n}{\partial t^n}f(0,t)=f(0,t+\tau)\,,\qquad |\tau|<r\, .
\ee
That is to say
\be 
\fl\qquad\qquad f(t,0)=f(0,t) \quad \forall t< t_\ast\quad\Longrightarrow\quad  f(t,0)=f(0,t) \quad \forall t< t_\ast+\tau\, .
\ee
Since $\tau$ is finite and \eref{eq:step0} holds, we conclude 
\be
f(t,0)=f(0,t) \qquad\forall\,t\,,
\ee
which is exactly \eref{eq:rep}. 
\end{proof}

\begin{corollary}\label{a:C:1}
(viz. Corollary \ref{C:1})
Lemma \ref{a:T:1} holds true in particular for local operators.
\end{corollary}

\begin{proof}
By translation invariance, the expectation value of any local operator $\mathcal O$ is equal to the expectation value per unit of length of the operator $\mathcal O_\ast\in\mathcal E$, obtained by shifting $\mathcal O$ along the chain and summing all the ($L$) terms. 
\end{proof}

\begin{corollary}\label{a:C:2}
(viz. Corollary \ref{C:2})
Let $\ket{\Psi_0}$ a translation invariant state with cluster decomposition properties and $H\in \mathcal E$. 
In the thermodynamic limit, the time evolution of the reduced density matrix (RDM) of some spin block $S$ is equal to the RDM in the state that time evolves with the mean-field Hamiltonian: 
\be\label{eq:RDM}
\rho_S(t)=\tr_{\bar S}[e^{-i H t}\ket{\Psi_0}\bra{\Psi_0}e^{i H t}]=\tr_{\bar S}[\bar U(t)\ket{\Psi_0}\bra{\Psi_0}\bar U^\dag(t)]\, .
\ee
\end{corollary}

\begin{proof}
This is a direct consequence of Corollary \ref{a:C:1}. 
\end{proof}

\begin{corollary}\label{a:C:3}
(viz. Corollary \ref{C:3})
Let $H\in\mathcal E$ and $\ket{\Psi}$ a state with cluster decomposition properties. 
If $\ket{\Psi}$ is an excited state of the corresponding mean-field Hamiltonian $H_{\rm MF}^{\Psi}$
\be
H_{\rm MF}^{\Psi}\ket{\Psi}=E_\Psi\ket{\Psi}\, ,
\ee
the expectation value of local observables in $e^{-i H t}\ket{\Psi}$ is independent of time. Therefore, $\ket{\Psi}$ behaves locally as an excited state of $H$. 

The reverse is also true. If an excited state of $H$ is locally equivalent to a state with cluster decomposition properties, then the latter is an excited state of the corresponding mean-field Hamiltonian.
\end{corollary}

\begin{proof}
Clearly the mean-field Hamiltonian $H_{\rm MF}^{\Psi}$ is the solution of \eref{eq:tran}. Being $e^{-i H_{\rm MF}^{\Psi}t}\ket{\Psi}\propto \ket{\Psi}$, by Corollary \ref{a:C:1} the expectation value of local observables is independent of time. The reverse holds true for analogous reasons. 
\end{proof}

\section{Self-consistency check of condition \eref{eq:condA}}\label{a:self}%

Here we show that neglecting the \emph{anomalous term} $L\mathcal A_{\{n\}}^{\{\alpha\}}$ (\emph{cf}. \eref{eq:WickJWls}) in the time averaged Hamiltonian is a self-consistent approximation. To this aim, we consider the time evolution of a Slater determinant $\ket{\Psi_0}$ under the Hamiltonian $\tilde{H}$, obtained from $H$ by removing $L\mathcal A_{\{n\}}^{\{\alpha\}}$.
In \Sref{s:mf} and Appendix \ref{a:MF} we proved that, as long as $\mathcal O$ is a local operator (but the class of allowed operators is in fact larger), $e^{-i \tilde{H}t}$ can be replaced by the mean-field time evolution operator $\bar U(t)$ \eref{eq:U}. 
 Here we show that inserting $L\mathcal A_{\{n\}}^{\{\alpha\}}$ back at time $t$ does not change the expectation value of local observables.  
In other words we are going to prove
\be
\label{a:eq:selfcons}
\lim_{L\rightarrow\infty}L\braket{\Psi_0|\bar U^\dag(t)[\mathcal A_{\{n\}}^{\{\alpha\}},\mathcal O]\bar U(t)|\Psi_0}= 0\qquad\qquad\forall t\,,
\ee 
where $\mathcal O$ is a generic local operator. 

Using the notations of Appendix \ref{a:Wick}, any local operator can be written as a linear combination of operators of the form
\be
\fl\qquad\mathcal O =  \mathbf{c}_{\ell_1}^{\gamma_1}\cdots\mathbf{c}_{\ell_{n}}^{\gamma_{n}}=\frac{1}{L^{n/2}}\sum_{\{p_i\}}{\mathcal F_{n}}^{\{\gamma_i\}}_{\{\sigma_i\}}(\{\ell_i\}|\{p_i\})\mathbf{b}_{\sigma_1}(p_1)\cdots\mathbf{b}_{\sigma_{n}}(p_{n})\,,
\ee
where
\be
{\mathcal F_{n}}^{\{\alpha_i\}}_{\{\beta_i\}}(\{j_i\}|\{p_i\})\equiv\prod_{i=1}^{n}\left( e^{i \alpha_i j_i p_i}U(p_i)^{\alpha_i}_{\beta_i}\right)\,.
\ee 
If $n$ is odd then \eref{a:eq:selfcons} is trivially satisfied because $\bar U(t)\ket{\Psi_0}$ is a Slater determinant by assumption and hence the expectation value of an odd number of fermions vanishes. 
We therefore focus on the case $n=2m$. The anomalous term $\mathcal A_{\{n\}}^{\{\alpha\}}$ of \eref{eq:WickJWls} can be written as follows
\be
\fl\mathcal A_{\{n\}}^{\{\alpha\}}=\frac{1}{L^2}\sum_{k_1, k_2} \sum_{\bar k_3,\bar k_4}{\mathcal F_{2}}^{\{\alpha_i\}}_{\{\beta_i\}}(\{j_i\}|\{k_i\}){\mathcal F_{2}}^{\{\alpha_i\}}_{\{\beta_i\}}(\{j_i\}|\{\bar k_i\})
\mathbf{b}_{\beta_1}(k_1)\mathbf{b}_{\beta_2}(k_2)\mathbf{b}_{\beta_3}(\bar k_3)\mathbf{b}_{\beta_4}(\bar k_4)
\ee
where $\bar k_3$ and $\bar k_4$ are the \emph{anomalous solutions} of system \eref{eq:EnergyConstraint}, \emph{i.e.} they are implicit functions of $k_1$ and $k_2$ defined by the system \eref{eq:EnergyConstraint} and in addition fulfilling
\be
\label{eq:conditionskbar}
k_{1,2}\pm\bar k_{3,4}\neq0\,,\qquad\qquad\bar k_{3}\pm\bar k_{4}\neq0\,,
\ee
almost everywhere.
Since $\ket{\tilde{\Psi}_t}=\bar U(t)\ket{\Psi_0}$ is a Slater determinant, we can use the Wick theorem to compute expectation values.  
We then have 
\be\label{eq:scaling3}
\fl\qquad L\braket{\tilde{\Psi}_t|\mathcal A_{\{n\}}^{\{\alpha\}} \mathcal O|\tilde{\Psi}_t}=L\braket{\tilde{\Psi}_t|\mathcal A_{\{n\}}^{\{\alpha\}}|\tilde{\Psi}_t}\braket{\tilde{\Psi}_t|\mathcal O|\tilde{\Psi}_t}+\mathcal C_2[\mathcal A_{\{n\}}^{\{\alpha\}} \mathcal O]_t +\mathcal C_4[\mathcal A_{\{n\}}^{\{\alpha\}} \mathcal O]_t
\ee   
where $\mathcal C_2[\mathcal A_{\{n\}}^{\{\alpha\}} \mathcal O]_t$ contains terms in which two of the $\mathbf{b}$'s in $\mathcal A_{\{n\}}^{\{\alpha\}}$ are contracted together and the other two are contracted with two $\mathbf{b}$'s in $\mathcal O$;  $\mathcal C_4[\mathcal A_{\{n\}}^{\{\alpha\}} \mathcal O]_t$ contains all the terms in which any $\mathbf{b}$ in $\mathcal A_{\{n\}}^{\{\alpha\}}$ is contracted with a $\mathbf{b}$ in $\mathcal O$.

According to the definition of $\mathcal A_{\{n\}}^{\{\alpha\}}$, any Wick contraction among $\mathbf{b}$'s in it gives zero (because of \eref{eq:conditionskbar}), hence the only non zero contribution to \eref{eq:scaling3} arises from $\mathcal C_4[\mathcal A_{\{n\}}^{\{\alpha\}} \mathcal O]_t$. To conclude the proof we will show that the terms in $\,\mathcal C_4[\mathcal A_{\{n\}}^{\{\alpha\}} \mathcal O]_t$ scale as $O(L^{-1})$ and in the thermodynamic limit their contribution can thus be neglected. To this end it is sufficient to consider a typical element of $\,\mathcal C_4[\mathcal A_{\{n\}}^{\{\alpha\}} \mathcal O]_t$
\be
\fl \frac{1}{L^{m+1}}
G_{\mathbf j,\mathbf k,\boldsymbol \ell,\mathbf p}^{\boldsymbol \alpha, \boldsymbol \beta,\boldsymbol \gamma,\boldsymbol \sigma}
\contraction{\mathbf{b}_{\beta_1}(k_1)\mathbf{b}_{\beta_2}(k_2)\mathbf{b}_{\beta_3}(\bar k_3)}{\mathbf{b}}{_{\beta_4}(\bar k_4)}
{\mathbf{b}} \contraction[2ex]{}{\mathbf{b}}{_{\beta_1}(k_1)\mathbf{b}_{\beta_2}(k_2)\mathbf{b}_{\beta_3}(\bar k_3)\mathbf{b}_{\beta_4}(\bar k_4)
\mathbf{b}_{\sigma_1}(p_1)\mathbf{b}_{\sigma_2}(p_2)
\mathbf{b}_{\sigma_3}(p_3)\cdots}{\mathbf{b}}\contraction[3ex]{\mathbf{b}_{\beta_1}(k_1)}{\mathbf{b}}{_{\beta_2}(k_2)\mathbf{b}_{\beta_3}(\bar k_3)\mathbf{b}_{\beta_4}(\bar k_4)
\mathbf{b}_{\sigma_1}(p_1)\mathbf{b}_{\sigma_2}(p_2)
}{\mathbf{b}}
\contraction[4ex]{\mathbf{b}_{\beta_1}(k_1)\mathbf{b}_{\beta_2}(k_2)}{\mathbf{b}}{_{\beta_3}(\bar k_3)\mathbf{b}_{\beta_4}(\bar k_4)
\mathbf{b}_{\sigma_1}(p_1)}
{\mathbf{b}}
\mathbf{b}_{\beta_1}(k_1)\mathbf{b}_{\beta_2}(k_2)\mathbf{b}_{\beta_3}(\bar k_3)\mathbf{b}_{\beta_4}(\bar k_4)
\mathbf{b}_{\sigma_1}(p_1)\mathbf{b}_{\sigma_2}(p_2)
\mathbf{b}_{\sigma_3}(p_3)\cdots\mathbf{b}_{\sigma_{2m}}(p_{2m})\, ,
\ee 
where sums over the momenta $\{p_i\}, \{k_i\}$ and the indices are understood, and we defined
\be
\contraction{}{\mathbf{b}}{_{\alpha}(p)}{\mathbf{b}}\mathbf{b}_{\alpha}(p)\mathbf{b}_{\beta}(q)=\braket{\tilde{\Psi}_t|\mathbf{b}_{\alpha}(p)\mathbf{b}_{\beta}(q)|\tilde{\Psi}_t}
\ee
\be
\fl \qquad G_{\mathbf j,\mathbf k,\boldsymbol \ell,\mathbf p}^{\boldsymbol \alpha, \boldsymbol \beta,\boldsymbol \gamma,\boldsymbol \sigma}\equiv{\mathcal F_{2}}^{\{\alpha_i\}}_{\{\beta_i\}}(\{j_i\}|\{k_i\}){\mathcal F_{2}}^{\{\alpha_i\}}_{\{\beta_i\}}(\{j_i\}|\{\bar k_i\}){\mathcal F_{n}}^{\{\gamma_i\}}_{\{\sigma_i\}}(\{\ell_i\}|\{p_i\})\,.
\ee
The $2m+2$ sums over the momenta are reduced to $m$ by the Kronecker deltas arising from the Wick contractions. 
Because the number of factors $L^{-1}$ exceeds by one the number of sums, the term turns out to be $O(L^{-1})$. The validity of equation \eref{a:eq:selfcons} is then established.

\section{Additional properties of \eref{eq:TFICNI}}\label{a:GS}
\subsection{Ground state and maximal energy state}
Here we show that the TFIC ground state $\ket{\psi_{g_0}}$, with $g_0$ satisfying \eref{eq:noquench1}, is the Slater determinant that minimises the energy \eref{eq:energy} of the Hamiltonian \eref{eq:TFICNI}.

The symbol of the correlation matrix of the most general reflection symmetric one-site shift invariant Slater determinant is
\be\label{a:eq:sym}
\Gamma(k)=n_x(k)\sigma^x+n_y(k)\sigma^y+n_z(k)\sigma^z\, ,
\ee
with $n_\alpha(k)$ real functions that satisfy 
\bea
n_x^2+n_y^2+n_z^2=1\label{a:eq:sympure}\\
n_{x,z}(k)=-n_{x,z}(-k)\\
n_y(k)=n_y(-k)\, .
\eea
\Eref{a:eq:sympure} manifests the fact that the initial state is pure, which indeed implies that the eigenvalues of $\Gamma(k)$ are $\pm 1$. The other conditions simply mean that the correlation matrix is a purely imaginary skew-symmetric matrix.  The absence of a term proportional to the identity in \eref{a:eq:sym} is a consequence of reflection symmetry. 

The energy \eref{eq:energy} can be written as follows:
\be
\fl\qquad\varepsilon=\lambda\Bigl(\int_{0}^\pi\frac{\mathrm d k}{\pi}n_y(k)\Bigr)^2-\tilde g\int_{0}^\pi\frac{\mathrm d k}{\pi}n_y(k)+\int_{0}^\pi\frac{\mathrm d k}{\pi} (\cos k\ n_y(k)-\sin k\ n_x(k))\, .
\ee
The minimisation can be worked out by zeroing the variation of the functional
\be
\Phi[n_x,n_y,n_z,\mu]=\varepsilon-\int_{0}^\pi\frac{\mathrm d k}{\pi}\mu(k) (n_x^2(k)+n_y^2(k)+n_z^2(k))
\ee
with respect to its arguments. 
We immediately see that the variation with respect to $n_z$ results in
\be
\mu(k)n_z(k)=0\, .
\ee 
If we assume $\mu(k)\neq 0$, then  $n_z(k)=0$ and we can  enforce \eref{a:eq:sympure} by setting $n_x=-\sin\theta$ and $n_y=\cos\theta$. Instead of working with the functional $\Phi$ we can therefore express the energy as
\be
\fl\quad\varepsilon[\theta]=\lambda\Bigl(\int_{0}^\pi\frac{\mathrm d k}{\pi}\cos\theta_k\Bigr)^2-\tilde g\int_{0}^\pi\frac{\mathrm d k}{\pi}\cos\theta_k+\int_{0}^\pi\frac{\mathrm d k}{\pi} (\cos k\ \cos\theta_k+\sin k\ \sin\theta_k)
\ee
and consider its variation with respect to $\theta_k$. We obtain
\be
 \Bigl[\tilde g-2\lambda\Bigl(\int_{0}^\pi\frac{\mathrm d p}{\pi}\cos\theta_p\Bigr)-\cos k\Bigr]\sin\theta_k+\sin k\ \cos\theta_k=0
\ee
which is solved by
\be\label{a:eq:GS}
e^{i\theta_k}=\frac{g_0-e^{i k}}{\sqrt{1+g_0^2-2 g_0\cos k}}\, ,
\ee
where $g_0$ satisfies \eref{eq:noquench1}.
One of the solutions of \eref{eq:noquench1} minimises the energy. 
The energy is instead maximal for
\be\label{a:eq:AGS}
e^{i\theta_k}=-\frac{g_1-e^{i k}}{\sqrt{1+g_1^2-2 g_1\cos k}}\, ,
\ee
where $g_1$ is a solution of
\be
\tilde g=g_1-2\lambda\Bigl(\int_{0}^\pi\frac{\mathrm d p}{\pi}\frac{g_1-\cos p}{\sqrt{1+g_1^2-2 g_1\cos p}}\Bigr)\, .
\ee
We emphasise that numerical data obtained by exact diagonalization of \eref{eq:TFICNI} for small chains are compatible with the two states of \eref{a:eq:GS} and \eref{a:eq:AGS} being the actual ground state and the maximal energy state, respectively, of \eref{eq:TFICNI} in the thermodynamic limit.  

\subsection{Excited states}
Following Corollary \ref{C:3}, we construct an infinite number of states for which the expectation value of local observables that time evolve under the Hamiltonian \eref{eq:TFICNI} is independent of time. To all intents and purposes, those states can be considered eigenstates of \eref{eq:TFICNI}.

A generic stationary solution of \eref{eq:system} should have $\tilde y'=0$ and 
\be
4(h-\cos k)\phi_k+16(h\cos k  -1)y_k=0\, ,
\ee
namely
\be
-16(h-\cos k)\tr[\Gamma(k)\sigma^ye^{i k\sigma^z}]+16(h\cos k  -1)\tr[\Gamma(k)\sigma^y]=0\, .
\ee
This is solved by
\be\label{a:eq:Gammaexc}
\Gamma_m(k)=-m(k)\frac{(h_m-\cos k)\sigma^y+\sin k\ \sigma^x}{\sqrt{1+h_m^2-2h_m\cos k}}\, ,
\ee
where $m(k)=m(-k)$ and we used $\tilde y'=0$ to remove any term proportional to $\sigma^z$.
The parameter $h_m$ is the solution of the equation
\be\label{a:eq:hm}
h_m=\tilde g+2\lambda\int_{-\pi}^\pi\frac{\mathrm d k}{2\pi}m(k)\frac{h_m-\cos k}{\sqrt{1+h_m^2-2h_m\cos k}}\, .
\ee
We recognise $\Gamma_m(k)$ as the symbol of the correlation matrix of a reflection symmetric (Slater determinant) excited state of the TFIC model with magnetic field $h_m$ \cite{AFC:exc, BKC:exc, KBC:exc, AEFS-B}. Within this interpretation, if $h_m$ was independent of $m$, $m(k_j)=\pm 1$ (for finite chains) would have provided an orthonormal basis of reflection symmetric excited states. We remind the reader that in the thermodynamic limit infinite excited states become locally equivalent and, in sufficiently regular cases, can be characterised by a function $-1\leq m(k)\leq 1$. 
However, in the latter characterisation, which exactly describes the local properties of the excited states, the manifest orthonormality properties of the basis are lost. This problem is even more pronounced in our case, where $h_m$ depends on the excitation, and a naive finite volume regularisation does not produce an orthonormal basis also for $m(k_j)=\pm 1$ (this is not unexpected, being our description valid only in the thermodynamic limit). 
Nevertheless, the noninteracting states with correlation matrices \eref{a:eq:Gammaexc} are locally equivalent to the excited states of \eref{eq:TFICNI}. 
In addition, it seems that they span the entire space of reflection symmetric states. 

We notice that in the thermodynamic limit generally there are infinite $m$ that give rise to the same $h_m$. Indeed we can always add to $m(k)$ an even function $\delta m(k)$ such that $-1\leq m+\delta m\leq 1$ and 
\be
\int_{0}^\pi\frac{\mathrm d k}{\pi}\delta m(k)\frac{h_m-\cos k}{\sqrt{1+h_m^2-2h_m\cos k}}=0\, ,
\ee
which physically means the the excitation must have the same magnetisation per unit of length in the $z$ direction. 

For the sake of completeness we report the expression of the energies of the excited states:
\be\label{a:eq:energyex}
\varepsilon_m=\frac{h_m^2-\tilde g^2}{4\lambda}-\int_{-\pi}^\pi\frac{\mathrm d k}{2\pi} m(k)\frac{h_m\cos k-1}{\sqrt{1+h_m^2-2h_m\cos k}}\,.
\ee

\begin{figure}[tbp]
\begin{center}
\includegraphics[width=0.8\textwidth]{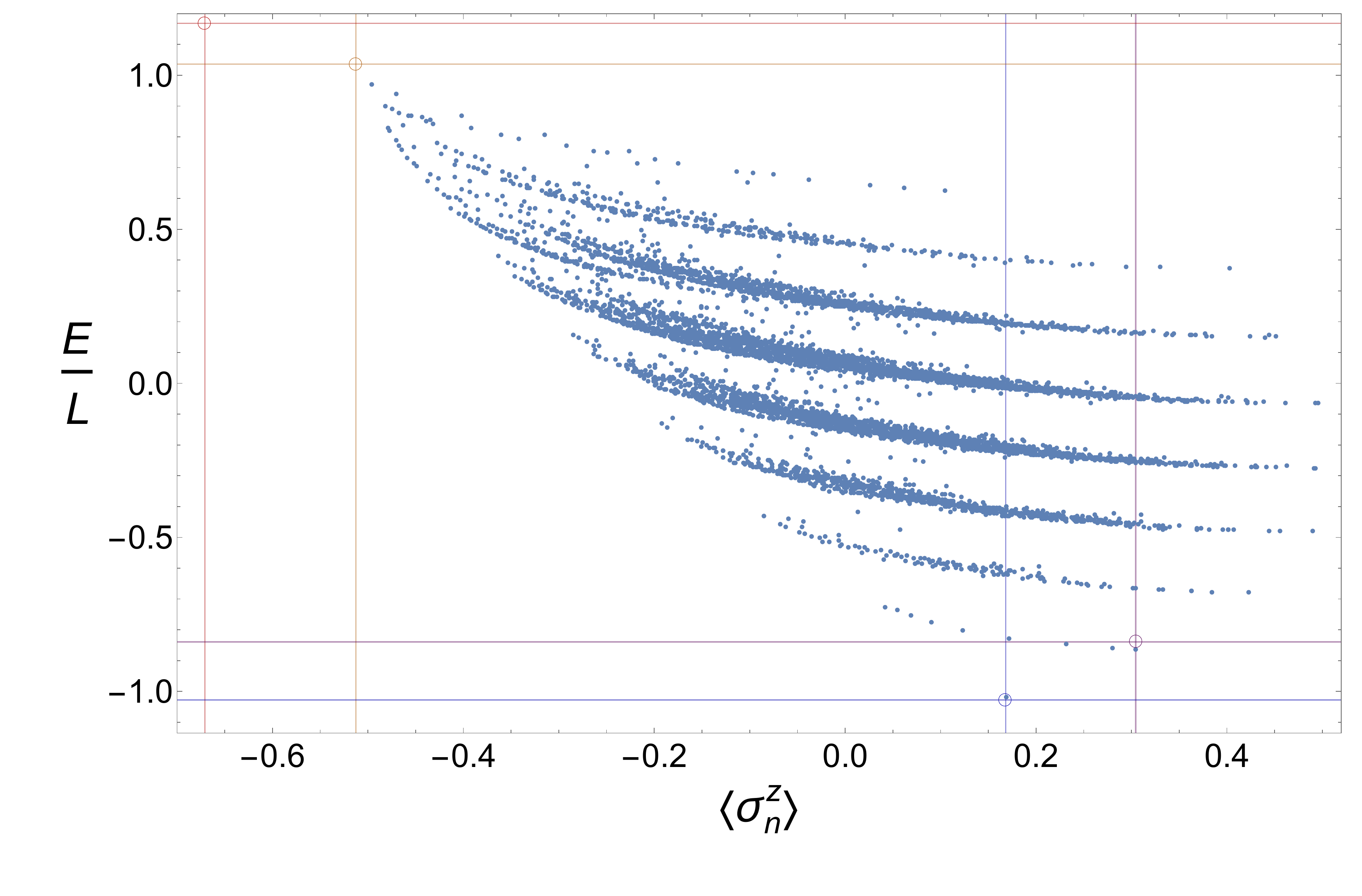}
\end{center}\caption{ The energy per unit of length as a function of the magnetisation along $z$ for the Hamiltonian \eref{eq:TFICNI} with $\tilde g=0.5$ and $\lambda=0.5$ in the reflection symmetric sector of the zero momentum subspace with spin-flip parity $\prod_\ell\sigma_\ell^z$ equal to 1 and $L=19$. The gridlines and circles indicate the values corresponding to the two lowest energy eigenvalues and the two highest ones, obtained self-consistently from the mean-field Hamiltonian with $L=19$. In the thermodynamic limit the eigenstates should be equivalent to excited states of the mean-field Hamiltonian with the associated magnetisation (\emph{cf}. Corollary \ref{C:3}), \emph{i.e.} points with the same abscissa are in correspondence with some excited states of the underlying Ising Hamiltonian.
}\label{f:IsingNI}
\end{figure}

\begin{figure}[tbp]
\begin{center}
\includegraphics[width=0.8\textwidth]{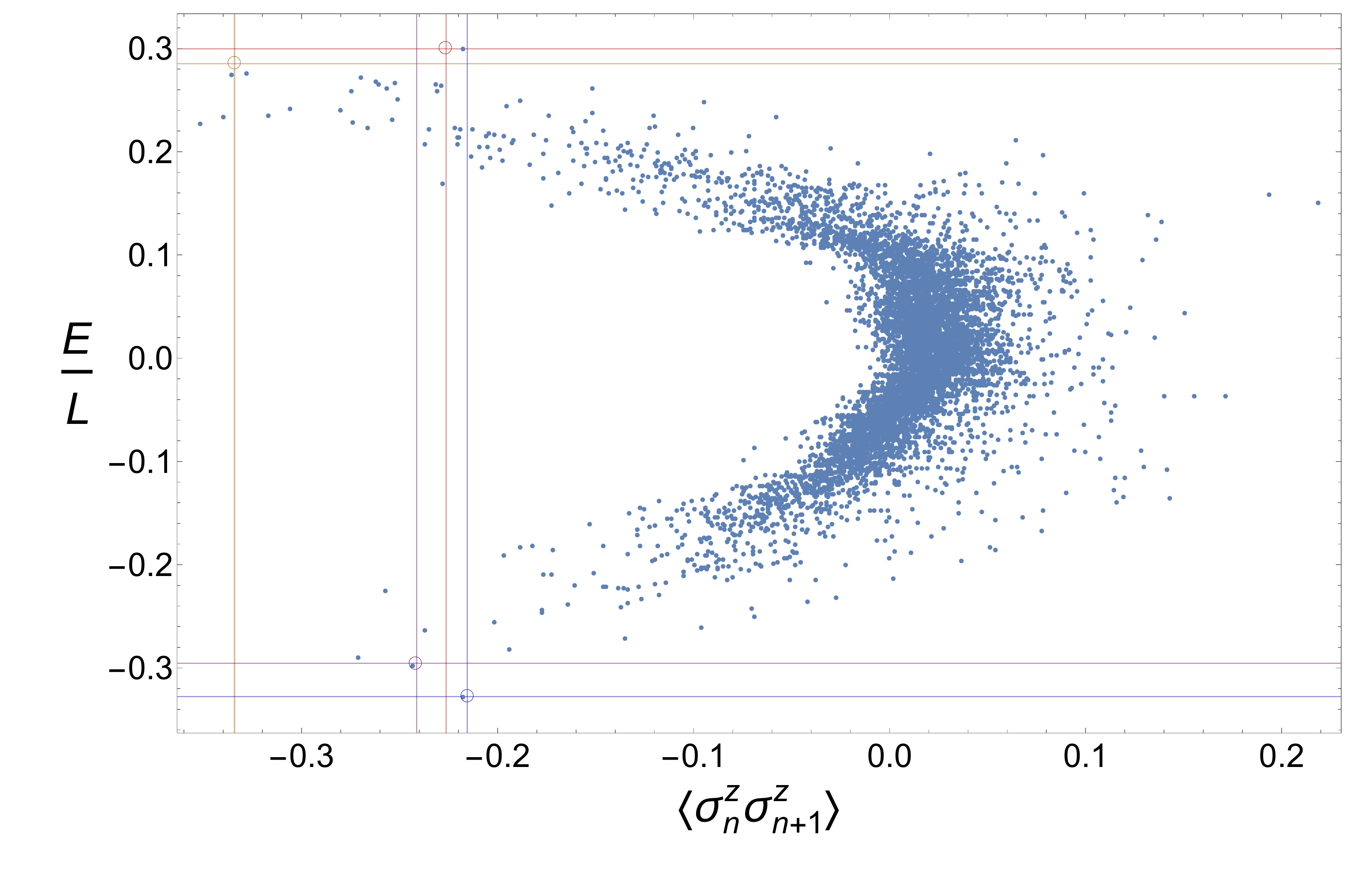}
\end{center}\caption{The energy per unit of length as a function of $\braket{\sigma_n^z\sigma_{n+1}^z}$ for the Hamiltonian \eref{eq:HWig} with $\gamma=0.25$ and $\lambda=0.5$ in the reflection symmetric sector of the zero momentum subspace with spin-flip parity $\prod_\ell\sigma_\ell^z$ equal to 1 and $L=19$. The notations are the same as in \Fref{f:IsingNI}.}\label{f:XYZ}
\end{figure}

We conclude this appendix with a comparison between finite chain data and expectations based on the mean-field correspondence, which strongly relies on the thermodynamic limit.  
In particular, we tried to apply Corollary \ref{a:C:3} to a finite chain to estimate the lowest and highest excitations. 
Figures \ref{f:IsingNI} and \ref{f:XYZ} show that the mean-field predictions are in fairly good agreement with numerical data even in small chains (notice that there can be $\mathcal O(1/L)$ corrections in both directions). 

\section*{References} %

\end{document}